\pdfoutput=1
\documentclass[a4paper,11pt]{article} 
\usepackage[utf8x]{inputenc} 
\usepackage[english]{babel}
\usepackage{amsmath,amsthm,amsfonts,amssymb,mathrsfs, bbm}
\usepackage{authblk}
\usepackage{lineno}
\usepackage[dvipsnames]{xcolor}
\usepackage{hyperref}
\usepackage{fourier}

\definecolor{BeauBlue}{rgb}{0, 0.2, .9}
\definecolor{BeauOrange}{rgb}{.8, .1, 0}
\usepackage{hyperref}
\hypersetup{
	colorlinks = true,
	linkcolor = BeauBlue,
	urlcolor = cyan,
	citecolor = BeauOrange,
}

\numberwithin{equation}{section}

\newtheorem{theorem}{Theorem}[section] % reset theorem numbering for each chapter
\newtheorem{proposition}[theorem]{Proposition} 

\newtheorem{lemma}[theorem]{Lemma}
\newtheorem{definition}[theorem]{Definition}  % definition numbers are dependent on theorem numbers

\newtheorem{remark}[theorem]{Remark}
\newtheorem{assumption}[theorem]{Assumption}
\newtheorem{cor}[theorem]{Corollary}

\DeclareMathOperator{\diam}{diam}
\DeclareMathOperator{\Tr}{Tr} 

\title{Adiabatic Evolution of Low-Temperature Many-Body Systems}
\author[1]{Rafael L. Greenblatt}
\author[2]{Markus Lange}
\author[3]{Giovanna Marcelli}
\author[4]{Marcello Porta}
\affil[1]{Mathematics Department, University of Rome ``Tor Vergata'', viale della Ricerca Scientifica 1, 00133 Roma, Italy}
\affil[2]{Institute for AI-Safety and Security, German Aerospace Center (DLR), Sankt Augustin \& Ulm, Germany}
\affil[3]{Department of Mathematical Sciences, Aalborg University, Skjernvej 4A, 9220 Aalborg, Denmark}
\affil[4]{Mathematics Area, SISSA, Via Bonomea 265, 34136 Trieste, Italy}

\date{\today}

\begin{document}

\maketitle

\begin{abstract}
We consider finite-range, many-body fermionic lattice models and we study the evolution of their thermal equilibrium state after introducing a weak and slowly varying time-dependent perturbation. Under suitable assumptions on the external driving, we derive a representation for the average of the evolution of local observables via a convergent expansion in the perturbation, for small enough temperatures. Convergence holds for a range of parameters that is uniform in the size of the system. Under a spectral gap assumption on the unperturbed Hamiltonian, convergence is also uniform in temperature. As an application, our expansion allows us to prove closeness of the time-evolved state to the instantaneous Gibbs state of the perturbed system, in the sense of expectation of local observables, at zero and at small temperatures. As a corollary, we also establish the validity of linear response. Our strategy is based on a rigorous version of the Wick rotation, which allows us to represent the Duhamel expansion for the real-time dynamics in terms of Euclidean correlation functions, for which precise decay estimates are proved using fermionic cluster expansion.
\end{abstract}

\tableofcontents

\section{Introduction}\label{sec:intro}

\paragraph{Adiabatic quantum dynamics.} The adiabatic theorem is a fundamental result in quantum mechanics, dating back to the work of Born and Fock \cite{BoFo} and Kato \cite{Ka}. Let us review its basic statement. Let $H(s)$ be a family of time-dependent Hamiltonians, depending smoothly on the time parameter $s$ for $s \in [-1,0]$. We shall suppose that $H(s)$ has a unique ground state $\varphi_{s}$, and that the energy of the ground state is separated from the rest of the spectrum by a positive gap. We are interested in the adiabatic regime, defined as follows. Let $\eta > 0$ and consider the time-dependent Schr\"odinger equation:
\begin{equation}\label{eq:tS}
i\partial_{t} \psi(t) = H(\eta t) \psi(t)\;,\qquad t\in [-1/\eta, 0].
\end{equation}
Suppose that at the initial time the system is prepared in the ground state of the Hamiltonian, $\psi(-1/\eta) = \varphi_{-1}$. We are interested in the evolution of such initial datum under (\ref{eq:tS}), in the adiabatic limit $\eta\to 0^{+}$. The adiabatic theorem states that:
\begin{equation}\label{eq:adia0}
\Big\| \psi(t) - \langle \psi(t), \varphi_{\eta t} \rangle \varphi_{\eta t} \Big\| \leq C\eta\qquad \text{for all $t\in [-1/\eta, 0]$.}
\end{equation}
This implies that, at all times $t$ and for $\eta$ small enough, the solution of the time-dependent Schr\"odinger equation (\ref{eq:tS}) is approximated by the instantaneous ground state, possibly up to a phase. This important result has been applied to study a wide class of physical systems, see \cite{Teu} for a monograph on the topic. It has been generalized to a class of contracting evolutions that includes the Schr\"odinger equation (\ref{eq:tS}) as a special case, see \cite{AFGG} and references therein.

The constant $C$ in Eq.~(\ref{eq:adia0}) depends on the details of the model, in particular on the regularity of $H(s)$. The way in which the regularity of $H(s)$ is quantified is typically via an estimate for $\| \dot{H}(s) \|$. This quantity is badly behaved in situations in which the Hamiltonian describes a many-body system, say an interacting Fermi gas on a lattice $\Lambda_{L} = [0,L]^{d} \cap \mathbb{Z}^{d}$, due to the fact that the norm of the Hamiltonian and of its derivatives grow linearly with the size of the system. Thus, the standard adiabatic theorem fails in describing the evolution of many-body quantum systems for $\eta$ small uniformly in $L$.

This is not a technical point. In fact, it turns out that the notion of convergence in Eq.~(\ref{eq:adia0}) is not the natural one for many-body systems: one cannot expect norm convergence for extensive many-body quantum systems uniformly in their size, see for instance the discussion in \cite{BdRFrev}. Instead, a much more natural notion of convergence involves the expectation value of {\it local} observables, which only probe a finite region in space. In this setting, a many-body adiabatic theorem for quantum spin systems has been recently proved in \cite{BdRF}. The result has then been extended to Fermi gases in \cite{MT}. Specifically, let $\mathcal{H}(s)$ be a time-dependent Hamiltonian for a quantum spin systems, or for lattice fermions, on a large but finite lattice $\Lambda_{L} \subset \mathbb{Z}^{d}$. Suppose that $\mathcal{H}(s)$ has a spectral gap for all times $s \in [-1, 0]$, and let $\Pi_{L}(s)$ be the projector associated with the ground state of $\mathcal{H}(s)$ on $\Lambda_{L}$. Let $P_{L}(t)$ be the solution of the evolution equation:
\begin{equation}\label{eq:PLt}
i\partial_{t} P_{L}(t) = [ \mathcal{H}(\eta t), P_{L}(t)],\qquad P_{L}(-1/\eta) = \Pi_{L}(-1).
\end{equation}
Consider the expectation value of a local operator on the time-dependent state, $\Tr \mathcal{O}_{X} P_{L}(t)$. Then, under reasonable regularity and locality assumptions on the Hamiltonian, the many-body adiabatic theorem states that \cite{BdRF, MT}:
\begin{equation}\label{eq:adiaT0}
\Big| \Tr \mathcal{O}_{X} P_{L}(t) - \Tr \mathcal{O}_{X} \Pi_{L}(\eta t)\Big| \leq C\eta\qquad \text{for all $t\in [-1/\eta, 0]$,}
\end{equation}
where the constant $C$ depends on the observable $\mathcal{O}_{X}$, but it is independent of $L$. An important application of this result is the proof of validity of linear response for extended, many-body quantum systems. To introduce the notion of linear response, let us further assume that the many-body Hamiltonian has the form:
\begin{equation}\label{eq:HLintro}
\mathcal{H}(\eta t) = \mathcal{H} + \varepsilon g(\eta t) \mathcal{P}
\end{equation}
where $\mathcal{H}$ and $\mathcal{P}$ are given by sums of local operators, and $g(t)$ is a switch function, that is a bounded function that decays fast enough at negative times. A standard choice is the exponential switch function, $g(t) = e^{t}$. Consider the dynamics generated by (\ref{eq:HLintro}) for $t\in (-\infty, 0]$. Let $P_{L}(t)$ be the solution of (\ref{eq:PLt}) with initial datum $P_{L}(-\infty) = \Pi_{L}(-\infty)$. Then, \cite{BdRF, MT} proved that, see also the reviews \cite{BdRFrev, HTrev}:
\begin{equation}\label{eq:lin0}
\lim_{\varepsilon \to 0} \lim_{\eta \to 0^{+}} \lim_{L\to \infty} \frac{1}{\varepsilon} \Big[ \Tr \mathcal{O}_{X} P_{L}(0) - \Tr \mathcal{O}_{X} P_{L}(-\infty) \Big] = \chi_{\mathcal{O},\mathcal{P}}
\end{equation}
where $\chi_{\mathcal{O},\mathcal{P}}$ agrees with the well-known Kubo formula for linear response. The statement (\ref{eq:lin0}) holds provided the thermodynamic limit of $P_{L}(-\infty)$ exists. In general, a similar result holds for $\varepsilon, \eta, L$ fixed, where $\varepsilon, \eta$ are small uniformly in $L$ and where $\varepsilon$ is small uniformly in $\eta$. The Kubo formula is equivalent to the first order term in the Duhamel expansion for the non-autonomous evolution:
\begin{equation}
\chi_{\mathcal{O},\mathcal{P}} = -i\lim_{\eta \to 0^{+}} \lim_{L\to \infty}\int_{-\infty}^{0} dt\, g(\eta t) \Tr\, \big[ \mathcal{O}_{X}, e^{i\mathcal{H}t}\mathcal{P}e^{-i\mathcal{H}t} \big] P_{L}(-\infty).
\end{equation}
These results have interesting applications in condensed matter physics. In particular, combined with \cite{HM, GMPhall, BBdRF}, they allow to prove the quantization of the Hall conductivity for gapped interacting fermions starting from the fundamental many-body Schr\"odinger equation. Among other extensions that have been obtained in the last few years we mention: the construction of non-equilibrium almost-stationary states and the application to the proof of validity of linear response for a class of perturbations that might close the spectral gap \cite{Teu2}; the proof of exactness of linear response for the quantum Hall effect \cite{BdRFL}; the extension of the many-body adiabatic theorem to infinite systems with a bulk gap \cite{HenTeu}.

Despite all this progress, an important limitation of the existing approaches is that they do not allow to study many-body quantum systems at positive temperature. In particular, the zero temperature limit is taken before the thermodynamic limit. It is of obvious physical relevance to consider the situation in which the thermodynamic limit is taken at fixed positive temperature, to make contact with experimental settings in which the temperature is possibly small but necessarily non-zero. In what follows, we will focus on interacting lattice fermionic models, which we shall describe in the grand-canonical Fock space formalism. We are interested in the following evolution equation:
\begin{equation}
i\partial_{t} \rho(t) = [ \mathcal{H}(\eta t), \rho(t) ],\qquad \rho(-\infty) = \rho_{\beta, \mu, L},
\end{equation}
with $\rho_{\beta, \mu, L} = e^{-\beta (\mathcal{H} - \mu \mathcal{N})} / \mathcal{Z}_{\beta, \mu, L}$ the grand-canonical equilibrium Gibbs state of the Hamiltonian $\mathcal{H}$ at temperature $T = 1/\beta$ and chemical potential $\mu$. A natural question is to understand under which conditions the many-body evolution of the equilibrium state can be approximated by an instantaneous Gibbs state, in the sense of the expectation of local observables. For instance, one would like to understand under which conditions
\begin{equation}\label{eq:adia00}
\Tr\mathcal{O}_{X} \rho(t) = \frac{\Tr e^{-\beta (\mathcal{H}(\eta t) - \mu \mathcal{N})} \mathcal{O}_{X}}{ \Tr e^{-\beta (\mathcal{H}(\eta t) - \mu \mathcal{N}) }} + o(1)
\end{equation} 
with $o(1)$ a quantity that vanishes as $\eta \to 0^{+}$, uniformly in $L$ (and possibly with a different temperature $T$ than the one used to define the initial datum).

\paragraph{Our result.} In this work, we introduce a different approach to study many-body quantum dynamics in the adiabatic regime, which applies to weakly interacting many-body systems at positive temperature. We consider finite-range, time-dependent Hamiltonians of the form (\ref{eq:HLintro}), under suitable assumptions discussed below. In our main result, Theorem \ref{thm:main}, we derive a representation of $\Tr\mathcal{O}_{X} \rho(t)$ via a convergent expansion in $\varepsilon$, uniformly in $\eta$ and in $L$, for small temperatures. ``Small'' means that the temperature parametrizing the initial Gibbs state is such that
\begin{equation}\label{eq:smallT}
T \ll |\varepsilon|^{-1} \eta^{d+2},
\end{equation}
uniformly in the size of the system. Under suitable assumptions on the decay of correlations of $\mathcal{H}$, the range of allowed $\varepsilon$ for which convergence holds is also uniform in $\beta$. These assumptions hold for example for finite-range Hamiltonians of the form
\begin{equation}\label{eq:weakly}
\mathcal{H}^{\lambda} = \mathcal{H}^{0} + \lambda \mathcal{V}
\end{equation}
with $\mathcal{H}^{0}$ the second-quantization of a gapped Hamiltonian, $\mathcal{V}$ a bounded local many-body interaction and $|\lambda|$ small. This is the type of models considered {\it e.g.} in \cite{GMPhall}, where the universality of the Hall conductivity for weakly interacting Fermi systems has been proved. This class of weakly interacting systems can be analyzed via fermionic cluster expansion techniques, which make it possible to prove essentially optimal estimates for the decay of the Euclidean correlation functions. For these models, the assumptions on the equilibrium Euclidean correlations required by Theorem \ref{thm:main} actually hold at positive temperature even without a gap condition on $\mathcal{H}^{0}$, however in this case one is forced to consider a range of $\lambda, \varepsilon$ that shrinks as $T \to 0$ (but still uniformly in $L$). 

Our method then allows us to prove the validity of an adiabatic theorem for local observables, in the form of Eq.~(\ref{eq:adia00}), for small temperatures in the sense of (\ref{eq:smallT}). In particular, the zero-temperature many-body adiabatic theorem (\ref{eq:adiaT0}) is recovered by taking the limit $\beta \to \infty$ at finite $L$.\footnote{For finite $L$ and as $\beta \to \infty$, the average over the Gibbs state of $\mathcal{H}(\eta t)$ converges to the average over the ground state projector, parametrized by the chemical potential $\mu$. This is a straightforward consequence of the fact that for finite $L$ the Hilbert space is finite-dimensional and so the spectrum of $\mathcal{H}(\eta t)$ is always discrete. This behavior can be extended to the case in which the limit $L\to \infty$ is taken before the limit $\beta \to \infty$, as long as such a limit of the Gibbs state exists \cite[Proposition~5.3.23]{BR2}. Our main result applies to this setting as well, provided the assumptions hold uniformly in (low) temperature. The existence of such a limit of the Gibbs state (and its other properties, for example whether it is a pure state) is largely independent from the subject matter of the present article, so we will not discuss it in detail. For weakly interacting gapped systems, the existence of the $\beta, L \to \infty$ limit can often be shown using the same techniques we use in Appendix~\ref{app:mb} to bound the correlation functions.} Furthermore, the method can also be used to prove the validity of linear response, and more generally to compute all higher-order response coefficients in terms of equilibrium correlations, see Corollary~\ref{cor:lin}. 

The proof is based on a rigorous Wick rotation, which makes it possible to rewrite the Duhamel expansion for the quantum evolution of the system in terms of time-ordered, Euclidean (or imaginary-time) connected correlation functions. Previously, this idea has been used to rigorously study the linear and quadratic response in a number of interacting gapped or gapless systems \cite{GMPhall, AMP, GMPweyl, MPmulti}. Here, we extend this strategy at all orders in the Duhamel expansion for the time-evolution of the state, and we use it to prove convergence of the Duhamel series for the real-time dynamics.

The method applies to a class of switch functions $g(\eta t)$ that can be approximated, for $\beta$ large, by functions $g_{\beta, \eta}(t)$, decaying rapidly for $t \to -\infty$, such that $g_{\beta,\eta}(t) = g_{\beta,\eta}(t-i\beta)$. This periodicity plays a key role in the proof of the Wick rotation. This requirement of course restricts the class of switch functions that we are able to consider; however, let us anticipate that this assumption holds for the standard exponential switch function, and more generally for the Laplace transform of suitable integrable functions.

Our method is completely different from that used in previous works on adiabatic theorems \cite{BdRF, MT}, and it allows to access small positive temperatures. With respect to the existing results, however, we assume that the time-dependent perturbation is slowly-varying and weak, since our method is ultimately based on a convergent expansion in $\varepsilon$, whereas in the previous works \cite{BdRF, MT} it is only assumed that the time-dependent Hamiltonian is slowly varying. The work \cite{BdRF, MT} further assume that the ground state of the time-dependent Hamiltonian $\mathcal{H}(\eta t)$ is separated by the rest of the spectrum by a uniform spectral gap, for all times. While we do not make this assumption, for the aforementioned example (\ref{eq:weakly}) it can also be proved for small $|\lambda|$ and small $|\varepsilon|$ \cite{DS}.

Besides the result itself, we believe that a relevant contribution of the present work is to import methods developed for interacting fermionic models at equilibrium to the study of real-time quantum dynamics.  In perspective, if combined with rigorous renormalization group techniques (see \cite{BG, Sa, Mabook} for reviews) we think that the approach of this paper could be extended to study the evolution of the Gibbs state of metallic or semimetallic systems, where the Fermi energy of the initial datum is not separated from the spectrum of the Hamiltonian uniformly in the size of the system. There, one does not expect an adiabatic theorem to hold; however, one might still have a convergent series expansion for the expectation of local observables in terms of Euclidean correlations, in a physically relevant range of parameters. This would be useful to establish the validity of linear response for gapless systems, widely used in applications.

Specifically, the combination of cluster expansion with rigorous renormalization group recently allowed to study the low temperature properties of a wide class of interacting gapless systems, and in particular to access their transport coefficients defined in the framework of linear response. Among the recent works, we mention the construction of the ground state of the two-dimensional Hubbard model on the honeycomb lattice \cite{GM} and the proof of universality of the longitudinal conductivity of graphene \cite{GMPgra}; the construction of the topological phase diagram of the Haldane-Hubbard model \cite{GJMP, GMPhald}; the proof of the non-renormalization for the chiral anomaly of Weyl semimetals \cite{GMPweyl}; the proof of Luttinger liquid behavior for interacting edge modes of two-dimensional topological insulators and the proof of universality of edge conductance \cite{AMP, MPspin, MPmulti}. It would be very interesting to prove the validity of linear response in the setting considered in these works, starting from many-body quantum dynamics.

Furthermore, it would be interesting to extend the methods presented in this work in the direction of studying spin transport, and prove the validity of linear response for spin-noncommuting many-body Hamiltonians. For non-interacting models, recent progress has been obtained in \cite{MPTeu, MPTau}.

The adiabatic evolution of positive temperature quantum systems has been studied in the last years, {\it e.g.} in \cite{ASF1, ASF2, JP}. The setting considered in these works is however different from the one of the present paper. The authors of \cite{ASF1, ASF2, JP} consider a small system coupled to reservoirs, and study the dynamics of the small system when the coupling with the reservoirs is introduced slowly in time. The key technical tool introduced in \cite{ASF1, ASF2, JP} is an isothermal adiabatic theorem, that proves norm-convergence of the evolved equilibrium state to the instantaneous equilibrium state of the perturbed system, in the adiabatic limit. The result holds under a suitable ergodicity assumption, which, as far as we know, has not been proved for the class of extended, interacting Fermi systems considered here. Finally, we mention the recent works \cite{JPT1, JPT2} showing that the validity of a many-body adiabatic theorem for quantum spin systems in the thermodynamic limit at fixed positive temperature and as $\eta \to 0^{+}$ is incompatible with the general notion of approach to equilibrium. We plan to further investigate the connection of our work with \cite{JPT1, JPT2} in the future.

\paragraph{Ideas of the proof.} Let us give a few more details about the method introduced in this paper. The proof starts by approximating the real-time dynamics generated by $\mathcal{H}(\eta t)$ by a suitable auxiliary dynamics, obtained from $\mathcal{H}(\eta t)$ by replacing the switch function $g(\eta t)$ by a function $g_{\beta,\eta}(t)$ such that $\lim_{\beta \to \infty} g_{\beta,\eta}(t) = g(\eta t)$ and $g_{\beta,\eta}(t) = g_{\beta,\eta}(t-i\beta)$. This approximation of course introduces an error, whose influence on the expectation values of local observables is estimated via Lieb-Robinson bounds for non-autonomous quantum dynamics \cite{BMNS, BP}. This error is responsible for the main limitation in the range of temperatures that we are able to consider. 
The advantage of replacing $g(\eta t)$ with $g_{\beta, \eta}(t)$ is that it lets us write the Duhamel series in $\varepsilon$ for the auxiliary evolution {\it exactly} in terms of Euclidean correlations, implementing a Wick rotation. 
This is made possible mainly because the periodicity of $g_{\beta,\eta}$ implies that the Kubo-Martin-Schwinger (KMS) identity remains true for the thermal expectation of modified observables of the form $g_{\beta, \eta}(t) e^{i\mathcal{H}t}\mathcal{O}_{X}e^{-i\mathcal{H}t}$. Once the Duhamel series is represented in terms of Euclidean correlations, the convergence of the series follows from the good decay properties of Euclidean correlations.  For weakly perturbed gapped models we use cluster expansion techniques to verify these assumptions. Finally, the connection with the instantaneous Gibbs state of $\mathcal{H}(\eta t)$ is obtained by noticing that, for $\eta$ small, the Wick-rotated Duhamel series agrees with the equilibrium perturbation theory in $\varepsilon$ for the Gibbs state of the Hamiltonian $\mathcal{H} + \varepsilon g(\eta t) \mathcal{P}$.

An important ingredient of our proof is the complex deformation argument of Propositions \ref{prop:wick1}, \ref{prop:wick2}, which allows us to prove the Wick rotation at all orders in the Duhamel series. Propositions \ref{prop:wick1}, \ref{prop:wick2} are the adaptation of Propositions 5.4.12, 5.4.13 of \cite{BR2} to our adiabatic setting. The main difference with respect to \cite{BR2} is that in our case the observables involved in the correlations are ``damped'' in time by $g_{\beta,\eta}(t)$: this allows to rule out the presence of spurious boundary terms at infinity in the complex deformation argument. In \cite{BR2}, these boundary terms are controlled by a suitable clustering assumption on the real-time correlation functions of the equilibrium state, which are very hard to prove for interacting models in the infinite volume limit. We are not aware of any result in this direction for the class of many-body lattice models considered in the present work.

\paragraph{Structure of the paper.} The paper is organized as follows. In Section \ref{sec:mathframework} we introduce the class of models considered in this work, we introduce their Gibbs state and Euclidean correlation functions, and we define the quantum dynamics. In Section \ref{sec:main} we state our main result, Theorem \ref{thm:main}, which provides a representation for the average of the real-time evolution of local observables via a convergent expansion in $\varepsilon$. As an application, this representation establishes a many-body adiabatic theorem for the evolution of thermal states at low temperature. A relevant consequence of the proof of our main result is the validity of linear response, Corollary \ref{cor:lin}. The proof of the main result will be given in Section \ref{sec:proof}. In Appendix \ref{app:switch} we further discuss the class of switch functions considered in the present work; in Appendix \ref{app:int} we discuss some well-known properties about time-ordered Euclidean correlations, which we include for completeness; and in Appendix \ref{app:mb} we review the verification of our Assumption \ref{ass:dec}, which is known to hold for many-body systems at zero temperature and with a spectral gap, or at positive temperature. This is done using fermionic cluster expansion, whose convergence is guaranteed by the Brydges-Battle-Federbush-Kennedy formula for cumulants.
\medskip

\paragraph{Acknowledgements.} We thank Martin Fraas for pointing out the discussion of the stability of KMS states in \cite{BR2}. We acknowledge financial support by the European Research Council (ERC) under the European Union’s Horizon 2020 research and innovation program ERC StG MaMBoQ, n.\ 802901. The work of R.L.G.\ was also partially supported by  the MIUR Excellence Department Project MatMod@TOV awarded to the Department of Mathematics, University of Rome Tor Vergata. The work of R.L.G., G.M.\ and of M.P.\ has been carried out under the auspices of the GNFM of INdAM. G.M. acknowledges support from the Independent Research Fund Denmark--Natural Sciences, grant DFF–10.46540/2032-00005B. We thank the anonymous Referees for their comments, which lead to a number of improvements with respect to our initial manuscript.

\section{The model}
\label{sec:mathframework} 
In this section we define the class of models we shall consider in this paper. We will focus on lattice fermionic systems, with finite-range interactions. We will then define the time-evolution of such systems, after introducing a time-dependent perturbation. 
\begin{remark}
Unless otherwise specified, the constants $C, K$ etc. appearing in the bounds do not depend on $\beta, L, \eta, \varepsilon$ and on time. Their values might change from line to line. Also, it will be understood that the natural numbers $\mathbb{N}$ include zero.
\end{remark}
\subsection{Lattice fermions}
Let $\Gamma$ be a $d$-dimensional lattice, namely
\[
\Gamma=\text{Span}_{\mathbb{Z}}\{a_1,\dots, a_d\} \cong \mathbb{Z}^{d},
\]
where $a_1,\dots, a_d$ are $d$ linearly independent vectors in $\mathbb{R}^{d}$. Let $L\in \mathbb{N}$, $L>0$. We define the lattice dilated by $L$ as $L\Gamma:= \text{Span}_{L\mathbb{Z}}\{a_1,\dots, a_d\} \cong L \mathbb{Z}^{d}$. The finite torus of side $L$ is defined as $\Gamma_L:=\Gamma/(L\Gamma)$, that is:
\[
\Gamma_L\cong \left\{\sum_{i=1}^d n_i a_i\, \Big|\, n_i\in\mathbb{Z},\; 0\le n_i<L  	\right\}
\]
with periodic boundary conditions. The Euclidean distance on the torus $\Gamma_L$ is given by 
\[
\|x-y\|_{L}:=\min_{v\in(L\Gamma)}\|x -  y+v\|,\qquad\forall\,x,y\in \Gamma_L.
\]
We shall denote by $M\in\mathbb{N}$, $M>0$ the total number of internal degrees of freedom of a particle. This might take into account the spin degrees of freedom, or sublattice labels. Setting $S_M:=\{1,\dots,M\}$, we define:
\[
\Lambda_L:=\Gamma_L\times S_M.
\]
We equip $\Lambda_L$ with the following distance, tracing only the space coordinates. For any ${\bf x}=(x,\sigma), {\bf y}=(y,\sigma')\in \Lambda_L$, we define:
\begin{equation}
\| {\bf x} - {\bf y} \|_{L}:= \| x- y \|_{L}.
\end{equation}
We shall describe fermionic particles on $\Lambda_{L}$, in a grand-canonical setting. To this end, we introduce the fermionic Fock space, as follows. Let the one-particle Hilbert space be $\mathfrak{h}_L:=\ell^2(\Lambda_L)$. The corresponding $N$-particle Hilbert space is its $N$-fold anti-symmetric tensor product $\mathfrak{H}_{L,N}:= \mathfrak{h}_{L}^{\wedge N}$; notice that the antisymmetric tensor product is trivial whenever $N > M L^{d}$. The fermionic Fock space is defined as usual:
\[
\mathcal{F}_L:=\bigoplus_{N=0}^{M L^d}\mathfrak{H}_{L,N},\qquad  \text{where $\mathfrak{h}_{L,0}:=\mathbb{C}$}.
\]
For finite $L$, the fermionic Fock space is a finite-dimensional vector space. Thus, any linear operator on $\mathcal{F}_L$ into itself is automatically bounded, and can be represented as a matrix. For any ${\bf x} \in \Lambda_{L}$, let $a_{{\bf x}}$ and $a^*_{{\bf x}}$ be the standard fermionic annihilation and creation operators, satisfying the canonical anti-commutation relations:
\[
\{a_{{\bf x}},a^*_{{\bf y}}\}=\delta_{{\bf x}, {\bf y}}\mathbbm{1}\quad \text{ and }\quad \{  a_{{\bf x}},a_{{\bf y}} \}=0=\{a^*_{{\bf x}},a^*_{{\bf y}}\}.
\]
For any subset $X\subseteq \Lambda_L$, we denote by $\mathcal{A}_X$ the algebra of polynomials over $\mathbb{C}$ generated by the fermionic operators restricted to $X$, $\{a_{{\bf x}}, a^*_{{\bf x}}\,:\, {\bf x} \in X \}$. An example of operator in $\mathcal{A}_{\Lambda_{L}}$ is the number operator, defined as:
\[
\mathcal{N} := \sum_{\mathbf{x}\in \Lambda_L}  a_{\mathbf{x}}^* a_{\mathbf{x}}.
\]
The operator $\mathcal{N}$ counts how many particles are present in a given sector of the Fock space: given $\psi \in \mathcal{F}$, it acts as
\[
\mathcal{N} \psi = (0 \psi^{(0)}, 1 \psi^{(1)} , \ldots, n \psi^{(n)},\ldots).
\]
We shall denote by $\mathcal{A}_X^\mathcal{N}$ the subset of $\mathcal{A}_X$ consisting of operators commuting with $\mathcal{N}$, also called gauge-invariant operators. Equivalently, these operators consist of polynomials in the creations and annihilation operators where the number of creation operators equals the number of annihilation operators. 

It is clear that any self-adjoint operator $\mathcal{O} \in \mathcal{A}_{\Lambda_{L}}$ can be represented as
\begin{equation}\label{eq:OOX}
\mathcal{O} = \sum_{X \subseteq \Lambda_{L}} \mathcal{O}_{X},
\end{equation}
where $\mathcal{O}_{X} \in \mathcal{A}_{X}$ and $\mathcal{O}_{X} = \mathcal{O}_{X}^{*}$. As $L$ varies, the operator $\mathcal{O}$ actually denotes a sequence of operators. In particular, the operators $\mathcal{O}_{X}$ in (\ref{eq:OOX}) might depend on $L$. With a slight abuse of notation, we will not display explicitly such dependence. Notice that if $X\cap Y = \emptyset$, and if $\mathcal{O}_{X}$ and $\mathcal{O}_{Y}$ are even in the number of fermionic creation and annihilation operators,
\begin{equation}
[\mathcal{O}_{X}, \mathcal{O}_{Y}] = 0\;.
\end{equation}
Finally, let us define the notion of finite-range operators. Given $X\subseteq \Lambda_{L}$, the diameter of $X$ is defined as:
\[
\diam(X):=\max_{{\bf x}, {\bf y} \in X} \|{\bf x} - {\bf y} \|_{L}.
\]
\begin{definition}[Finite-range operators]\label{def:loc} We say that $\mathcal{O} \in \mathcal{A}_{\Lambda_{L}}$ is a {\it finite-range operator} if the following holds true. There exists $R>0$ independent of $L$ such that $\mathcal{O}_{X} = 0$ whenever $\text{diam}(X) > R$. Furthermore, there exists a constant $S>0$ independent of $L$ such that, for all $X\subseteq \Lambda_{L}$:
\[
\| \mathcal{O}_{X} \| \leq S.
\]
\end{definition}
Examples of finite range operators introduced below are the Hamiltonian $\mathcal{H}$ and the perturbation $\mathcal{P}$.

\subsection{Dynamics}
\paragraph{Hamiltonian and Gibbs state.} The Hamiltonian $\mathcal{H}$ is a self-adjoint, finite-range operator in $\mathcal{A}_{\Lambda_{L}}^\mathcal{N}$. The Heisenberg time-evolution of an observable $\mathcal{O} \in \mathcal{A}_{\Lambda_{L}}$ generated by $\mathcal{H}$ is, for $t\in \mathbb{R}$:
\begin{equation}
\tau_{t}(\mathcal{O}) := e^{i\mathcal{H} t} \mathcal{O} e^{-i\mathcal{H} t}.
\end{equation}
Later, we will also consider the Heisenberg evolution for complex times $t$, whose definition poses no problem due to the finite-dimensionality of the Hilbert space.

An example of a Hamiltonian which will play an important role in this work is
\begin{equation}
\sum_{{\bf x},{\bf y} \in \Lambda_{L}} a^{*}_{{\bf x}} H({\bf x}; {\bf y}) a_{{\bf y}} + \sum_{{\bf x},{\bf y} \in \Lambda_{L}} a^{*}_{{\bf x}} a^{*}_{{\bf y}} v({\bf x};{\bf y}) a_{{\bf y}} a_{{\bf x}}, 
\end{equation}
with $H({\bf x}; {\bf y})$ and $v({\bf x};{\bf y})$ finite-range, that is both $H({\bf x}; {\bf y})$ and $v({\bf x};{\bf y})$ are vanishing if $\|{\bf x} - {\bf y}\|_{L} > R$. More generally, we shall say that $\mathcal{H}$ is the Hamiltonian for a weakly interacting lattice model if it has the form:
\begin{equation}\label{eq:weak}
\sum_{{\bf x},{\bf y} \in \Lambda_{L}} a^{*}_{{\bf x}} H({\bf x}; {\bf y}) a_{{\bf y}} + \lambda \mathcal{V}
\end{equation}
with $\lambda \in \mathbb{R}$, $|\lambda|$ small in a sense to be made precise, and $\mathcal{V}$ finite-range and of degree higher than two in the fermionic operators. We shall say that the non-interacting Hamiltonian $\mathcal{H}^{0} = \sum_{{\bf x},{\bf y} \in \Lambda_{L}} a^{*}_{{\bf x}} H({\bf x}; {\bf y}) a_{{\bf y}}$ is gapped if the spectrum of $H$ has a spectral gap uniformly in $L$.

Given $\beta > 0$, $\mu \in \mathbb{R}$, the grand-canonical equilibrium state $\langle \cdot \rangle_{\beta, \mu, L}$ associated with the Hamiltonian $\mathcal{H}$, also called equilibrium Gibbs state, is defined as:
\[
\langle \cdot \rangle_{\beta, \mu, L} := \Tr \cdot \rho_{\beta, \mu, L},\quad \rho_{\beta, \mu, L} :=  \frac{e^{-\beta (\mathcal{H} - \mu \mathcal{N})}}{\mathcal{Z}_{\beta, \mu, L}},\quad \mathcal{Z}_{\beta, \mu, L} :=  \Tr e^{-\beta (\mathcal{H} - \mu \mathcal{N})},
\]
where the trace is over the fermionic Fock space $\mathcal{F}_{L}$. Obviously, the Gibbs state is invariant under time evolution:
\[
\langle \mathcal{O} \rangle_{\beta, \mu, L} = \langle \tau_{t}(\mathcal{O}) \rangle_{\beta, \mu, L}\qquad \forall t\in \mathbb{C}.
\]
It will also be convenient to define the imaginary-time, or Euclidean, evolution of $\mathcal{O}$ as:
\begin{equation}
\gamma_{t}(\mathcal{O}) := e^{t(\mathcal{H} - \mu \mathcal{N})} \mathcal{O} e^{-t(\mathcal{H} - \mu \mathcal{N})}\qquad t\in \mathbb{R}.
\end{equation}
For $\mathcal{O} \in \mathcal{A}^{\mathcal{N}}_{\Lambda_{L}}$, one has
\begin{equation}\label{eq:gammat}
\gamma_{t}(\mathcal{O}) = \tau_{-it}(\mathcal{O})
\end{equation}
(the restriction to $ \mathcal{A}^{\mathcal{N}}_{\Lambda_{L}}$ is needed because $\gamma$, unlike $\tau$, includes a chemical potential term).
Notice that the imaginary-time evolution is no longer unitary, and the norm of $\gamma_{t}(\mathcal{O})$ might grow in time. Finally, the following property, also called {\it KMS identity}, holds:
\begin{equation}\label{eq:KMS}
\langle \gamma_{t_{1}}(\mathcal{O}_{1}) \gamma_{t_{2}}(\mathcal{O}_{2}) \rangle_{\beta, \mu, L} = \langle \gamma_{t_{2} + \beta}(\mathcal{O}_{2})\gamma_{t_{1}}(\mathcal{O}_{1})  \rangle_{\beta, \mu, L}
\end{equation}
for any $\mathcal{O}_{1}$ and $\mathcal{O}_{2}$ in $\mathcal{A}_{\Lambda_{L}}$. For finite $L$, which is our case, this identity simply follows from the definition of Gibbs state, and from the cyclicity of the trace. In order for (\ref{eq:KMS}) to hold, it is crucial that the generator of the Euclidean dynamics $\gamma_{t}$ includes the chemical potential term $-\mu\mathcal{N}$ in its definition. Notice that the dynamics $\gamma_{t}$ in (\ref{eq:gammat}) trivially extends to all complex times $t$; thus, the identity (\ref{eq:KMS}) actually holds replacing $t_{1}, t_{2}$ by any two complex numbers $z_{1}, z_{2}$. Eq.~(\ref{eq:KMS}) will play a fundamental role in our analysis.

\paragraph{Time ordering.} Let $a^{\sharp}_{{\bf x}, t} := \gamma_{t}(a^{\sharp}_{{\bf x}})$, where $a^{\sharp}_{{\bf x}}$ can be either $a_{{\bf x}}$ or $a^{*}_{{\bf x}}$. Let $t_{1}, \ldots, t_{n}$ in $[0,\beta)$. We define the time-ordering of $a^{\sharp_{1}}_{{\bf x}_{1}, t_{1}},\ldots, a^{\sharp_{n}}_{{\bf x}_{n}, t_{n}}$ as:
\begin{equation}\label{eq:time}
{\bf T} a^{\sharp_{1}}_{{\bf x}_{1}, t_{1}} \cdots a^{\sharp_{n}}_{{\bf x}_{n}, t_{n}} =  (-1)^{\pi} \mathbbm{1}(t_{\pi(1)} \geq \cdots \geq t_{\pi(n)}) a^{\sharp_{\pi(1)}}_{{\bf x}_{\pi(1)}, t_{\pi(1)}} \cdots a^{\sharp_{\pi(n)}}_{{\bf x}_{\pi(n)}, t_{\pi(n)}},
\end{equation}
where $\pi$ is the permutation needed in order to bring the times in a decreasing order, from the left, with sign $(-1)^{\pi}$, and $\mathbbm{1}(\text{condition})$ is equal to $1$ if the condition is true or $0$ otherwise. In case two or more times are equal, the ambiguity is solved by putting the fermionic operators into normal order. Other solutions of the ambiguity are of course possible; it is worth anticipating that in our applications this arbitrariness will play no role, since it involves a zero measure set of times. The above definition extends to operators in $\mathcal{A}_{\Lambda_{L}}$ by linearity. In particular, for $\mathcal{O}_{1}, \ldots, \mathcal{O}_{n}$ even in the number of creation and annihilation operators, we have:
\begin{equation}\label{eq:Tord1}
\begin{split}
{\bf T} \gamma_{t_{1}}(\mathcal{O}_{1}) \cdots &\gamma_{t_{n}}(\mathcal{O}_{n}) \\ &= \mathbbm{1}(t_{\pi(1)} \geq t_{\pi(2)} \geq \cdots \geq t_{\pi(n)}) \gamma_{t_{\pi(1)}}(\mathcal{O}_{\pi(1)}) \cdots \gamma_{t_{\pi(n)}}(\mathcal{O}_{\pi(n)}).
\end{split}
\end{equation}
The lack of the overall sign is due to the fact that the observables involve an even number of creation and annihilation operators.

\paragraph{Euclidean correlation functions.} Let $t_{i} \in [0,\beta)$, for $i=1,\ldots, n$. Given operators $\mathcal{O}_{1}, \ldots, \mathcal{O}_{n}$ in $\mathcal{A}_{\Lambda_{L}}$, we define the time-ordered Euclidean correlation function as:
\begin{equation}\label{eq:Tord2}
\langle {\bf T} \gamma_{t_{1}}(\mathcal{O}_{1}) \cdots \gamma_{t_{n}}(\mathcal{O}_{n}) \rangle_{\beta, \mu, L}.
\end{equation}
From the definition of fermionic time-ordering, and from the KMS identity, it is not difficult to check that:
\begin{equation}\label{eq:extension}
\begin{split}
\langle {\bf T} \gamma_{t_{1}}(\mathcal{O}_{1}) \cdots \gamma_{\beta}(\mathcal{O}_{k}) &\cdots \gamma_{t_{n}}(\mathcal{O}_{n}) \rangle_{\beta, \mu, L} \\ &= (\pm 1) \langle {\bf T} \gamma_{t_{1}}(\mathcal{O}_{1}) \cdots \gamma_{0}(\mathcal{O}_{k}) \cdots \gamma_{t_{n}}(\mathcal{O}_{n}) \rangle_{\beta, \mu, L};
\end{split}
\end{equation}
in the special case in which the operators involve an even number of creation and annihilation operators, which will be particularly relevant for our analysis, the overall sign is $+1$. The property (\ref{eq:extension}) allows to extend in a periodic (sign $+1$) or antiperiodic (sign $-1$) way the correlation functions to all times $t_{i} \in \mathbb{R}$. From now on, when discussing time-ordered correlations we shall always assume that this extension has been taken, unless otherwise specified.

Next, we define the connected time-ordered Euclidean correlation functions, or time-ordered Euclidean cumulants, as:
\begin{equation}\label{eq:Tcumul}
\begin{split}
&\langle {\bf T} \gamma_{t_{1}}(\mathcal{O}_{1}); \cdots; \gamma_{t_{n}}(\mathcal{O}_{n}) \rangle_{\beta,\mu,L} \\
&\qquad := \frac{\partial^{n}}{\partial \lambda_{1} \cdots \partial \lambda_{n}} \log \Big\{ 1 + \sum_{I \subseteq \{1, 2,\ldots, n\}} \lambda(I) \langle {\bf T} \mathcal{O}(I)  \rangle_{\beta,\mu,L} \Big\}\Big|_{\lambda_{i} = 0}
\end{split}
\end{equation}
where $I$ is a non-empty ordered subset of $\{1, 2,\ldots, n\}$, $\lambda(I) = \prod_{i\in I} \lambda_{i}$ and $\mathcal{O}(I) = \prod_{i\in I} \gamma_{t_{i}}(\mathcal{O}_{i})$. For $n=1$, this definition reduces to $\langle {\bf T} \gamma_{t_{1}}(\mathcal{O}_{1}) \rangle \equiv \langle \gamma_{t_{1}}(\mathcal{O}_{1}) \rangle = \langle \mathcal{O}_{1} \rangle$, while for $n=2$ one gets $\langle {\bf T} \gamma_{t_{1}}(\mathcal{O}_{1}); \gamma_{t_{2}}(\mathcal{O}_{2}) \rangle = \langle {\bf T} \gamma_{t_{1}}(\mathcal{O}_{1}) \gamma_{t_{2}}(\mathcal{O}_{2}) \rangle - \langle {\bf T} \gamma_{t_{1}}(\mathcal{O}_{1}) \rangle \langle {\bf T} \gamma_{t_{2}}(\mathcal{O}_{2}) \rangle$. More generally, the following relation between correlation functions and connected correlation function holds true:
\[
\langle {\bf T} \gamma_{t_{1}}(\mathcal{O}_{1}) \cdots \gamma_{t_{n}}(\mathcal{O}_{n}) \rangle_{\beta,\mu,L} = \sum_{P} \prod_{J\in P} \langle {\bf T} \gamma_{t_{j_{1}}}(\mathcal{O}_{j_{1}}); \cdots; \gamma_{t_{j_{|J|}}}(\mathcal{O}_{j_{|J|}}) \rangle_{\beta,\mu,L},
\]
where $P$ is the set of all partitions of $\{1, 2, \ldots, n\}$ into ordered subsets, and $J$ is an element of the partition $P$, $J = \{ j_{1}, \ldots, j_{|J|} \}$.

\paragraph{Driving the system out of equilibrium.} We are interested in driving the system out of its initial equilibrium, by adding a slowly varying time-dependent perturbation to the Hamiltonian $\mathcal{H}$. We define, for $t\leq 0$:
\begin{equation}\label{eq:Htimedep}
\mathcal{H}(\eta t) := \mathcal{H} + g(\eta t) \varepsilon \mathcal{P},
\end{equation}
where: $\eta > 0$, $\varepsilon \in \mathbb{R}$; $g(\cdot)$ is a smooth function vanishing at $-\infty$, whose further properties will be specified later on; and $\mathcal{P} \in \mathcal{A}_{\Lambda_{L}}^{\mathcal{N}}$ is a self-adjoint and finite-range operator.
As an example, we might consider:
\[
\mathcal{P} = \sum_{{\bf x}\in \Lambda_{L}} \mu({\bf x}) a^{*}_{{\bf x}} a_{{\bf x}},
\]
with $\mu({\bf x})$ bounded uniformly in $L$. More generally, we will not require $\mathcal{P}$ to be quadratic in the fermionic operators. 

The Hamiltonian $\mathcal{H}(\eta t)$ generates the following the Schr\"odinger-von Neumann non-autonomous evolution:
\begin{equation}\label{eq:vN}
i\partial_{t} \rho(t) = [ \mathcal{H}(\eta t), \rho(t) ],\qquad \rho(-\infty) = \rho_{\beta, \mu, L},\qquad t\leq 0.
\end{equation}
We shall denote by $\mathcal{U}(t;s)$ the unitary group generated by $\mathcal{H}(t)$:
\begin{equation}
i\partial_{t} \mathcal{U}(t;s) = \mathcal{H}(\eta t) \mathcal{U}(t;s),\qquad \mathcal{U}(s;s) = \mathbbm{1}\;.
\end{equation}
Using this unitary group, the solution of Eq.~(\ref{eq:vN}) can be written as
\begin{equation}
\rho(t) = \mathcal{U}(t;-\infty) \rho_{\beta, \mu, L} \mathcal{U}(t;-\infty)^{*}.
\end{equation}
Let $\mathcal{O}_{X} \in \mathcal{A}_{X}^{\mathcal{N}}$ be a local operator. We will be interested in studying its expectation value in the time-dependent state $\Tr \mathcal{O}_{X} \rho(t)$. In particular, we will be interested in understanding the dependence of this quantity on the external perturbation, and in establishing the validity of linear response, uniformly in the size of the system.
\section{Main result}\label{sec:main}
In what follows, we will consider Hamiltonians $\mathcal{H}(\eta t) = \mathcal{H} + \varepsilon g(\eta t) \mathcal{P}$ of the form introduced above. We shall denote by $\langle \cdot \rangle_{t}$ the instantaneous Gibbs state of $\mathcal{H}(\eta t)$,
\begin{equation}\label{eq:insta}
\langle \mathcal{O}_{X} \rangle_{t} := \frac{\Tr e^{-\beta (\mathcal{H}(\eta t) - \mu \mathcal{N})} \mathcal{O}_{X}}{\Tr e^{-\beta (\mathcal{H}(\eta t) - \mu \mathcal{N})}}.
\end{equation}
Our main result holds under the following assumptions on the Hamiltonian $\mathcal{H}$ (through its Gibbs state) and on the switch function $g(t)$.
\begin{assumption}[Integrability of time-ordered cumulants]\label{ass:dec} Let $S>0$, $R>0$. For $n\geq 1$ and for $i = 1, \ldots, n+1$, let $\mathcal{O}^{(i)}$ be finite-range operators, such that $\| \mathcal{O}^{(i)}_{X_{i}} \| \leq S$ and $\mathcal{O}^{(i)}_{X_{i}} = 0$ for $\text{diam}(X_{i}) > R$, uniformly in $L$. For all $\beta > 0$, there exists a constant $\mathfrak{c} \equiv \mathfrak{c}(\beta, S, R) > 0$ such that the following holds, for all $L \in \mathbb{N}$ and for all $n\in \mathbb{N}$:
\begin{equation}\label{eq:assA}
\int_{[0,\beta]^{n}} d\underline{t}\, (1+|\underline{t}|_{\beta}) \sum_{X_{i} \subseteq \Lambda_{L}}  \big| \big\langle {\bf T} \gamma_{t_{1}}(\mathcal{O}^{(1)}_{X_1});  \cdots; \gamma_{t_{n}}(\mathcal{O}^{(n)}_{X_n}); \mathcal{O}^{(n+1)}_{X} \big\rangle_{\beta, \mu, L} \big| \leq \mathfrak{c}^{n} n!
\end{equation}
where:
\begin{equation}\label{eq:normbeta}
|\underline{t}|_{\beta} := \sum_{i=1}^{n} \min_{m \in \mathbb{Z}} |t_{i} - m\beta|\;.
\end{equation}
\end{assumption}
\begin{remark} 
\begin{itemize}
\item[(i)] For weakly interacting fermionic lattice models, recall Eq.~(\ref{eq:weak}), Assumption \ref{ass:dec} can be proved via cluster expansion techniques for $|\lambda|$ small enough: the bound (\ref{eq:assA}) holds true for all finite $\beta$ and $L$, with a constant $\mathfrak{c}$ that might grow with $\beta$ but it is independent of $L$. Moreover, if the non-interacting Hamiltonian in Eq.~(\ref{eq:weak}) is gapped, and if the chemical potential $\mu$ is chosen in the spectral gap, the bound (\ref{eq:assA}) holds for $|\lambda|$ small uniformly in $\beta$, with a constant $\mathfrak{c}$ that is independent of $\beta$. We shall review these facts in Appendix \ref{app:mb}. There, we shall focus on the case of local, quartic interactions; however, the method could also be applied to cover a larger class of local interactions. The same methods can actually be used to prove the stability of the spectral gap for many-body Hamiltonians \cite{DS}.

\item[(ii)] It is known that the existence of a spectral gap for the many-body Hamiltonian implies the spatial exponential decay of correlations \cite{HK}. In general, it would be interesting to understand whether the bound (\ref{eq:assA}) can be established under the assumption of locality and of a spectral gap for the many-body Hamiltonian, with the same $n$-dependence as in the right-hand side of (\ref{eq:assA}). We are not aware of any result in this direction, at zero or at positive temperature.
 \end{itemize}
\end{remark}
The next assumption specifies the class of switch functions $g(t)$ that we are able to consider.
\begin{assumption}[Properties of the switch function]\label{ass:switch} We assume that $g(t)$ has the form, for all $t\leq 0$:
\begin{equation}\label{eq:switch} 
g(t) = \int_{0}^{\infty} d\xi\, e^{\xi t} h(\xi)\qquad \text{with $h(\xi)\in L^{1}(\mathbb{R}_{+})$,}
\end{equation}
and where $h$ is a function such that
\begin{equation}\label{eq:hbd}
\int_0^{1}  d \xi \,\frac{|h(\xi)|}{\xi^{d+2}}< \infty,\qquad \int_1^\infty  d \xi \, \xi |h(\xi)| < \infty.
\end{equation}
Alternatively, the function $h(\xi)$ can be replaced by a finite linear combination of Dirac delta distributions supported on $\mathbb{R}_{+}$.
\end{assumption}
\begin{remark}
Thus, $g$ is the Laplace transform of the function $h$. As discussed in Appendix \ref{app:switch}, the properties (\ref{eq:hbd}) are implied by suitable decay properties of the function $g(z)$ for complex times. Our setting allows us to include the function $g(t) = e^{t}$, a widely used switch function in applications, by choosing $h(\xi) = \delta(\xi - 1)$.  
\end{remark}
Next, we introduce a suitable approximation of the switch function, which will play an important role in our analysis. 
\begin{definition}[Approximation of the switch function]\label{def:apprsw} Let $\eta > 0$ and suppose that $g(t)$ satisfies Assumption \ref{ass:switch}. We define:
\begin{equation}\label{eq:gbeta}
\begin{split}
g_{\beta,\eta}(t) &:= \sum_{m = 0}^\infty \int_{\frac{2\pi}{\beta \eta} m}^{\frac{2\pi}{\beta \eta}(m+1)} d\xi\, h(\xi)   e^{\frac{2\pi}{\beta \eta}(m+1) \eta t} \\
&\equiv \sum_{\omega \in \frac{2\pi}{\beta} \mathbb{N}} \tilde g_{\beta,\eta}(\omega) e^{\omega t}, 
\end{split}
\end{equation}
\end{definition}
where $\tilde g_{\beta,\eta}(0) := 0$ and for $\omega \geq \frac{2\pi}{\beta}$:
\begin{equation}
\tilde g_{\beta,\eta}(\omega) := \int_{\frac{\omega}{\eta} - \frac{2\pi}{\beta \eta}}^{\frac{\omega}{\eta}} d\xi\, h(\xi).
\end{equation}
\begin{remark}\label{rem:switch}
\begin{itemize}
\item[(i)] The approximation of the switch function satisfies the following key identity:
\begin{equation}
g_{\beta,\eta}(t) = g_{\beta,\eta}(t - i\beta).
\end{equation}
\item[(ii)] The following estimate holds:
%_{1}
\begin{equation}\label{eq:tildegbd}
\sum_{\omega \in \frac{2\pi}{\beta} \mathbb{N}} | \tilde g_{\beta,\eta}(\omega) | \leq \sum_{m = 0}^\infty \int_{\frac{2\pi}{\beta \eta} m}^{\frac{2\pi}{\beta \eta}(m+1)} d\xi\, |h(\xi)| = \| h \|_{1},
\end{equation}
where $\|h\|_{1} \equiv \|h\|_{L^{1}(\mathbb{R}_{+})}$.
\item[(iii)] Using that, for $\xi_{1} > \xi_{2}$ and $t\leq 0$,
\begin{equation}
\Big|e^{\xi_1\eta t}-e^{\xi_2\eta t}\Big| = \eta |t| \Big|\int_{\xi_2}^{\xi_1}d \xi\, e^{\xi\eta t} \Big|\leq \eta |t|(\xi_1-\xi_2)e^{\xi_2\eta t},
\end{equation}
we have:
\begin{equation}\label{eq:diffswitch}
\begin{split}
| g_{\beta,\eta}(t) - g(\eta t) | &\leq \sum_{m = 0}^\infty \int_{\frac{2\pi}{\beta \eta} m}^{\frac{2\pi}{\beta \eta}(m+1)} d\xi\, |h(\xi)|   \Big| e^{\xi \eta t} -  e^{\frac{2\pi}{\beta \eta}(m+1) \eta t}\Big| \\
&\leq \sum_{m = 0}^\infty \int_{\frac{2\pi}{\beta \eta} m}^{\frac{2\pi}{\beta \eta}(m+1)} d \xi\, |h(\xi)| \eta |t|\frac{2\pi}{\beta \eta}e^{\xi \eta t}\\
&= \frac{2\pi|t|}{\beta} \int_{0}^{\infty} d \xi\, |h(\xi)| e^{\xi \eta t}.
\end{split}
\end{equation}
Therefore,
\begin{equation}\label{eq:ggdiff}
\begin{split}
| g_{\beta,\eta}(t) - g(\eta t) | &\leq \frac{2\pi}{\beta \eta} \int_{0}^{\infty}d \xi\, \frac{|h(\xi)|}{\xi} \xi \eta |t| e^{\xi \eta t} \\
&\leq \frac{2\pi}{e \beta \eta} \Big\| \frac{h}{\xi}  \Big\|_{1},
\end{split}
\end{equation}
and the right-hand side is finite, thanks to Assumption \ref{ass:switch}.
\end{itemize}
\end{remark}
We are now ready to state our main result.
\begin{theorem}[Main result]\label{thm:main} Let $\rho(t)$ be the solution of Eq.~(\ref{eq:vN}), with time-dependent Hamiltonian (\ref{eq:Htimedep}).  Suppose that for some $S>0$, $R>0$ independent of $L$, the Hamiltonian $\mathcal{H}$ and the perturbation $\mathcal{P}$ satisfy 
\begin{equation}
\|  \mathcal{H}_{X}\| \leq S,\; \| \mathcal{P}_{X} \| \leq S\;,\qquad \mathcal{H}_{X} =0,\; \mathcal{P}_{X} = 0\quad \text{for $\text{diam}(X) > R$}
\end{equation}
for all $X\subseteq \Lambda_{L}$. Suppose that the Gibbs state $\langle \cdot \rangle_{\beta, \mu, L}$ of $\mathcal{H}$ satisfies Assumption \ref{ass:dec} with $\mathfrak{c} \equiv \mathfrak{c}(\beta, S, R)$, and that $g(t)$ satisfies Assumption \ref{ass:switch}. Let $\mathcal{O}_{X} \in \mathcal{A}_{X}$ with $\text{diam}(X) \leq R$ and $\| \mathcal{O}_{X} \| \leq S$.  Then there exists $\varepsilon_{0} \equiv \varepsilon_{0}(\mathfrak{c}, h)$ such that for $|\varepsilon| < \varepsilon_{0}$ the following holds:
\begin{equation}\label{eq:mainexp}
\begin{split}
\Tr \mathcal{O}_{X}\rho(t) &= \langle \mathcal{O}_{X} \rangle_{\beta,\mu,L} + \sum_{n\geq 1} \frac{(-\varepsilon)^{n}}{n!} I^{(n)}_{\beta,\mu,L}(\eta,t) + R_{\beta, \mu, L}(\varepsilon,\eta,t)
\end{split}
\end{equation}
where the functions $I^{(n)}_{\beta,\mu,L}(\eta,t)$ are given by
\begin{equation}\label{eq:In}
\begin{split}
&I^{(n)}_{\beta,\mu,L}(\eta,t) \\
&\quad = \int_{[0,\beta)^{n}} d\underline{s}\, \Big[ \prod_{j=1}^{n} g_{\beta,\eta}(t-is_{j}) \Big] \langle {\bf T} \gamma_{s_{1}}(\mathcal{P}); \gamma_{s_{2}}(\mathcal{P}); \cdots ; \gamma_{s_{n}}(\mathcal{P}); \mathcal{O}_{X}  \rangle_{\beta,\mu,L}\\
\end{split}
\end{equation}
and satisfy the estimate
\begin{equation}\label{eq:Iest}
|I^{(n)}_{\beta,\mu,L}(\eta,t)|\leq \| h \|_{1}^{n} \mathfrak{c}^{n} n!.
\end{equation}
The error term $R_{\beta, \mu, L}(\varepsilon,\eta,t)$ in (\ref{eq:mainexp}) is bounded as:
\begin{equation}\label{eq:Rest}
|R_{\beta, \mu, L}(\varepsilon,\eta,t)| \leq \frac{K |\varepsilon|}{\eta^{d+2} \beta},
\end{equation}
where the constant $K \equiv K(S, R, h) >0$ does not depend on $\mathfrak{c}$. Furthermore, we also have:
\begin{equation}\label{eq:adia}
\Big| \Tr \mathcal{O}_{X} \rho(t) - \langle \mathcal{O}_{X} \rangle_{t}\Big|\leq \frac{K|\varepsilon|}{\eta^{d+2}\beta} + C_{1}|\varepsilon| \Big( \eta + \frac{1}{\beta}\Big) + \frac{C_{2}|\varepsilon|}{\beta \eta},
\end{equation}
where $\langle \cdot \rangle_{t}$ is the instantaneous Gibbs state of $\mathcal{H}(\eta t)$, Eq.~(\ref{eq:insta}), and $C_{i}\equiv C_{i}(\mathfrak{c}, h)$ for $i=1,2$.
\end{theorem}
\begin{remark} 
{\begin{enumerate}
			\renewcommand{\labelenumi}{(\roman{enumi})}
\item The series in (\ref{eq:mainexp}) turns out to be equal to the Duhamel series for the quantum dynamics generated by the Hamiltonian 
\begin{equation}\label{eq:hbetaeta}
\mathcal{H}_{\beta,\eta}(t) = \mathcal{H} + \varepsilon g_{\beta,\eta}(t) \mathcal{P},
\end{equation}
after a complex deformation from real time to imaginary times (Wick rotation).  Thus, our result in particular includes the statement that the Duhamel series for the dynamics generated by (\ref{eq:hbetaeta}) is convergent in $\varepsilon$ uniformly in $L$ and in $\eta$, under the Assumption \ref{ass:dec}. This information is very useful because, as proved later in Proposition \ref{prop:comp}, the dynamics generated by $\mathcal{H}_{\beta,\eta}(t)$ is close to the dynamics generated by the original Hamiltonian $\mathcal{H}(\eta t)$, in the sense of evolution of local operators, for $\beta$ large enough.
\item If $\mathfrak{c}$ can be taken to be independent of $\beta$, then the radius of convergence in $\varepsilon$ is independent of $\beta$ as well and we can use the results to describe the $\beta \to \infty $ limit.  As commented after Assumption \ref{ass:dec}, this is the case for many-body perturbations of non-interacting gapped lattice models, with Hamiltonian of the form (\ref{eq:weak}) and for $|\lambda|$ small enough. 
\item If $\mathfrak{c}$ does not depend on $\beta$, our result allows to take the zero temperature limit $\beta \to \infty$ after the thermodynamic limit $L\to \infty$. By Eqs. (\ref{eq:mainexp}), (\ref{eq:In}), the existence of these limits is implied by the existence of the same limits for the equilibrium Gibbs state $\langle \cdot \rangle_{\beta,\mu,L}$. To the best of our knowledge, all previous works on many-body adiabatic dynamics considered the case in which the temperature is sent to zero before the thermodynamic limit. 
\item For finite $L$, the Hilbert space is finite dimensional and the spectrum of $\mathcal{H}(\eta t)$ is discrete. Thus,  it is straightforward to prove that as $\beta \to \infty$ and for fixed $L$ the average over the instantaneous Gibbs state in (\ref{eq:adia}) converges to the average over the ground state projector, which a priori might have a nontrivial degeneracy (we do not know whether Assumption \ref{ass:dec} has implications for the multiplicity of the ground state). This allows to recover the zero temperature many-body adiabatic theorem, for the class of systems satisfying the assumptions of Theorem \ref{thm:main}.
\item To illustrate how the adiabatic theorem (\ref{eq:adia}) is implied by (\ref{eq:mainexp}), we observe that the first two terms in the right-hand side of (\ref{eq:mainexp}) reconstruct the average over the instantaneous Gibbs state of $\mathcal{H}(\eta t)$, after replacing the functions $g_{\beta,\eta}(t - is_{j})$ in Eq. (\ref{eq:In}) with $g(\eta t)$. To see this, we use the representation of the instantaneous Gibbs state of $\mathcal{H}(\eta t)$ in terms of a convergent cumulant expansion in $\varepsilon$; see Eq. (\ref{eq:equi}) below.

Let us now discuss the origin of the various error terms in the right-hand side of (\ref{eq:adia}). The first term is due to $R_{\beta, \mu, L}$, which arises from the approximation of the real-time dynamics generated by $\mathcal{H}(\eta t)$ with the real-time dynamics generated by $\mathcal{H}_{\beta,\eta}(t)$. This error term is estimated via Lieb-Robinson bounds, and its estimate (\ref{eq:Rest}) does not use any information about the state. This bound introduces the strongest constraint on the range of temperatures that we are able to consider. The second error term in (\ref{eq:adia}) arises from the replacement of $g_{\beta,\eta}(t-is_{j})$ with $g_{\beta,\eta}(t)$; the bound for the difference introduces a factor $(\eta + (1/\beta)) |s_{j}|_{\beta}$, and the factor $|s_{j}|_{\beta}$ is controlled using the assumption on the Euclidean correlations (\ref{eq:assA}). Finally, the last error term in (\ref{eq:adia}) arises from the replacement of $g_{\beta,\eta}(t)$ with $g(\eta t)$, and it relies on the estimate (\ref{eq:ggdiff}).
\item If one restricts the attention to the switch functions $g_{\beta,\eta}(t)$ of the type (\ref{eq:gbeta}) the first and the last error terms in (\ref{eq:adia}) are absent. Thus, for this special class of switch functions it is possible to prove that:
\begin{equation}\label{eq:diffsimp}
\Big| \Tr \mathcal{O}_{X} \rho(t) - \frac{\Tr \mathcal{O}_{X} e^{-\beta (\mathcal{H}_{\beta,\eta}(t) - \mu \mathcal{N})}}{\Tr e^{-\beta (\mathcal{H}_{\beta,\eta}(t) - \mu \mathcal{N})}} \Big| \leq C |\varepsilon| \Big( \eta + \frac{1}{\beta} \Big)\;.
\end{equation}
Referring to the proof of the main result in Section \ref{sec:proof}, the $n$-th order contribution in $\varepsilon$ to the difference in the left-hand side of (\ref{eq:diffsimp}) is only due to the term $R^{(n)}_{2,1}(t)$ defined in (\ref{eq:RRR}), which is estimated in (\ref{eq:R21}). Notice that the special switch function $g_{\beta,\eta}(t)$ are superpositions of exponentials $e^{\frac{2\pi}{\beta} (m+1)t}$ for $m \in \mathbb{N}$; thus for fixed $\beta$, the dependence of $g_{\beta,\eta}$ on $\eta$ is in general not a rescaling of time. The smallest abiabatic parameter that can be reached with this type of switch functions is $2\pi / \beta$.
\item In Eq.~(\ref{eq:adia}), we compare the time-evolved state with the instantaneous Gibbs state, defined with the same temperature as the initial datum: in the small temperature regime we are considering, we cannot resolve the heating of the system due to the perturbation. A better approximation should be obtained introducing a suitable, time-dependent, renormalization of the instantaneous Gibbs state. We plan to come back to this point in the future.
\end{enumerate}}
\end{remark}
The many-body adiabatic theorem (\ref{eq:adia}) can be improved, under the additional assumption that the first $m$ derivatives of the switch function vanish at zero.
\begin{cor}[Improved adiabatic convergence]\label{rem:improv} Under the same assumptions of Theorem \ref{thm:main}, the following is true. Suppose that $\partial_{t}^{j} g(0) = 0$ for all $1\leq j \leq m$. Furthermore, suppose that 
\begin{equation}\label{eq:assAm}
\int_{[0,\beta]^{n}} d\underline{t}\, (1+|\underline{t}|^{m+1}_{\beta}) \sum_{X_{i} \subseteq \Lambda_{L}}  \big| \big\langle {\bf T} \gamma_{t_{1}}(\mathcal{O}^{(1)}_{X_1});  \cdots; \gamma_{t_{n}}(\mathcal{O}^{(n)}_{X_n}); \mathcal{O}^{(n+1)}_{X} \big\rangle_{\beta, \mu, L} \big| \leq D_{m+1} \mathfrak{c}^{n} n!
\end{equation}
and that 
\begin{equation}
\int_1^\infty d \xi\, \xi^{m+1}\left|h(\xi)\right|  < \infty.\label{eq:hbd_enhanced}
\end{equation}
Then, the following improved many-body adiabatic theorem holds:
\begin{equation}\label{eq:etam}
	\Big| \Tr \mathcal{O}_{X} \rho(0) - \langle \mathcal{O}_{X} \rangle_{0}\Big|\leq \frac{K|\varepsilon|}{\eta^{d+2}\beta} + C_{1,m+1}|\varepsilon| \Big( \eta^{m+1}+\frac1\beta \Big) + \frac{C_{2}|\varepsilon|}{\beta \eta}.
\end{equation}
\end{cor}
\begin{remark}
{\begin{enumerate}\renewcommand{\labelenumi}{(\roman{enumi})}
\item These switch functions are allowed by our setting; for example, we might consider $g(t) = 1-(1-e^t)^m$.
\item As $\beta \to \infty$, a similar result has been first obtained by \cite{BdRF}. We observe that, with respect to \cite{BdRF}, here we show that the improved convergence (Corollary~\ref{rem:improv}) holds under the assumption that the first $m$ derivatives at zero of the switch function are vanishing, while in \cite{BdRF} it is assumed that the first $m+d+1$ derivatives vanish (with $d$ the spatial dimension of the system).
\item The assumption (\ref{eq:assAm}) holds true for many-body perturbations of gapped lattice models, and it can be proved via the analysis of Appendix \ref{app:mb}.
\end{enumerate}}	
\end{remark}
Combined with a few straightforward estimates (Eqs. \eqref{eq:corlin_start} to \eqref{eq:corlin_end}), we also have the following result.
\begin{cor}[Validity of linear response]\label{cor:lin} Under the same assumptions as Theorem \ref{thm:main}, 
\begin{equation}\label{eq:lin}
\begin{split}
&\Tr \mathcal{O}_{X} \rho(t) - \langle \mathcal{O}_{X} \rangle_{\beta, \mu, L} \\&\qquad = -\varepsilon  \int_{0}^{\beta} ds\, g_{\beta,\eta}(t - is) \langle \gamma_{s}(\mathcal{P})\;; \mathcal{O}_{X} \rangle_{\beta, \mu, L} + R_{\beta, \mu, L}(\varepsilon,\eta,t)
\end{split}
\end{equation}
where the error term $R_{\beta, \mu, L}(\varepsilon,\eta,t)$ is bounded as:
\begin{equation}\label{eq:linerr}
|R_{\beta, \mu, L}(\varepsilon,\eta,t)| \leq \frac{K |\varepsilon|}{\eta^{d+2} \beta} + C |\varepsilon|^{2}
\end{equation}
with $K$ is as in (\ref{eq:Rest}) and $C$ depends on $\mathfrak{c}$. In Eq.~(\ref{eq:lin}), the function $g_{\beta,\eta}(t - is)$ can be replaced by $g(\eta t)$, up to replacing the error term $R_{\beta, \mu, L}$ by $\widetilde{R}_{\beta, \mu, L}$, such that:
\begin{equation}\label{eq:linerr2}
|\widetilde{R}_{\beta, \mu, L}(\varepsilon,\eta,t)| \leq C|\varepsilon| \Big( \eta + \frac{1}{\beta \eta} \Big) + \frac{K |\varepsilon|}{\eta^{d+2} \beta} + C|\varepsilon|^{2}.
\end{equation}
Furthermore, the main term in Eq.~(\ref{eq:lin}) is equal to the first order term in the Duhamel expansion, up to small errors:
\begin{equation}\label{eq:last}
\begin{split}
\Big| \int_{0}^{\beta} ds\, g_{\beta,\eta}(t - is) \langle \gamma_{s}(\mathcal{P})\;; \mathcal{O}_{X} \rangle_{\beta, \mu, L} &- i\int_{-\infty}^{t} d s\, g(\eta s) \langle \left[ \tau_{t}(\mathcal{O}_{X}), \tau_{s}(\mathcal{P}) \right] \rangle_{\beta, \mu, L}\Big| \\
&\leq \frac{K}{ \eta^{d+2}\beta}.
\end{split}
\end{equation}
\end{cor}
\begin{remark} Eq. (\ref{eq:linerr2}) shows that, up to an error term vanishing as $\beta \to \infty$ and $\eta \to 0^{+}$, the first order in $\varepsilon$ in the Duhamel expansion for the real-time dynamics is equal to the first order in $\varepsilon$ in the expansion for the instantaneous Gibbs state $\langle \cdot \rangle_{t}$. To see this, we rely on the cumulant expansion in $\varepsilon$ for the instantaneous Gibbs state, Eq. (\ref{eq:equi}). More generally, the argument can be extended to show that the $n$-th order term in $\varepsilon$ in the real-time Duhamel expansion for the dynamics generated by $\mathcal{H}(\eta t)$ is equal to the $n$-th order term in $\varepsilon$ in the expansion of the instantaneous Gibbs state of $\mathcal{H}(\eta t)$, up to vanishing errors as $\beta\to \infty$ and as $\eta \to 0^{+}$.
\end{remark}
The proof of the main result will be given in Section \ref{sec:proof}, and it is organized as follows. In Section \ref{sec:duha} we recall how to derive the Duhamel expansion for the many-body evolution, in a finite volume. In Section \ref{sec:LR} we introduce the auxiliary dynamics, obtained after replacing $g(\eta t)$ with $g_{\beta,\eta}(t)$ in $\mathcal{H}(\eta t)$, and we prove the closeness of the two dynamics for $\beta$ large enough, in the sense of expectation of local observables using Lieb-Robinson bounds. In Section \ref{sec:wick} we represent the Duhamel expansion for the auxiliary dynamics in a finite volume  via the Wick rotation: this allows to get an identity for every term in the Duhamel expansion in terms of time-ordered Euclidean correlations. We then use Assumption \ref{ass:dec} to establish convergence of the (Wick-rotated) Duhamel series, uniformly in the size of the system. In Section \ref{sec:compa} we recall the cumulant expansion in $\varepsilon$ for the instantaneous Gibbs state of $\mathcal{H}(\eta t)$. Finally, in Section \ref{sec:fin} we put everything together, and we prove Theorem \ref{thm:main}.

\section{Proof of Theorem \ref{thm:main}}\label{sec:proof}
\subsection{Duhamel expansion}\label{sec:duha}
We start by recalling how to derive the well-known Duhamel series for the expectation of local observables. Given a time-dependent Hamiltonian $\mathcal{H}(t) = \mathcal{H} + \varepsilon g(\eta t) \mathcal{P}$, let us consider the associated unitary evolution:
\begin{equation}
\begin{split}
i\partial_{t} \mathcal{U}(t;s) &= \mathcal{H}(t) \mathcal{U}(t;s) \\
\mathcal{U}(s;s) &= \mathbbm{1}\;.
\end{split}
\end{equation}
For $\varepsilon = 0$ one trivially has $\mathcal{U}(t;s)  = e^{-i\mathcal{H}(t-s)}$. We are interested in deriving a perturbative expansion around the evolution generated by $\mathcal{H}$. To this end, we define the unitary evolution in the interaction picture as:
\begin{equation}
\mathcal{U}_{\text{I}}(t;s) := e^{i\mathcal{H}t} \mathcal{U}(t;s) e^{-i\mathcal{H}s}\;.
\end{equation}
Clearly, $\mathcal{U}_{\text{I}}(s;s) = \mathbbm{1}$, and:
\begin{equation}\label{eq:genI}
\begin{split}
i\partial_{t} \mathcal{U}_{\text{I}}(t;s) &= e^{i\mathcal{H}t} (-\mathcal{H} + \mathcal{H}(t)) \mathcal{U}(t;s) e^{-i\mathcal{H}s} \\
&\equiv \varepsilon g(\eta t) \tau_{t}(\mathcal{P}) \mathcal{U}_{\text{I}}(t;s)\;.
\end{split}
\end{equation}
Next, we write, for $T>0$ and for $0\geq t\geq -T$:
\begin{equation}
\begin{split}
\Tr \mathcal{O} \mathcal{U}(t;-T) &\rho_{\beta, \mu, L} \mathcal{U}(t;-T)^{*} - \Tr \mathcal{O} \rho_{\beta, \mu, L} \\
&= \Tr \tau_{t}(\mathcal{O}) \mathcal{U}_{\text{I}}(t;-T) \rho_{\beta, \mu, L} \mathcal{U}_{\text{I}}(t;-T)^{*} - \Tr \tau_{t}(\mathcal{O}) \rho_{\beta, \mu, L}
\end{split}
\end{equation}
where we used the cyclicity of the trace and the invariance of $\rho_{\beta, \mu, L}$ under the dynamics generated by $\mathcal{H}$. Finally, by Eq.~(\ref{eq:genI}):
\begin{equation}
\begin{split}
\Tr \mathcal{O} &\mathcal{U}(t;-T) \rho_{\beta, \mu, L} \mathcal{U}(t;-T)^{*} - \Tr \mathcal{O} \rho_{\beta, \mu, L}\\
& = (-i \varepsilon) \int_{-T}^{t} ds\, g(\eta s) \Tr \tau_{t}(\mathcal{O}) [ \tau_{s}(\mathcal{P}), \mathcal{U}_{\text{I}}(s;-T) \rho_{\beta, \mu, L} \mathcal{U}_{\text{I}}(s;-T)^{*}] \\
& =  (-i \varepsilon) \int_{-T}^{t} ds\, g(\eta s) \Tr [ \tau_{t}(\mathcal{O}),  \tau_{s}(\mathcal{P})] \mathcal{U}_{\text{I}}(s;-T) \rho_{\beta, \mu, L} \mathcal{U}_{\text{I}}(s;-T)^{*}\;.
\end{split}
\end{equation}
The procedure can be iterated. One gets:
\begin{equation}\label{eq:duh}
\begin{split}
&\Tr \mathcal{O} \mathcal{U}(t;-T) \rho_{\beta, \mu, L} \mathcal{U}(t;-T)^{*} = \Tr \mathcal{O} \rho_{\beta, \mu, L}\\
& \quad + \sum_{n=1}^{m} (-i\varepsilon)^{n} \int_{-T \leq s_{n} \leq \ldots \leq s_{1} \leq t} d \underline{s}\, g(\eta s_{1}) \cdots g(\eta s_{n}) \\& \qquad \cdot \langle [ \cdots [[ \tau_{t}(\mathcal{O}), \tau_{s_{1}}(\mathcal{P})], \tau_{s_{2}}(\mathcal{P})] \cdots \tau_{s_{n}}(\mathcal{P}) ] \rangle_{\beta, \mu, L} \\ &\quad + R_{\beta, \mu, L}^{(m+1)}(-T;t),
\end{split}
\end{equation}
where $R_{\beta, \mu, L}^{(m+1)}(-T;t)$ is the Taylor remainder of the expansion, given by:
\begin{equation}
\begin{split}
R&_{\beta, \mu, L}^{(m+1)}(-T;t) = (-i\varepsilon)^{m+1} \int_{-T \leq s_{m+1} \leq \ldots \leq s_{1} \leq t} d \underline{s}\, g(\eta s_{1}) \cdots g(\eta s_{m+1}) \\& \cdot \Tr [ \cdots [ \tau_{t}(\mathcal{O}), \tau_{s_{1}}(\mathcal{P})], \cdots \tau_{s_{m+1}}(\mathcal{P}) ] \mathcal{U}_{\text{I}}(s_{m+1};-T) \rho_{\beta, \mu, L} \mathcal{U}_{\text{I}}(s_{m+1};-T)^{*}\;.
\end{split}
\end{equation}
On a finite lattice and for $\eta > 0$, the series is absolutely convergent. In fact, by using the boundedness of the fermionic operators, and the unitarity of the time evolution, we have the following crude estimate:
\begin{equation}
\begin{split}
\big|R_{\beta, \mu, L}^{(m+1)}(-T;t)\big| \leq &|\varepsilon|^{m+1} \int_{-T \leq s_{m+1} \leq \ldots \leq s_{1} \leq t} d \underline{s}\, |g(\eta s_{1})| \cdots |g(\eta s_{m+1})| \\
& \cdot 2^{m+1} \| \mathcal{O} \| \| \mathcal{P} \|^{m+1} \\
\leq & \| \mathcal{O} \| \frac{C^{m+1} |\varepsilon|^{m+1} |\Lambda_{L}|^{m+1} \eta^{-m-1}}{(m+1)!} \Big[\int_{-\infty}^{0} ds\, |g(s)|\Big]^{m+1}\;,
\end{split}
\end{equation}
for a universal constant $C>0$. Thus, taking $m$ large enough, uniformly in $T$, the error term can be made as small as wished. Hence, we have, in the $T\to \infty$ limit:
\begin{equation}\label{eq:duh2}
\begin{split}
&\Tr \mathcal{O} \mathcal{U}(t;-\infty) \rho_{\beta, \mu, L} \mathcal{U}(t;-\infty)^{*} =  \Tr \mathcal{O} \rho_{\beta, \mu, L} \\
&\quad + \sum_{n=1}^{\infty} (-i\varepsilon)^{n} \int_{-\infty \leq s_{n} \leq \ldots \leq s_{1} \leq t} d \underline{s}\, g(\eta s_{1}) \cdots g(\eta s_{n}) \\&\qquad \cdot \langle [ \cdots [[ \tau_{t}(\mathcal{O}), \tau_{s_{1}}(\mathcal{P})], \tau_{s_{2}}(\mathcal{P})] \cdots \tau_{s_{n}}(\mathcal{P}) ] \rangle_{\beta, \mu, L}\;.
\end{split}
\end{equation}
Eq.~(\ref{eq:duh2}) is the Duhamel expansion for the average of $\mathcal{O}$ on the time-dependent state $\rho(t) := \mathcal{U}(t;-\infty) \rho_{\beta, \mu, L} \mathcal{U}(t;-\infty)^{*}$. In order to extract useful information from this representation, we need estimates for the various terms that are uniform in the size of the system. In particular, we would like to prove that the series converges uniformly in $\varepsilon$, as $L\to \infty$ and for $\eta$ small.

The main difficulty to achieve this is the control of the time-integral, uniformly in $\eta$. For fixed $\eta$, this problem might be approached using Lieb-Robinson bounds, see \cite{NSY} for a review. This bound reads, for two operators $\mathcal{O}_{X}$ and $\mathcal{O}_{Y}$ supported on $X, Y$:
\begin{equation}
\| [ \mathcal{O}_{X}, \tau_{t}(\mathcal{O}_{Y}) ] \| \leq C e^{v|t| - c\cdot \text{dist}(X,Y)},
\end{equation}
for suitable positive constants $C, c, v$. Combined with the boundedness of the fermionic operators, this estimate (and its extension to multi-commutators, \cite{BP}) can be used to prove that the series in (\ref{eq:duh2}) is convergent uniformly in $L$, however only for $|\varepsilon| \leq \varepsilon(\eta)$ with $\varepsilon(\eta) \to 0^{+}$ as $\eta \to 0^{+}$. In the next section, we shall study the time-evolution using a different approach, which gives estimates that are uniform in $\eta$.
\subsection{The auxiliary dynamics}\label{sec:LR}
Let
\begin{equation}\label{eq:tildeH}
	\mathcal{H}_{\beta,\eta}( t) := \mathcal{H} + \varepsilon g_{\beta,\eta}(t) \mathcal{P},
\end{equation}
with $g_{\beta,\eta}(t)$ introduced in Definition \ref{def:apprsw}. Here we will prove that the evolutions generated by $\mathcal{H}(\eta t)$ and by $\mathcal{H}_{\beta,\eta}(t)$ are close, at small enough temperature, in the sense of the expectation of local observables. To compare the two evolutions, we will use Lieb-Robinson bounds for non-autonomous dynamics \cite{BMNS, BP}. 
\begin{proposition}[Comparison of dynamics]\label{prop:comp} Let $\rho(t)$, $\tilde \rho(t)$ be the time-dependent states evolving with $\mathcal{H}(\eta t)$, $\mathcal{H}_{\beta,\eta}(t)$ respectively, with initial data given by $\rho(-\infty) = \tilde \rho(-\infty) = \rho_{\beta, \mu, L}$. Let $\mathcal{O}_{X}$ be a local observable. Then, there exists $K>0$, independent of $\mathfrak{c}$ and dependent on $h$, such that for all $\varepsilon, \eta, \beta, L$:
\begin{equation}\label{eq:compdyn}
\big|\Tr \mathcal{O}_{X} (\rho(t) - \tilde{\rho}(t))\big| \leq \frac{K|\varepsilon|}{\eta^{d+2}\beta},\qquad \text{for all $t\leq 0$.}
\end{equation}
\end{proposition}
\begin{proof} We start by writing:
\begin{equation}\label{eq:zetabeta}
g(\eta t) = g_{\beta,\eta}(t) + \zeta_{\eta, \beta}(t),
\end{equation}
where $g_{\beta,\eta}(t)$ is defined in Eq.~(\ref{eq:gbeta}), and the error term satisfies the bound, by Eq.~(\ref{eq:diffswitch}):
\begin{equation}\label{eq:zeta}
|\zeta_{\eta, \beta}(t)| \leq \frac{2\pi|t|}{\beta} \int_{0}^{\infty}d \xi\, |h(\xi)|  e^{\xi \eta t}.
\end{equation}
Next, we write:
\begin{equation}
\begin{split}
&\Tr \mathcal{O}_{X} (\rho(t) - \tilde{\rho}(t)) = \\
&\lim_{T\to -\infty} \Tr \mathcal{O}_{X} (\mathcal{U}(t;-T) \rho_{\beta, \mu, L} \mathcal{U}(t;-T)^{*} - \widetilde{\mathcal{U}}(t; -T) \rho_{\beta, \mu, L} \widetilde{\mathcal{U}}(t;-T)^{*})
\end{split}
\end{equation}
where $\mathcal{U}(t;s)$, $\widetilde{\mathcal{U}}(t;s)$ are the unitary groups generated by $\mathcal{H}(\eta t)$, $\mathcal{H}_{\beta,\eta}(t)$, respectively. We estimate the argument of the limit as:
\begin{equation}
\begin{split}
&\Big|\Tr  \Big(\mathcal{U}(t;-T)^{*} \mathcal{O}_{X} \mathcal{U}(t;-T) - \widetilde{\mathcal{U}}(t;-T)^{*} \mathcal{O}_{X} \widetilde{\mathcal{U}}(t;-T)\Big) \rho_{\beta, \mu, L}\Big| \\
& \leq \Big\| \mathcal{O}_{X}  - \mathcal{U}(t;-T)\widetilde{\mathcal{U}}(t;-T)^{*} \mathcal{O}_{X} \widetilde{\mathcal{U}}(t;-T) \mathcal{U}(t;-T)^{*} \Big\|,
\end{split}
\end{equation}
where we used that $ \rho_{\beta, \mu, L} \geq 0$, $\Tr \rho_{\beta, \mu, L} = 1$ and the unitarity of time-evolution. Next, we rewrite the argument of the norm as:
\begin{equation}
\begin{split}
&\mathcal{O}_{X}  - \mathcal{U}(t;-T)\widetilde{\mathcal{U}}(t;-T)^{*} \mathcal{O}_{X} \widetilde{\mathcal{U}}(t;-T) \mathcal{U}(t;-T)^{*} \\
&\quad = -i \int_{-T}^{t} ds\, i\frac{\partial}{\partial s} \mathcal{U}(t;s)\widetilde{\mathcal{U}}(t;s)^{*} \mathcal{O}_{X} \widetilde{\mathcal{U}}(t;s) \mathcal{U}(t;s)^{*} \\
&\quad = -i \int_{-T}^{t} ds\,  \mathcal{U}(t;s)\Big[ -\mathcal{H}(\eta s) + \mathcal{H}_{\beta,\eta}(s),\, \widetilde{\mathcal{U}}(t;s)^{*} \mathcal{O}_{X} \widetilde{\mathcal{U}}(t;s)\Big] \mathcal{U}(t;s)^{*} \\
&\quad \equiv i \int_{-T}^{t} ds\, \varepsilon \zeta_{\eta, \beta}(s) \mathcal{U}(t;s)\Big[ \mathcal{P},\, \widetilde{\mathcal{U}}(t;s)^{*} \mathcal{O}_{X} \widetilde{\mathcal{U}}(t;s)\Big] \mathcal{U}(t;s)^{*},
\end{split}
\end{equation}
where in the third line we used that $\mathcal{U}(t;s)^{*} = \mathcal{U}(s;t)$, and in the last line we used Eq.~(\ref{eq:zetabeta}). Therefore,
\begin{equation}\label{eq:norm}
\begin{split}
&\Big\| \mathcal{O}_{X}  - \mathcal{U}(t;-T)\widetilde{\mathcal{U}}(t;-T)^{*} \mathcal{O}_{X} \widetilde{\mathcal{U}}(t;-T) \mathcal{U}(t;-T)^{*} \Big\| \\
&\leq \int_{-T}^{t} ds\, |\varepsilon| \big|\zeta_{\eta, \beta}(s)\big| \Big\| \Big[ \mathcal{P},\, \widetilde{\mathcal{U}}(t;s)^{*} \mathcal{O}_{X} \widetilde{\mathcal{U}}(t;s)\Big] \Big\|.
\end{split}
\end{equation}
Next, we claim that:
\begin{equation}\label{eq:LRcons}
\Big\| \Big[ \mathcal{P},\, \widetilde{\mathcal{U}}(t;s)^{*} \mathcal{O}_{X} \widetilde{\mathcal{U}}(t;s)\Big] \Big\| \leq C(|t-s|^{d} + 1).
\end{equation}
This inequality stems from the Lieb-Robinson bound for non-autonomous dynamics, see Theorem 4.6 of \cite{BMNS} for quantum spin systems, or Theorem 5.1 of \cite{BP} for the case of lattice fermions:
\begin{equation}\label{eq:LRnonauto}
\Big\| \Big[ \mathcal{O}_{Y}, \widetilde{\mathcal{U}}(t;s)^{*} \mathcal{O}_{X} \widetilde{\mathcal{U}}(t;s) \Big] \Big\| \leq C e^{v|t-s| - c \cdot \text{dist}(X,Y)}
\end{equation}
for any two local operators $\mathcal{O}_{X}$, $\mathcal{O}_{Y}$. The proof of (\ref{eq:LRcons}) is standard, and we give it here for completeness. Representing the perturbation $\mathcal{P}$ in terms of its local potentials, we have:
\begin{equation}
\begin{split}
\Big\| \Big[ \mathcal{P},\, \widetilde{\mathcal{U}}(t;s)^{*} \mathcal{O}_{X} \widetilde{\mathcal{U}}(t;s)\Big] \Big\| &\leq  \sum_{Y \subseteq \Lambda_{L}} \Big\| \Big[ \mathcal{P}_{Y},\, \widetilde{\mathcal{U}}(t;s)^{*} \mathcal{O}_{X} \widetilde{\mathcal{U}}(t;s)\Big] \Big\| \\
&= \sum_{\substack{Y \subseteq \Lambda_{L} \\ \text{dist}(X,Y) \leq D |t-s|}} \Big\| \Big[ \mathcal{P}_{Y},\, \widetilde{\mathcal{U}}(t;s)^{*} \mathcal{O}_{X} \widetilde{\mathcal{U}}(t;s)\Big] \Big\| \\\\
&\quad +  \sum_{\substack{Y \subseteq \Lambda_{L} \\ \text{dist}(X,Y) > D |t-s|}} \Big\| \Big[ \mathcal{P}_{Y},\, \widetilde{\mathcal{U}}(t;s)^{*} \mathcal{O}_{X} \widetilde{\mathcal{U}}(t;s)\Big] \Big\|
\end{split}
\end{equation}
with $D$ large enough to be chosen below. By the boundedness of the fermionic operators, and by the unitarity of the dynamics, the first term in the right-hand side is estimated as:
\begin{equation}
\sum_{\substack{Y \subseteq \Lambda_{L} \\ \text{dist}(X,Y) \leq D |t-s|}} \Big\| \Big[ \mathcal{P}_{Y},\, \widetilde{\mathcal{U}}(t;s)^{*} \mathcal{O}_{X} \widetilde{\mathcal{U}}(t;s)\Big] \Big\| \leq K(|t-s|^{d} + 1)
\end{equation}
where we used the fact that the sum is restricted to sets $Y$ of bounded diameter. For the second term, we use the Lieb-Robinson bound (\ref{eq:LRnonauto}), to get:
\begin{equation}
\begin{split}
&\sum_{\substack{Y \subseteq \Lambda_{L} \\ \text{dist}(X,Y) > D|t-s|}} \Big\| \Big[ \mathcal{P}_{Y},\, \widetilde{\mathcal{U}}(t;s)^{*} \mathcal{O}_{X} \widetilde{\mathcal{U}}(t;s)\Big] \Big\| \\&\qquad \leq \sum_{\substack{Y \subseteq \Lambda_{L} \\ \text{dist}(X,Y) > D|t-s|,\; \text{diam}(Y) \leq R}}C e^{v|t-s| - c \cdot \text{dist}(X,Y)}\;.
\end{split}
\end{equation}
Choosing $D$ large enough, we have:
\begin{equation}
\begin{split}
&\sum_{\substack{Y \subseteq \Lambda_{L} \\ \text{dist}(X,Y) > D|t-s|}} \Big\| \Big[ \mathcal{P}_{Y},\, \widetilde{\mathcal{U}}(t;s)^{*} \mathcal{O}_{X} \widetilde{\mathcal{U}}(t;s)\Big] \Big\| \\&\qquad\qquad \leq \sum_{\substack{Y \subseteq \Lambda_{L} \\ \text{dist}(X,Y) > D|t-s|,\; \text{diam}(Y) \leq R}}C e^{- (c/2) \cdot \text{dist}(X,Y)} \\
&\qquad\qquad \leq K.
\end{split}
\end{equation}
This concludes the check of (\ref{eq:LRcons}). Using the bound (\ref{eq:LRcons}) in (\ref{eq:norm}), we get:
\begin{equation}\label{eq:norm2}
\begin{split}
&\Big\| \mathcal{O}_{X}  - \mathcal{U}(t;-T)\widetilde{\mathcal{U}}(t;-T)^{*} \mathcal{O}_{X} \widetilde{\mathcal{U}}(t;-T) \mathcal{U}(t;-T)^{*} \Big\| \\
&\qquad \leq C\int_{-T}^{t} ds\, |\varepsilon| \big|\zeta_{\eta, \beta}(s)\big| (|t - s|^{d} + 1) \\
&\qquad \leq \frac{K|\varepsilon|}{\beta} \int_{0}^{\infty}d \xi\, |h(\xi)|  \int_{-T}^{t} ds\, e^{\xi \eta s} (|s|^{d+1} + 1),
\end{split}
\end{equation}
where in the last step we used the bound (\ref{eq:zeta}) and we exchanged the order of integration. Using that:
\begin{equation}
\int_{-T}^{t} ds\, e^{\xi \eta s} |s|^{d+1} \leq \frac{C}{(\eta \xi)^{d+2}},
\end{equation}
we finally obtain:
\begin{equation}\label{eq:norm3}
\begin{split}
&\Big\| \mathcal{O}_{X}  - \mathcal{U}(t;-T)\widetilde{\mathcal{U}}(t;-T)^{*} \mathcal{O}_{X} \widetilde{\mathcal{U}}(t;-T) \mathcal{U}(t;-T)^{*} \Big\| \\
&\qquad \leq \frac{K|\varepsilon|}{\eta^{d+2}\beta} \int_{0}^{\infty}d \xi\, \frac{|h(\xi)|}{\xi^{d+2}}.
\end{split}
\end{equation}
Eq.~(\ref{eq:compdyn}) follows from assumption (\ref{eq:hbd}), after a redefinition of the constant $K$. This concludes the proof.
\end{proof}
\subsection{Wick rotation}\label{sec:wick}
Here we shall represent each coefficient in the Duhamel expansion (\ref{eq:duh2}) for the auxiliary dynamics generated by (\ref{eq:tildeH}) in terms of Euclidean correlation functions, via a complex deformation argument. The advantage is that useful space-time decay estimates for Euclidean correlations can be proved using statistical mechanics tools, such as the cluster expansion. This complex deformation is known in physics as Wick rotation, and here it will be established rigorously for the auxiliary dynamics. The next lemma is the main result of the section. Its proof is based on the adaptation of ideas of Section 5.4 of \cite{BR2} to our adiabatic setting.
\begin{lemma}[Wick rotation]\label{lem:wick} Let $A \in \mathcal{A}_{\Lambda_{L}}$, $B \in \mathcal{A}^{\mathcal{N}}_{\Lambda_{L}}$. Let $n\in \mathbb{N}$. Let $a(s)$ be a periodic function with period $\beta$, such that:
\begin{equation}\label{eq:adef0}
a(s) = \sum_{\omega\in \frac{2\pi}{\beta} \mathbb{N}} \tilde a(\omega) e^{-i\omega s},\qquad \sum_{\omega\in \frac{2\pi}{\beta} \mathbb{N}} |\tilde a(\omega)| \leq C,\qquad \tilde a(0) = 0.
\end{equation}
Then, the following identity holds true, for all $t\leq 0$:
\begin{equation}
\begin{split}
&\int_{-\infty \leq s_{n} \leq \ldots \leq s_{1} \leq t}  d\underline{s}\, \Big[\prod_{j=1}^{n} a(is_{j})\Big] \langle [ \cdots [[\tau_{t}(A), \tau_{s_{1}}(B)], \tau_{s_{2}}(B)] \cdots \tau_{s_{n}}(B) ] \rangle_{\beta,\mu,L} \\
&\; = \frac{(-i)^{n}}{n!} \int_{[0,\beta)^{n}} d\underline{s}\, \Big[\prod_{j=1}^{n} a(it+s_{j})\Big] \langle {\bf T} \gamma_{s_{1}}(B); \gamma_{s_{2}}(B); \cdots; \gamma_{s_{n}}(B); A \rangle_{\beta,\mu,L}.
\end{split}
\end{equation}
\end{lemma}
\begin{remark}\label{rem:az} 
\begin{itemize}
\item[(i)] The function $s\mapsto g_{\beta,\eta}(-is)$ satisfies the properties (\ref{eq:adef0}), recall Definition \ref{def:apprsw}.
\item[(ii)] Notice that the function defined in (\ref{eq:adef0}) extends to a function $a(z)$ on the lower-half complex plane, that is analytic for $\text{Im}\, z < 0$ and continuous for $\text{Im}\, z\leq 0$.
\end{itemize}
\end{remark}
The proof will be broken in a few intermediate steps. In what follows, it will be convenient to use the following notations. We define inductively:
\begin{equation}\label{eq:Cn}
C_{0} := \tau_{t}(A)\;,\qquad C_n(t_1,\dots,t_n) := a(it_{n}) [ C_{n-1}(t_1,\dots,t_{n-1}), \tau_{t_n}(B) ].
\end{equation}
Moreover, we set:
\begin{equation}\label{eq:B_def}
B_{0} := 1\;,\qquad B_{n}(t_{1}, \ldots, t_{n}) := \Big[\prod_{i=1}^{n} a(it_{i})\Big]  \tau_{t_{1}}(B) \cdots \tau_{t_{n}}(B).
\end{equation}
Also, we shall introduce the $n$-dimensional symplex of side $\beta$ as:
\begin{equation}
\label{eqn:simplex}
\Delta^n_{\beta}:=\big\{ (s_1,\dots,s_n)\in \mathbb{R}^n\,:\, \beta>s_1>\dots>s_{n}>0  \big\}.
\end{equation}
The combination of Propositions \ref{prop:wick1}, \ref{prop:wick2} below is the adaptation of Propositions 5.4.12, 5.4.13 of \cite{BR2} to our adiabatic setting. Differently from \cite{BR2}, our results hold without clustering assumptions on the real-time correlations.
\begin{proposition}[Basic complex deformation]\label{prop:wick1} Let $B\in \mathcal{A}^{\mathcal{N}}_{\Lambda_{L}}$, $C\in \mathcal{A}_{\Lambda_{L}}$. For every $j \in \mathbb{N}$ and for all $t\leq 0$:
\begin{equation}\label{eq:j>0}
\begin{split}
&\int_{-\infty}^{t} d r\, a(ir) \int_{\Delta^j_{\beta}} d \underline{s}\,   \langle  B_j\left( r - i s_1,\dots, r - i s_j \right) \left[ \tau_{r}(B),C \right]\rangle_{\beta, \mu, L}\\ & \qquad = i \int_{\Delta^{j+1}_{\beta}} d\underline{s}\,
\langle  B_{j+1}\left( t - i s_1,\dots, t - i s_{j+1} \right) C \rangle_{\beta, \mu, L}.
\end{split}
\end{equation}
\end{proposition}
\begin{proof} To start, let us prove the $j=0$ case, which reads:
\begin{equation}\label{eq:j0}
\int_{-\infty}^{t} d r\, a(ir) \langle \left[ \tau_{r}(B),C \right] \rangle_{\beta, \mu, L} = i \int_{0}^{\beta} ds\,\langle B_{1}\left( t - i s\right) C \rangle_{\beta, \mu, L}.
\end{equation}
Let $T > 0$. By the KMS identity, Eq.~(\ref{eq:KMS}), and using that $B$ commutes with the number operator:
\begin{equation}
\begin{split}
	\int_{-T}^{t} d r\, a(ir) &\langle \left[ \tau_{r}(B),C \right] \rangle_{\beta, \mu, L} = \int_{-T}^{t} d r\, a(ir) \Big[ \langle \tau_{r}(B)C \rangle_{\beta, \mu, L} - \langle C \tau_{r}(B)\rangle_{\beta, \mu, L} \Big]\\
&= \int_{-T}^{t} d r\, a(ir)\, \Big[\langle \tau_{r}(B)C \rangle_{\beta, \mu, L} - \langle \tau_{r - i\beta}(B)C \rangle_{\beta, \mu, L}\Big].
\end{split}\end{equation}
By assumption (\ref{eq:adef0}), we use the trivial but crucial fact $a(ir) = a(ir + \beta)$ to write:
\begin{equation}
\begin{split}
&\int_{-T}^{t} d r\, a(ir) \langle \left[ \tau_{r}(B),C \right] \rangle_{\beta, \mu, L} \\
&\qquad = \int_{-T}^{t} d r\, \Big[ a(ir) \langle \tau_{r}(B)C \rangle_{\beta, \mu, L} - a(i(r - i\beta)) \langle \tau_{r - i\beta}(B)C \rangle_{\beta, \mu, L}\Big].
\end{split}
\end{equation}
Now, consider the function, for $z\in \mathbb{C}$:
\begin{equation}
f(z) = a(iz)\langle \tau_{z}(B) C\rangle_{\beta, \mu, L}.
\end{equation}
For finite $L$ and finite $\beta$, this function is analytic on $\text{Re}\, z <0$, and it is continuous on $\text{Re}\, z \leq 0$. In fact, the function $z\mapsto \langle \tau_{z}(B) C\rangle_{\beta, \mu, L}$ is entire for finite $L, \beta$, while $a(iz)$ is analytic for $\text{Re}\, z < 0$ and continuous for $\text{Re}\, z\leq 0$, recall Definition (\ref{eq:adef0}) and Remark \ref{rem:az}.

For $\varepsilon > 0$ small enough, let $\Gamma$ be the complex path for $(\text{Re}\,z , \text{Im}\,z)$:
\begin{equation}
\Gamma = (-T,0) \to (t -\varepsilon, 0) \to (t - \varepsilon, -\beta) \to (-T, -\beta) \to (-T,0)\;,
\end{equation}
where every arrow corresponds to an oriented straight line in the complex plane. By Cauchy's integral theorem,
\begin{equation}\label{eq:caugamma}
\int_{\Gamma}dz\, f(z) = 0.
\end{equation} 
We start by writing:
\begin{equation}\label{eq:cau0}
\int_{-T}^{t} d r\, a(ir) \langle \left[ \tau_{r}(B),C \right] \rangle_{\beta, \mu, L} = \lim_{\varepsilon \to 0^{+}} \int_{-T}^{t - \varepsilon} d r\, a(i r) \langle \left[ \tau_{r}(B),C \right] \rangle_{\beta, \mu, L}\;,
\end{equation}
and using Eq.~(\ref{eq:caugamma}):
\begin{equation}\label{eq:cau}
\int_{-T}^{t - \varepsilon} d r\, a(i r) \langle \left[ \tau_{r}(B),C \right] \rangle_{\beta, \mu, L} = i \int_0^\beta d s\, f(t - \varepsilon - i s) - i \int_0^\beta d s\, f(-T - i s).
\end{equation}
We claim that the last term vanishes as $T\to \infty$. In fact:
\begin{equation}\label{eq:estbdwick}
\begin{split}
|f(-T - i s)| &\leq \Big|a(s - iT) \langle \tau_{-T - is}(B)C \rangle_{\beta, \mu, L}\Big| \\
&\leq \Big(\sum_{\omega} |\tilde a(\omega)|\Big) e^{-\frac{2\pi}{\beta}T} \|\tau_{-T - is}(B)\| \| C \| \\
&\leq C e^{-\frac{2\pi}{\beta}T} \|B\| \|C\| e^{2s \|\mathcal{H}\|},
\end{split}
\end{equation}
where we used the unitarity of the real-time dynamics. Notice that all norms in (\ref{eq:estbdwick}) are finite: we are on a finite lattice with side $L$, and the fermionic Fock space for models on a finite lattice is finite-dimensional. Hence, the bound (\ref{eq:estbdwick}) shows that the $T\to \infty$ limit of the second term in the right-hand side of (\ref{eq:cau}) vanishes, for $\beta$ and $L$ finite. We thus have
\begin{equation}
\begin{split}
&\lim_{T\to \infty}\int_{-T}^{t} d r\, a(ir) \langle \left[ \tau_{r}(B),C \right] \rangle_{\beta, \mu, L} \\
&\qquad = \lim_{\varepsilon \to 0^{+}} i \int_{0}^\beta d s\, a(i(t-\varepsilon) + s)\langle \tau_{t - \varepsilon -is}(B) C  \rangle_{\beta, \mu, L} \\
&\qquad = i \int_{0}^\beta d s\, a(it + s)\langle \tau_{t-is}(B) C  \rangle_{\beta, \mu, L}
,
\end{split}
\end{equation}
which proves Eq.~(\ref{eq:j0}). Let us now discuss the $j>0$ case, Eq.~(\ref{eq:j>0}). By the KMS identity and $a(ir) = a(ir + \beta)$, we get:
\begin{equation}
\begin{split}
\int_{-T}^{t} d r\, &\int_{\Delta^{j}_{\beta}} d \underline{s}\, a(ir) 
\langle B_j\left( r - i s_1,\dots, r - i s_j \right) \left[ \tau_{r}(B),C \right] \rangle_{\beta, \mu, L}\\ = & \int_{-T}^{t} d r \int_{\Delta^{j}_{\beta}} d \underline{s}\, a(ir)   \langle  B_j\left( r - i s_1,\dots, r - i s_j \right)\tau_{r} (B) C\rangle_{\beta, \mu, L} \\ & - \int_{-T}^{t} d r \int_{\Delta^{j}_{\beta}} d \underline{s}\, a(i (r - i\beta)) \langle \tau_{r - i\beta }(B) B_j\left( r - i s_1,\dots, r - i s_j \right) C  \rangle_{\beta, \mu, L}.
\end{split}\label{eq:L1-L2_def}
\end{equation}
We further rewrite this expression as, recalling the definition of $B_{j}(\cdot)$, Eq.~(\ref{eq:B_def}):
\begin{equation}\label{eq:L1L2}
\begin{split}
(\ref{eq:L1-L2_def}) &=  \int_{\Delta^{j}_{\beta}} d \underline{s}\int_{-T}^{t} d r\, \langle B_{j+1}\left( r - i s_1,\dots, r - i s_j, r  \right) C \rangle_{\beta, \mu, L} \\ & - \int_{\Delta^{j}_{\beta}} d \underline{s} \int_{-T}^{t} d r\,
\langle  B_{j+1}\left(r-i\beta, r - i s_1,\dots, r - i s_j \right) C  \rangle_{\beta, \mu, L} \\ &=:  L^{T}_1 - L^{T}_2\,.
\end{split}
\end{equation}
Let us now introduce the change of variables, for $1\leq k < j$:
\begin{equation} s_k' = s_{k+1}- s_1 + \beta,\qquad  \quad s_j' = \beta - s_1.
\end{equation}
Notice that $\beta > s_{1} > s_{2} > \cdots > s_{j} > 0$, we also have $\beta > s'_{1} > s'_{2} > \cdots > s'_{j} > 0$, that is $(s'_{1}, \ldots, s'_{n}) \in \Delta^{j}_{\beta}$. In terms of these variables, for $2\leq k \leq j$:
\begin{equation}
\begin{split}
s_{1} &= \beta - s'_{j} \\
s_{k} &= s'_{k-1} + s_{1} - \beta \equiv s'_{k-1} - s'_{j}.
\end{split}
\end{equation}
We then rewrite the term $L_{1}^{T}$ in (\ref{eq:L1L2}) as:
\begin{equation}\label{eq:L1T0}
\begin{split}
L^{T}_1 &= \int_{\Delta^{j}_{\beta}} d\underline{s}' \int_{-T}^{t} d r\, \\& \quad \langle B_{j+1}\left(r - i (\beta - s'_j), r - i (s'_1 - s'_j), \dots, r - i (s'_{j-1} - s'_j),r \right) C \rangle_{\beta, \mu, L}.
\end{split}
\end{equation}
Let us now introduce the function:
\begin{equation}\label{eq:fdef}
f_{(\beta,s_1,\dots,s_j)}(z):= \langle B_{j+1}(z - i(\beta - s_{j}), z - i (s_1 - s_{j}),\dots, z ) C\rangle_{\beta, \mu, L}.
\end{equation}
The function $f_{(\beta,s_1,\dots,s_j)}(z)$ is analytic for $\text{Re}\, z < 0$ and continuous on $\text{Re}\, z\leq 0$. We have:
\begin{equation}\label{eq:L1}
L^{T}_{2} = \int_{\Delta^{j}_{\beta}} d\underline{s} \int_{-T}^{t} dr\, f_{(\beta,s_1,\dots,s_j)}(r-is_{j});
\end{equation}
also, relabelling the $s'$ variables in $s$ variables in Eq.~(\ref{eq:L1T0}):
\begin{equation}\label{eq:L2}
L^{T}_{1} = \int_{\Delta^{j}_{\beta}} d\underline{s} \int_{-T}^{t} dr\, f_{(\beta,s_1,\dots,s_j)}(r).
\end{equation}
As for the $j=0$ case, we will use a complex deformation argument to rewrite $L_{1}^{T} - L_{2}^{T}$ in a convenient way. To this end, let us now define the complex path for $(\text{Re}\, z, \text{Im}\, z)$, for $\varepsilon > 0$ small enough:
\begin{equation}
\Gamma = (-T,0) \to (t  -\varepsilon, 0) \to (t - \varepsilon, -s_{j}) \to (-T, -s_{j}) \to (-T,0).
\end{equation}
By continuity of $f_{(\beta,s_1,\dots,s_j)}(z)$:
\begin{equation}\label{eq:L22L1}
\begin{split}
&L^{T}_{1} - L^{T}_{2} = \int_{\Delta^{j}_{\beta}} d\underline{s} \lim_{\varepsilon \to 0^{+}}\Big[ \int_{-T}^{t - \varepsilon} dr\, f(r) -  \int_{-T}^{t-\varepsilon} dr\, f(r-is_{j})\Big] \\
&\; = i\int_{\Delta^{j}_{\beta}} d\underline{s} \lim_{\varepsilon \to 0^{+}}\Big[ \int_{0}^{s_{j}} ds_{j+1}\, f(t - \varepsilon - is_{j+1}) - \int_{0}^{s_{j}} ds_{j+1}\, f(-T - is_{j+1})\Big] \\
&\; =  i\int_{\Delta^{j}_{\beta}} d\underline{s}\, \Big[ \int_{0}^{s_{j}} ds_{j+1}\, f(t - is_{j+1}) - \int_{0}^{s_{j}} ds_{j+1}\, f(-T - is_{j+1})\Big]\;,
\end{split}
\end{equation}
where the second identity follows from Cauchy theorem and the last from the continuity of the integrand. The last term in the right-hand side of (\ref{eq:L22L1}) vanishes as $T\to \infty$. This is implied by the following estimate, recall Eq.~(\ref{eq:fdef}):
\begin{equation}\label{eq:Tj}
| f_{(\beta,s_1,\dots,s_j)}(-T - i s_{j+1}) | \leq \| \tilde{a} \|_{1}^{j+1} \| C \| \| B \|^{j+1} e^{2\beta (j+1) \|\mathcal{H}\|} e^{- \frac{2\pi}{\beta} T(j+1)}.
\end{equation}
Consider now the first term in the right-hand side of (\ref{eq:L22L1}). The integrand has the form, for a function $g$ entire in all its arguments, recall (\ref{eq:fdef}):
\begin{equation}
\begin{split}
&f_{(\beta,s_1,\dots,s_j)}(t - i s_{j+1}) \\
&\; =g(t - i(\beta - s_{j} + s_{j+1}),\ldots, t - i (s_{k-1} - s_{j} + s_{j+1}),\ldots, t - i s_{j+1})\;.
\end{split}
\end{equation}
Let us introduce the change of variables, for $2 \leq k \leq j$:
\begin{equation}
s'_1=\beta-s_j+s_{j+1},\quad s'_k=s_{k-1}-s_j+s_{j+1}, \quad s'_{j+1}=s_{j+1}.
\end{equation}
We notice that $\beta > s'_{1} > \ldots > s'_{k} > \ldots > s'_{j+1} > 0$. Thus, the second term in the r.h.s. of (\ref{eq:L22L1}) can be written as the integral over the symplex $\Delta^{\beta}_{j+1}$: 
\begin{equation}\label{eq:48}
\begin{split}
&i \int_{\Delta^{j}_{\beta}} d\underline{s} \int_{0}^{s_j} d s_{j+1}\, f_{(\beta,s_1,\dots,s_j)}(t - i s_{j+1}) \\&\quad = i \int_{\Delta^{j+1}_{\beta}} d\underline{s}'\, g(t - is'_{1},\ldots, t - i s'_{k},\ldots, t - i s'_{j+1}) \\
&\quad \equiv i\int_{\Delta^{j+1}_{\beta}} d\underline{s}'\, \langle B_{j+1}(t - is'_{1},\dots, t - i s'_{k},\ldots, t - i s'_{j+1} ) C\rangle_{\beta, \mu, L}.
\end{split}
\end{equation}
All in all, from (\ref{eq:L1L2}), (\ref{eq:L22L1}), (\ref{eq:Tj}), (\ref{eq:48}), relabelling the $s'$ variables as $s$ variables:
\begin{equation}
\begin{split}
\int_{-\infty}^{t} d r\, \int_{\Delta^{j}_{\beta}} &d \underline{s}\, a(ir)
\langle B_j\left( r - i s_1,\dots, r - i s_j \right)\left[ \tau_{r}(B),C \right] \rangle_{\beta, \mu, L} \\
& = L_{1}^{\infty} - L_{2}^{\infty} \\
& = i\int_{\Delta^{j+1}_{\beta}} d\underline{s}\, \langle B_{j+1}(t - is_{1},\dots, t - i s_{k},\ldots, t - i s_{j+1} ) C\rangle_{\beta, \mu, L}
\end{split}
\end{equation}
which concludes the proof of the proposition.
\end{proof}
Next, we use Proposition \ref{prop:wick1} to rewrite the coefficients appearing in the Duhamel expansion in terms of imaginary-time correlations.
\begin{proposition}[Multiple complex deformation]\label{prop:wick2} Under the same assumptions of Lemma \ref{lem:wick} the following identity holds:
\begin{equation}\label{eq:multicomplex}
\begin{split}
&\int_{-\infty \leq s_{n} \leq \ldots \leq s_{1} \leq t}  d\underline{s}\, \Big[ \prod_{i=1}^{n} a(is_{i}) \Big] \langle [ \cdots [[\tau_{t}(A), \tau_{s_{1}}(B)], \tau_{s_{2}}(B)] \cdots \tau_{s_{n}}(B) ] \rangle_{\beta,\mu,L}\\
&= (-i)^n \int_{0}^{\beta} d s_1 \dots \int_0^{s_{n-1}} d s_n\, \Big[\prod_{j=1}^{n} a(it+s_{j})\Big] \langle \gamma_{s_{1}}(B) \cdots \gamma_{s_{n}}(B) A\rangle_{\beta,\mu,L}.
\end{split}
\end{equation}
\end{proposition}
\begin{proof} To avoid confusion, in the proof we shall call $\{r_{j}\}$ the variables corresponding to real-time integrations and $\{s_{j}\}$ the variables corresponding to imaginary-time integrations. To simplify the notations, we will omit the $\beta,\mu,L$ subscript in the Gibbs state. Using the notation (\ref{eq:Cn}), we rewrite:
\begin{equation}
\begin{split}
&\int_{-\infty \leq r_{n} \leq \ldots \leq r_{1} \leq t}  d\underline{r}\, \Big[ \prod_{i=1}^{n} a(ir_{i}) \Big] \langle [ \cdots [[\tau_{t}(A), \tau_{r_{1}}(B)], \tau_{r_{2}}(B)] \cdots \tau_{r_{n}}(B) ] \rangle \\
&\quad \equiv \int_{-\infty \leq r_{n} \leq \ldots \leq r_{1} \leq t}  d\underline{r}\, a(ir_{n}) \langle [ C_{n-1}(r_1,\dots, r_{n-1}), \tau_{r_{n}} (B) ]\rangle.
\end{split}
\end{equation}
We have:
\begin{equation}
\begin{split}
&\int_{-\infty \leq r_{n} \leq \ldots \leq r_{1} \leq t}  d\underline{r}\, a(ir_{n}) \langle [ C_{n-1}(r_1,\dots,r_{n-1}), \tau_{r_n} (B) ]\rangle \\
&\; =  -\int_{-\infty}^{t} d r_1 \dots \int_{-\infty}^{r_{n-2}} d r_{n-1} \int_{-\infty}^{r_{n-1}} dr_{n} a(ir_{n}) \langle [ \tau_{r_n} (B), C_{n-1}(r_1,\dots,r_{n-1}) ]\rangle\\
& \; = -i \int_{-\infty}^{t} d r_1 \dots \int_{-\infty}^{r_{n-2}} d r_{n-1} \int_{0}^{\beta} d s_1\, \langle  B_1(r_{n-1} - i s_1) C_{n-1}(r_1,\dots, r_{n-1})\rangle
\end{split}
\end{equation}
where the last equality follows from Proposition \ref{prop:wick1} for $j=0$, applied to the $r_{n}$ integration. Next, using again (\ref{eq:Cn}), we write:
\begin{equation}
C_{n-1}(r_1,\dots, r_{n-1}) = -a(ir_{n-1}) [ \tau_{r_{n-1}}(B) , C_{n-2}(r_1,\dots, r_{n-2}) ]
\end{equation}
and hence:
\begin{equation}
\begin{split}
&\int_{-\infty \leq r_{n} \leq \ldots \leq r_{1} \leq t}  d\underline{r}\, a(ir_{n}) \langle [ C_{n-1}(r_1,\dots,r_{n-1}), \tau_{r_n} (B) ]\rangle \\
&\quad = i \int_{-\infty}^{t} d r_1 \dots \int_{-\infty}^{r_{n-2}} d r_{n-1}\, a(ir_{n-1}) \\&\quad \qquad \cdot \int_{0}^{\beta} d s_1\, \langle  B_1(r_{n-1} - i s_1) [ \tau_{r_{n-1}}(B) , C_{n-2}(r_1,\dots, r_{n-2}) ]\rangle\\
& \quad =  (i)^{2} \int_{-\infty}^{t} d r_1 \dots \int_{-\infty}^{r_{n-3}} d r_{n-2} \\
&\qquad \cdot \int_{\Delta_{\beta}^{2}} d\underline{s}\, \langle B_{2}(r_{n-2} - is_{1}, r_{n-2} - is_{2}) C_{n-2}(r_1,\dots, r_{n-2}) \rangle\;, 
\end{split}
\end{equation}
where in the last step we applied Proposition \ref{prop:wick1} for $j=1$ to the $r_{n-1}$ integration. We continue applying Proposition \ref{prop:wick1} until all commutators are exhausted. We find:
\begin{equation}\label{eq:BB0}
\begin{split}
&\int_{-\infty \leq r_{n} \leq \ldots \leq r_{1} \leq t}  d\underline{r}\, a(ir_{n}) \langle [ C_{n-1}(r_1,\dots,r_{n-1}), \tau_{r_n} (B) ]\rangle \\ 
&\qquad  = (-i)^{n} \int_{\Delta^{n}_{\beta}} d\underline{s}\, \langle B_{n} (t-is_1,\dots, t-is_n) \tau_{t}(A) \rangle\;.
\end{split}
\end{equation}
To conclude, recall that by Eq.~(\ref{eq:B_def}):
\begin{equation}\label{eq:BB}
\begin{split}
B_{n} (t-is_1,\dots, t-is_n) &= \Big[ \prod_{i=1}^{n} a(it + s_{j}) \Big] \tau_{t - is_{1}}(B) \tau_{t - is_{2}}(B) \cdots \tau_{t - is_{n}}(B) \\ 
& \equiv \Big[ \prod_{i=1}^{n} a(it + s_{j}) \Big] \tau_{t} \Big( \gamma_{s_{1}}(B) \gamma_{s_{2}} (B) \cdots \gamma_{s_{n}}(B) \Big)
\end{split}
\end{equation}
where in the last step we used that $B$ commutes with the number operator, which implies $\tau_{-is}(B) = \gamma_{s}(B)$. Plugging (\ref{eq:BB}) into the right-hand side of (\ref{eq:BB0}), and using the invariance of the Gibbs state under time-evolution, the final claim (\ref{eq:multicomplex}) follows.
\end{proof}
Next, we rewrite the imaginary-time expressions appearing after the Wick rotation as connected correlation functions. We recall the following relation between the expectation value of a product of operators, and the truncated expectations:
\begin{equation}\label{eq:trunc}
\langle O_{i_{1}} \cdots O_{i_{n}} \rangle = \sum_{P} \prod_{J\in P} \langle O(J) \rangle^{T}
\end{equation}
where $P$ are partitions of the ordered set $\{i_1,\dots, i_{n}\}$, with elements $J = \{ j_{1}, \ldots, j_{|J|} \}$ which inherit the order of $\{i_1,\dots, i_{n}\}$, and
\begin{equation}
\langle O(J) \rangle^{T} := \langle O_{j_{1}} ; O_{j_{2}} ; \cdots ; O_{j_{|J|}} \rangle.
\end{equation}
The next result is a straightforward consequence of the definition of truncated expectation. We shall use the notation, for $J = \{ j_{1}, \ldots, j_{m} \}$:
\begin{equation}
B(-i\underline{s}_{J}) := \gamma_{s_{j_1}}(B) \cdots \gamma_{s_{j_{m}}}(B).
\end{equation}
\begin{proposition}[Factorization property]\label{prop:fact} The following identity holds true:
\begin{equation}
\begin{split}
&\langle \gamma_{s_{1}}(B) \gamma_{s_{2}} (B) \cdots \gamma_{s_{n}}(B) A\rangle \\
&\qquad = \sum_{J\subseteq \{1,\ldots, n\}} \big\langle \gamma_{s_{j_1}}(B); \gamma_{s_{j_2}} (B); \cdots; \gamma_{s_{j_{|J|}}}(B); A \big\rangle \big\langle B(-i\underline{s}_{\{1,\ldots, n\} \setminus J})  \big\rangle
\end{split}
\end{equation}
where the sum is over ordered subsets of $\{1, \ldots, n\}$.
\end{proposition}
\begin{proof} Let:
\begin{equation}
O_{1} = \gamma_{s_{1}}(B),\quad O_{2} = \gamma_{s_{2}}(B),\quad \ldots\quad, O_{n}=\gamma_{s_{n}}(B)\quad,\quad  O_{n+1} = A.
\end{equation}
From (\ref{eq:trunc}):
\begin{equation}\label{eq:C0}
\begin{split}
\langle O_{1} O_{2} \cdots O_{n+1} \rangle &= \sum_{P} \prod_{J\in P} \langle O(J) \rangle^{T}\\
&= \sum_{P} \langle O(J_{n+1}) \rangle^{T} \prod_{J\in P:\, n+1\notin J} \langle O(J) \rangle^{T} 
\end{split}
\end{equation}
where the sum is over partitions $P$ of $\{1, \ldots, n+1\}$, and $J$ are the elements of the partition. In particular, $J_{n+1}$ is the element of the partition that contains $n+1$. The right-hand side of (\ref{eq:C0}) can be rewritten as:
\begin{equation}
\langle O_{1} O_{2} \cdots O_{n+1} \rangle = \sum_{J_{n+1}} \langle O(J_{n+1}) \rangle^{T} \sum_{\widetilde{P}\, \text{of}\, \{1,\ldots, n+1\} \setminus J_{n+1}} \prod_{J\in \widetilde{P}} \langle O(J) \rangle^{T} 
\end{equation}
which we rewrite as, using again (\ref{eq:trunc}):
\begin{equation}
\begin{split}
&\langle O_{1} O_{2} \cdots O_{n+1} \rangle \\&= \sum_{J \subseteq \{1,\ldots, n\}} \langle \gamma_{s_{j_1}}(B); \gamma_{s_{j_2}} (B); \cdots; \gamma_{s_{j_{|J|}}}(B); A \rangle^{T} \Big\langle \prod_{j\in \{1,\ldots, n\} \setminus J} O_{j} \Big\rangle.
\end{split}
\end{equation}
This concludes the proof of the proposition.
\end{proof}
The next proposition allows to rewrite the right-hand side of (\ref{eq:multicomplex}) in terms of truncated correlation functions, in Euclidean time.
\begin{proposition}[Reduction to connected Euclidean correlations]\label{prop:conn} Under the same assumptions of Lemma \ref{lem:wick} the following identity holds:
\begin{equation}\label{eq:pv}
\begin{split}
\int_{\Delta^{n}_\beta} d\underline{s}\, &\Big[\prod_{j=1}^{n} a(it+s_{j}) \Big] \langle \gamma_{s_{1}}(B) \gamma_{s_{2}}(B) \cdots \gamma_{s_{n}}(B) A\rangle_{\beta, \mu, L} \\& = \int_{\Delta^{n}_{\beta}} d\underline{s}\, \Big[\prod_{j=1}^{n} a(it+s_{j})\Big] \langle \gamma_{s_{1}}(B); \gamma_{s_{2}}(B); \cdots; \gamma_{s_{n}}(B); A \rangle_{\beta, \mu, L}.
\end{split}
\end{equation}
\end{proposition}
\begin{proof} We omit the $\beta, \mu, L$ labels for simplicity. By Proposition \ref{prop:fact} we have:
\begin{equation}
\begin{split}
&\langle \gamma_{s_{1}}(B) \gamma_{s_{2}} (B) \cdots \gamma_{s_{n}}(B) A\rangle \\
&\quad = \sum_{J\subseteq \{1,\ldots, n\}} \big\langle \gamma_{s_{j_1}}(B); \gamma_{s_{j_2}} (B); \cdots; \gamma_{s_{j_{|J|}}}(B); A \big\rangle \big\langle B(-i\underline{s}_{\{1,\ldots, n\} \setminus J})  \big\rangle.
\end{split}
\end{equation}
Hence, we can rewrite the left-hand side of (\ref{eq:pv}) as:
\begin{equation}\label{eq:pvv}
\begin{split}
&\int_{\Delta^{n}_\beta} d\underline{s}\, \Big[\prod_{j=1}^{n} a(it+s_{j}) \Big] \langle \gamma_{s_{1}}(B) \gamma_{s_{2}} (B) \cdots \gamma_{s_{n}}(B) A\rangle \\
& =\sum_{m=0}^{n} \int_{\Delta^{n}_\beta} d\underline{s}\, \sum_{\substack{J \subseteq \{1,\ldots, n\} \\ |J| = m}} \Big[\prod_{j\in J} a(it+s_{j}) \Big] \big\langle \gamma_{s_{j_1}}(B); \cdots; \gamma_{s_{j_{m}}}(B); A \big\rangle \\& \qquad\qquad\qquad\qquad\qquad\qquad \cdot \Big[\prod_{j\in \{1,\ldots, n\} \setminus J} a(it+s_{j}) \Big] \big\langle B(-i\underline{s}_{\{1,\ldots, n\} \setminus J})  \big\rangle.
\end{split}
\end{equation}
Next, we shall use the following identity:
\begin{equation}\label{eq:pvvv}
\begin{split}
&\int_{\Delta^{n}_\beta} d\underline{s}\, \sum_{\substack{J \subseteq \{1,\ldots, n\} \\ |J| = m}} \Big[\prod_{j\in J} a(it+s_{j}) \Big] \big\langle \gamma_{s_{j_1}}(B); \cdots; \gamma_{s_{j_{m}}}(B); A \big\rangle \\& \qquad\qquad\qquad\qquad\qquad\qquad \cdot \Big[\prod_{j\in \{1,\ldots, n\} \setminus J} a(it+s_{j}) \Big] \big\langle B(-i\underline{s}_{\{1,\ldots, n\} \setminus J})  \big\rangle \\
& \qquad = \int_{\Delta^{m}_\beta} d\underline{s}\, \Big[ \prod_{i=1}^{m} a(it+s_{i})\Big]  \big\langle \gamma_{s_{1}}(B); \cdots; \gamma_{s_{m}}(B); A \big\rangle \\&\qquad \qquad \cdot \int_{\Delta^{n-m}_\beta} d\underline{s}\, \Big[ \prod_{i=m+1}^{n} a(it+s_{i})\Big] \langle \gamma_{s_{m+1}}(B) \gamma_{s_{m+2}}(B)\cdots \gamma_{s_{n}}(B)  \rangle.
 \end{split}
\end{equation}
Eq.~(\ref{eq:pvvv}) is obtained via the application of Proposition \ref{lem:splitting} in Appendix \ref{app:int}, with the following choices for the functions $f$ and $g$:
\begin{equation}
\begin{split}
f(s_{1}, \ldots, s_{m}) &= \Big[ \prod_{i=1}^{m} a(it+s_{i})\Big]  \big\langle \gamma_{s_{1}}(B); \cdots; \gamma_{s_{m}}(B); A \big\rangle \\
g(s_{m+1}, \ldots, s_{n}) &= \Big[ \prod_{i=m+1}^{n} a(it+s_{i})\Big] \langle \gamma_{s_{m+1}}(B) \gamma_{s_{m+2}}(B)\cdots \gamma_{s_{n}}(B)  \rangle.
\end{split}
\end{equation}
We claim that, for all $k>0$:
\begin{equation}\label{eq:claim0}
\int_{\Delta^{k}_\beta} d\underline{s}\, \Big[ \prod_{i=1}^{k} a(it+s_{i})\Big] \langle \gamma_{s_{1}}(B) \gamma_{s_{2}}(B)\cdots \gamma_{s_{k}}(B)  \rangle = 0.
\end{equation}
Combined with (\ref{eq:pvv}), (\ref{eq:pvvv}), this implies the final statement, Eq.~(\ref{eq:pv}): the only term contributing to the sum over $m$ in Eq.~(\ref{eq:pvv}) is $m=n$. To prove Eq.~(\ref{eq:claim0}), we proceed as follows. First, we write:
\begin{equation}\label{eq:permut}
\begin{split}
\int_{\Delta^{k}_\beta} d\underline{s}\, &\Big[ \prod_{i=1}^{k} a(it+s_{i})\Big] \langle \gamma_{s_{1}}(B) \gamma_{s_{2}}(B)\cdots \gamma_{s_{k}}(B)  \rangle \\
&= \frac{1}{k!} \int_{[0,\beta]^{k}} d\underline{s}\,  \Big[ \prod_{i=1}^{k} a(it+s_{i})\Big] \sum_{\pi} \mathbbm{1}(s_{\pi(1)} > s_{\pi(2)} > \ldots > s_{\pi(k)}) \\&\qquad \qquad \cdot  \langle \gamma_{s_{\pi(1)}}(B) \gamma_{s_{\pi(2)}}(B)\cdots \gamma_{s_{\pi(k)}}(B)  \rangle
\end{split}
\end{equation}
where the sum is over permutations of $\{1, \ldots, k\}$. Let:
\begin{equation}
\begin{split}
&G(s_{1}, \ldots, s_{k}) \\
&\quad := \sum_{\pi} \mathbbm{1}(s_{\pi(1)} > s_{\pi(2)} > \ldots > s_{\pi(k)}) \langle \gamma_{s_{\pi(1)}}(B) \gamma_{s_{\pi(2)}}(B)\cdots \gamma_{s_{\pi(k)}}(B)  \rangle.
\end{split}
\end{equation}
We claim that $G$ is $\beta$-periodic in all its arguments:
\begin{equation}\label{eq:period}
G(s_{1}, \ldots, s_{i-1}, 0, s_{i+1}, \ldots, s_{k}) = G(s_{1}, \ldots, s_{i-1}, \beta, s_{i+1}, \ldots, s_{k}).
\end{equation}
In particular, the function $G$ extends to a periodic function on $\mathbb{R}^{k}$, with period $\beta$ in all variables. With a slight abuse of notation, let us denote by $G$ such periodic extension. Notice that the function $s\mapsto a(it+s)$ is also periodic with period $\beta$ (recall definition (\ref{eq:adef0})), which means that the whole integrand in the right-hand side of (\ref{eq:permut}) can be extended to a $\beta$-periodic function in all its arguments. Furthermore, we claim that, for all $\sigma \in \mathbb{R}$:
\begin{equation}\label{eq:transinv}
G(s_{1}, s_{2}, \ldots, s_{k}) = G(s_{1} + \sigma, s_{2} + \sigma, \ldots, s_{k} + \sigma),
\end{equation}
that is, the function $G$ is translation invariant. Both (\ref{eq:period}), (\ref{eq:transinv}) are well known; they ultimately follow from the KMS identity. For the sake of completeness, Eqs. (\ref{eq:period}), (\ref{eq:transinv}) will be reviewed in Appendix \ref{app:int}, Proposition \ref{prop:period}. Thus, one gets:
\begin{equation}
\begin{split}
\int_{\Delta^{k}_\beta} d\underline{s}\, &\Big[ \prod_{i=1}^{k} a(it+s_{i})\Big] \langle \gamma_{s_{1}}(B) \gamma_{s_{2}}(B)\cdots \gamma_{s_{k}}(B)  \rangle \\& \equiv \frac{1}{k!} \int_{(S^{1}_{\beta})^{k}} d\underline{s}\,  \Big[ \prod_{i=1}^{k} a(it+s_{i})\Big] G(s_{1}, s_{2}, \ldots, s_{k}),
\end{split}
\end{equation}
where $S^{1}_{\beta} = \mathbb{R} / \beta \mathbb{Z}$. We rewrite this expression as:
\begin{equation}
\begin{split}
&\frac{1}{k!} \sum_{\omega_{i} \in \frac{2\pi}{\beta} \mathbb{N}} \Big[ \prod_{i=1}^{k} \tilde a(\omega_{i}) e^{\omega_{i} t}\Big]  \\
&\qquad \cdot \int_{(S^{1}_{\beta})^{k}} d\underline{s}\, e^{-i\sum_{j=1}^{k} \omega_{j} s_{1} } e^{-i\sum_{j=1}^{k} \omega_{j} (s_{j} - s_{1})} G(0, s_{2} - s_{1}, \ldots, s_{k} - s_{1}) \\
&= \frac{1}{k!} \sum_{\omega_{i} \in \frac{2\pi}{\beta} \mathbb{N}} \Big[ \prod_{i=1}^{k} \tilde a(\omega_{i}) e^{\omega_{i} t}\Big] \\
&\qquad \cdot \int_{S^{1}_{\beta}} ds_{1}\, e^{-i \sum_{j=1}^{k} \omega_{j} s_{1} }\int_{(S^{1}_{\beta})^{k-1}} d\underline{s}\,e^{-i\sum_{j=2}^{k} \omega_{j} s_{j}} G(0, s_{2}, \ldots, s_{k}),
\end{split}
\end{equation}
where in the last step we used that $e^{-i\sum_{j=1}^{k} \omega_{j} (s_{j} - s_{1})} G(0, s_{2} - s_{1}, \ldots, s_{k} - s_{1})$, as a function of $s_{j}$, $j=2,\ldots, k$, is a function on $(S^{1}_{\beta})^{k-1}$. Then, the claim (\ref{eq:claim0}) follows from (recall that we can assume $\omega_{j} \geq \frac{2\pi}{\beta}$, since $\tilde a(0) = 0$):
\begin{equation}
\int_{S^{1}_{\beta}} ds_{1}\, e^{-i \sum_{j=1}^{k} \omega_{j} s_{1} } = 0.
\end{equation}
This concludes the proof of Proposition \ref{prop:conn}.
\end{proof}
\begin{remark}\label{rem:T}
	By the same arguments used in the proof of Proposition \ref{prop:conn}, Eq.~(\ref{eq:pv}) can also be rephrased as:
\begin{equation}\label{eq:Tidentity}
\begin{split}
&\int_{\Delta^{n}_\beta} d\underline{s}\, \Big[\prod_{j=1}^{n} a(it+s_{j}) \Big] \langle  \gamma_{s_{1}}(B) \gamma_{s_{2}}(B) \cdots \gamma_{s_{n}}(B) A\rangle_{\beta, \mu, L} \\
&\qquad = \frac{1}{n!} \int_{(S^{1}_{\beta})^{n}} d\underline{s}\,  \Big[\prod_{j=1}^{n} a(it+s_{j}) \Big] \langle {\bf T} \gamma_{s_{1}}(B); \gamma_{s_{2}}(B); \cdots \gamma_{s_{n}}(B); A\rangle_{\beta, \mu, L}
\end{split}
\end{equation}
where ${\bf T}$ denotes the time-ordering, as defined in Eq.~(\ref{eq:time}). This is initially defined for operators whose imaginary-time arguments are in $[0,\beta)$. The resulting expression is then extended to a periodic function with period $\beta$ on the whole $\mathbb{R}^{n}$. See Proposition \ref{prop:period} and Remark \ref{rem:period} for further details.
\end{remark}
We are now ready to prove Lemma \ref{lem:wick}.
\begin{proof}[Proof of Lemma \ref{lem:wick}] By Proposition \ref{prop:wick2}:
\begin{equation}
\begin{split}
&\int_{-\infty \leq s_{n} \leq \ldots \leq s_{1} \leq t}  d\underline{s}\, \Big[\prod_{j=1}^{n} a(is_{j})\Big] \langle [ \cdots [[\tau_{t}(A), \tau_{s_{1}}(B)], \tau_{s_{2}}(B)] \cdots \tau_{s_{n}}(B) ] \rangle_{\beta,\mu,L} \\
&= (-i)^{n} \int_{0}^{\beta} d s_1 \dots \int_0^{s_{n-1}} d s_n\, \Big[\prod_{j=1}^{n} a(it + s_{j})\Big] \langle \gamma_{s_{1}}(B) \cdots \gamma_{s_{n}}(B) A\rangle_{\beta,\mu,L}.
\end{split}
\end{equation}
Next, by Proposition \ref{prop:conn}:
\begin{equation}
\begin{split}
&\int_{-\infty \leq s_{n} \leq \ldots \leq s_{1} \leq t}  d\underline{s}\, \Big[\prod_{j=1}^{n} a(is_{j})\Big] \langle [ \cdots [[\tau_{t}(A), \tau_{s_{1}}(B)], \tau_{s_{2}}(B)] \cdots \tau_{s_{n}}(B) ] \rangle_{\beta,\mu,L}\\
&= (-i)^{n} \int_{0}^{\beta} d s_1 \dots \int_0^{s_{n-1}} d s_n\, \Big[\prod_{j=1}^{n} a(it + s_{j})\Big] \langle \gamma_{s_{1}}(B); \cdots; \gamma_{s_{n}}(B); A\rangle_{\beta,\mu,L}.
\end{split}
\end{equation}
Finally, by Remark \ref{rem:T}:
\begin{equation}
\begin{split}
&\int_{-\infty \leq s_{n} \leq \ldots \leq s_{1} \leq t}  d\underline{s}\,  \Big[\prod_{j=1}^{n} a(is_{j})\Big] \langle [ \cdots [[\tau_{t}(A), \tau_{s_{1}}(B)], \tau_{s_{2}}(B)] \cdots \tau_{s_{n}}(B) ] \rangle_{\beta,\mu,L} \\
&\quad = \frac{(-i)^{n}}{n!}\int_{(S^{1}_{\beta})^{n}} d\underline{s}\, \Big[\prod_{j=1}^{n} a(it + s_{j})\Big] \langle {\bf T} \gamma_{s_{1}}(B); \cdots; \gamma_{s_{n}}(B); A\rangle_{\beta,\mu,L}\;,
\end{split}
\end{equation}
which concludes the proof of Lemma \ref{lem:wick}.
\end{proof}
\subsection{Cumulant expansion for the instantaneous Gibbs state}\label{sec:compa}
In this section we shall review the well-known cumulant expansion for the Gibbs state of the Hamiltonian $\mathcal{H}(\eta t)$,
\begin{equation}
\mathcal{H}(\eta t) = \mathcal{H} + \varepsilon g(\eta t) \mathcal{P},
\end{equation}
that is:
\begin{equation}
\langle \mathcal{O}_{X} \rangle_{t} = \frac{ \Tr \mathcal{O}_{X} e^{-\beta (\mathcal{H}(\eta t) - \mu \mathcal{N})} }{ \Tr e^{-\beta (\mathcal{H}(\eta t) - \mu \mathcal{N})} }.
\end{equation}
Perturbation theory in $\varepsilon$ is generated by the following chain of identities:
\begin{equation}
\begin{split}
e^{\beta (\mathcal{H}- \mu \mathcal{N})} &e^{-\beta (\mathcal{H}(\eta t) - \mu \mathcal{N})} - \mathbbm{1} \\ &= \int_{0}^{\beta} ds\, \frac{\partial}{\partial s} e^{s (\mathcal{H}- \mu \mathcal{N})} e^{-s (\mathcal{H}(\eta t) - \mu \mathcal{N})} \\
&= -\varepsilon g(\eta t) \int_{0}^{\beta} ds\,  e^{s (\mathcal{H}- \mu \mathcal{N})} \mathcal{P} e^{-s (\mathcal{H}(\eta t) - \mu \mathcal{N})} \\
&= -\varepsilon g(\eta t) \int_{0}^{\beta} ds\,  e^{s (\mathcal{H}- \mu \mathcal{N})} \mathcal{P} e^{-s (\mathcal{H} - \mu \mathcal{N})} e^{s (\mathcal{H} - \mu \mathcal{N})} e^{-s (\mathcal{H}(\eta t) - \mu \mathcal{N})} \\
&\equiv -\varepsilon g(\eta t) \int_{0}^{\beta} ds\,  \gamma_{s}(\mathcal{P}) e^{s (\mathcal{H}- \mu \mathcal{N})} e^{-s (\mathcal{H}(\eta t) - \mu \mathcal{N})}.
\end{split}
\end{equation}
Iterating:
\begin{equation}
\begin{split}
&e^{\beta (\mathcal{H}- \mu \mathcal{N})} e^{-\beta (\mathcal{H}(\eta t) - \mu \mathcal{N})} \\
&\quad = \mathbbm{1} + \sum_{n\geq 1} (-\varepsilon g(\eta t))^{n} \int_{0}^{\beta} ds_{1} \int_{0}^{s_{1}} ds_{2} \cdots \int_{0}^{s_{n-1}} ds_{n}\, \gamma_{s_{1}}(\mathcal{P})\cdots \gamma_{s_{n}}(\mathcal{P})
\end{split}
\end{equation}
which we can also write as:
\begin{equation}
\begin{split}
&e^{-\beta (\mathcal{H}(\eta t) - \mu \mathcal{N})} \\&\qquad = e^{-\beta (\mathcal{H}- \mu \mathcal{N})}\Big[ \mathbbm{1} + \sum_{n\geq 1} \frac{(-\varepsilon g(\eta t))^{n}}{n!} \int_{[0,\beta)^{n}} d\underline{s}\, \mathbf{T} \gamma_{s_{1}}(\mathcal{P})\cdots \gamma_{s_{n}}(\mathcal{P}) \Big].
\end{split}
\end{equation}
For finite $L$ and finite $\beta$, the series is norm convergent, thanks to the boundedness of the fermionic operators. Thus, the expectation value of a local operator on the Gibbs state of $\mathcal{H}(\eta t)$ can be written as:
\begin{equation}\label{eq:OX}
\langle \mathcal{O}_{X} \rangle_{t} = \frac{\langle \mathcal{O}_{X} \rangle_{\beta,\mu,L} + \sum_{n\geq 1} \frac{(-\varepsilon g(\eta t))^{n}}{n!}\int_{[0,\beta)^{n}} d\underline{t}\, \langle \mathbf{T} \gamma_{t_{1}}(\mathcal{P})\cdots \gamma_{t_{n}}(\mathcal{P}) \mathcal{O}_{X}  \rangle_{\beta, \mu, L}}{1 + \sum_{n\geq 1} \frac{(-\varepsilon g(\eta t))^{n}}{n!}\int_{[0,\beta)^{n}} d\underline{t}\, \langle \mathbf{T} \gamma_{t_{1}}(\mathcal{P})\cdots \gamma_{t_{n}}(\mathcal{P})\rangle_{\beta, \mu, L}}
\end{equation}
which is analytic in $\varepsilon$ for $|\varepsilon|$ small enough. We would like to show that analyticity in $\varepsilon$ extends to a ball whose radius is bounded uniformly in $L$. To this end, Eq.~(\ref{eq:OX}) can be further rewritten as (omitting the $\beta, \mu, L$ labels in the state and the $[0,\beta)^{n}$ domain in the integral):
\begin{equation}
\begin{split}\label{eq:equi0}
&\langle \mathcal{O}_{X} \rangle_{t} \\
&= \frac{\partial}{\partial \zeta} \log \Big( \sum_{n,m \geq 0} \frac{(-\varepsilon g(\eta t))^{n}}{n!} \frac{\zeta^{m}}{m!} \int d\underline{s}\, \langle \mathbf{T} \gamma_{s_{1}}(\mathcal{P})\cdots \gamma_{s_{n}}(\mathcal{P}) \mathcal{O}_{X}^{m}  \rangle \Big)\Big|_{\zeta = 0} \\
&= \sum_{n\geq 0} \frac{\varepsilon^{n}}{n!} \\&\cdot \frac{\partial^{n}}{\partial \varepsilon^{n}} \frac{\partial}{\partial \zeta} \log \Big( \sum_{\ell,m \geq 0}\frac{(-\varepsilon g(\eta t))^{\ell}}{\ell!} \frac{\zeta^{m}}{m!} \int d\underline{s}\, \langle \mathbf{T} \gamma_{s_{1}}(\mathcal{P})\cdots \gamma_{s_{\ell}}(\mathcal{P}) \mathcal{O}_{X}^{m}  \rangle \Big)\Big|_{\substack{\varepsilon = 0 \\ \zeta = 0}}.
\end{split}
\end{equation}
Then, it is not difficult to see that the right-hand side can be written as a sum over time-ordered cumulants, defined as in Eq.~(\ref{eq:Tcumul}). We have:
\begin{equation}\label{eq:equi}
\langle \mathcal{O}_{X} \rangle_{t} = \langle \mathcal{O}_{X} \rangle + \sum_{n\geq 1} \frac{(-\varepsilon g(\eta t))^{n}}{n!} \int d\underline{s}\, \langle \mathbf{T} \gamma_{s_{1}}(\mathcal{P}); \cdots; \gamma_{s_{n}}(\mathcal{P}); \mathcal{O}_{X}  \rangle.
\end{equation}
Under the assumption (\ref{eq:assA}), the series converges for $|\varepsilon|$ small enough, uniformly in $L$. By using Lemma \ref{lem:wick}, we will show that, for $\eta$ small enough, the Duhamel series of the auxiliary dynamics is term-by-term close to the cumulant expansion of the instantaneous Gibbs state, Eq.~(\ref{eq:equi}).

\subsection{Conclusion: proof of Theorem \ref{thm:main}}\label{sec:fin}

We are now ready to prove our main result, Theorem \ref{thm:main}. 
\begin{proof}[Proof of Theorem \ref{thm:main}]
By Proposition \ref{prop:comp} we have, for all $t\leq 0$:
\begin{equation}\label{eq:first}
\begin{split}
\Tr \mathcal{O}_{X} \rho(t) &= \Tr \mathcal{O}_{X} \tilde \rho(t) + R_{1}(t) \\
|R_{1}(t)| &\leq \frac{K|\varepsilon|}{\eta^{d+2}\beta},
\end{split}
\end{equation}
where $\tilde \rho(t)$ is the evolution of the equilibrium state under the Hamiltonian $\widetilde{\mathcal{H}}(\eta t)$, Eq.~(\ref{eq:tildeH}). Next, we rewrite the first term via its Duhamel series, as discussed in Section \ref{sec:duha}. We have, from Eq.~(\ref{eq:duh2}), replacing $g(\eta t)$ with $g_{\beta, \eta}(t)$:
\begin{equation}\label{eq:expa}
\begin{split}
&\Tr \mathcal{O}_{X} \tilde \rho(t) \\
&=  \Tr \mathcal{O}_{X} \rho_{\beta, \mu, L} + \sum_{n=1}^{\infty} (-i\varepsilon)^{n} \int_{-\infty \leq s_{n} \leq \ldots \leq s_{1} \leq t} d \underline{s}\, \Big[ \prod_{i=1}^{n} g_{\beta, \eta}(s_{i}) \Big] \\&\qquad \qquad \qquad \cdot \langle [ \cdots [[\tau_{t}(\mathcal{O}_{X}), \tau_{s_{1}}(\mathcal{P})], \tau_{s_{2}}(\mathcal{P})] \cdots \tau_{s_{n}}(\mathcal{P}) ] \rangle_{\beta, \mu, L}.
\end{split}
\end{equation} 
Consider the integral. We apply Lemma \ref{lem:wick}, choosing:
\begin{equation}
A = \mathcal{O}_{X}\;,\qquad B = \mathcal{P}\;,\qquad a(s) = g_{\beta, \eta}(-is)\;.
\end{equation}
We have, omitting the $\beta, \mu, L$ labels:
\begin{equation}\label{eq:wickproof}
\begin{split}
&\int_{-\infty \leq s_{n} \leq \ldots \leq s_{1} \leq t} d \underline{s}\, \Big[ \prod_{i=1}^{n} g_{\beta, \eta}(s_{i}) \Big] \langle [ \cdots [[\tau_{t}(\mathcal{O}_{X}), \tau_{s_{1}}(\mathcal{P})], \tau_{s_{2}}(\mathcal{P})] \cdots \tau_{s_{n}}(\mathcal{P}) ] \rangle \\
& = \frac{(-i)^{n}}{n!} \int_{[0,\beta)^{n}} d\underline{s}\, \Big[\prod_{j=1}^{n} g_{\beta, \eta}(t-is_{j})\Big] \langle {\bf T} \gamma_{s_{1}}(\mathcal{P}); \gamma_{s_{2}}(\mathcal{P}); \cdots; \gamma_{s_{n}}(\mathcal{P}); \mathcal{O}_{X} \rangle.
\end{split}
\end{equation}
Eqs. (\ref{eq:first}), (\ref{eq:expa}), (\ref{eq:wickproof}) prove the identity (\ref{eq:mainexp}). The estimate (\ref{eq:Rest}) follows from the bound in (\ref{eq:first}). To prove the bound (\ref{eq:Iest}), we use that:
\begin{equation}
\begin{split}
&|I^{(n)}_{\beta,\mu,L}(\eta,t)| \\
&\qquad = \Big|\int_{[0,\beta)^{n}} d\underline{s}\, \Big[\prod_{j=1}^{n} g_{\beta, \eta}(t-is_{j})\Big] \langle {\bf T} \gamma_{s_{1}}(\mathcal{P}); \gamma_{s_{2}}(\mathcal{P}); \cdots; \gamma_{s_{n}}(\mathcal{P}); \mathcal{O}_{X} \rangle\Big| \\
&\qquad \leq \| h \|_{1}^{n} \sum_{\substack{X_{1}, \ldots, X_{n} \subseteq \Lambda_{L} \\ \text{diam} X_{i} \leq R}} \int_{[0,\beta)^{n}} d\underline{s}\, \Big|\langle {\bf T} \gamma_{s_{1}}(\mathcal{P}_{X_{1}}); \cdots; \gamma_{s_{n}}(\mathcal{P}_{X_{n}}); \mathcal{O}_{X} \rangle\Big| \\
&\qquad \leq \| h \|_{1}^{n} \mathfrak{c}^{n} n!
\end{split}
\end{equation}
where in the first inequality we used the estimate (\ref{eq:tildegbd}), while the last inequality follows from Assumption \ref{ass:dec}. This proves the bound (\ref{eq:Iest}), which shows that series in Eq.~(\ref{eq:mainexp}) is absolutely convergent for:
\begin{equation}
|\varepsilon| < \frac{1}{\mathfrak{c} \|h\|_{1}}.
\end{equation}
To conclude, let us prove Eq.~(\ref{eq:adia}). Rewriting the functions $g_{\beta, \eta}(t-is_{j})$ as in (\ref{eq:gbeta}), we get:
\begin{equation}
\begin{split}
&\int_{[0,\beta)^{n}} d\underline{s}\, \Big[\prod_{j=1}^{n} g_{\beta, \eta}(t-is_{j})\Big] \langle {\bf T} \gamma_{s_{1}}(\mathcal{P}); \gamma_{s_{2}}(\mathcal{P}); \cdots; \gamma_{s_{n}}(\mathcal{P}); \mathcal{O}_{X} \rangle \\
& = \sum_{\underline{\omega} \in \frac{2\pi}{\beta} \mathbb{N}^{n}} \Big[ \prod_{i=1}^{n} \tilde g_{\beta, \eta}(\omega_{i}) e^{\omega_{i} t}\Big] \\ &\quad \cdot \int_{[0,\beta)^{n}} d\underline{s}\, e^{-i\sum_{i=1}^{n} \omega_{i} s_{i}} \langle {\bf T} \gamma_{s_{1}}(\mathcal{P}); \gamma_{s_{2}}(\mathcal{P}); \cdots; \gamma_{s_{n}}(\mathcal{P}); \mathcal{O}_{X} \rangle.
\end{split}
\end{equation}
Let:
\begin{equation}
\begin{split}
&\langle {\bf T} \widehat{\mathcal{P}}_{\omega_{1}}; \widehat{\mathcal{P}}_{\omega_{2}}; \cdots; \widehat{\mathcal{P}}_{\omega_{n}}; \mathcal{O}_{X} \rangle \\&\quad := 
\int_{[0,\beta)^{n}} d\underline{s}\, e^{-i\sum_{i=1}^{n} \omega_{i} s_{i}} \langle {\bf T} \gamma_{s_{1}}(\mathcal{P}); \gamma_{s_{2}}(\mathcal{P}); \cdots; \gamma_{s_{n}}(\mathcal{P}); \mathcal{O}_{X} \rangle.
\end{split}
\end{equation}
We can rewrite Eq.~(\ref{eq:expa}) in terms of these functions as:
\begin{equation}\label{eq:PP}
\begin{split}
&\Tr \mathcal{O}_{X} \tilde \rho(t) =  \Tr \mathcal{O}_{X} \rho_{\beta, \mu, L}\\&\quad + \sum_{n=1}^{\infty} \frac{(-i\varepsilon)^{n}}{n!} \sum_{\underline{\omega} \in \frac{2\pi}{\beta} \mathbb{N}^{n}} \Big[ \prod_{i=1}^{n} \tilde g_{\beta, \eta}(\omega_{i}) e^{\omega_{i} t}\Big] \langle {\bf T} \widehat{\mathcal{P}}_{\omega_{1}}; \widehat{\mathcal{P}}_{\omega_{2}}; \cdots; \widehat{\mathcal{P}}_{\omega_{n}}; \mathcal{O}_{X} \rangle,
\end{split}
\end{equation}
which is absolutely convergent, since as implied by Assumption \ref{ass:dec}
\begin{equation}\label{eq:PPPest}
|\langle {\bf T} \widehat{\mathcal{P}}_{\omega_{1}}; \widehat{\mathcal{P}}_{\omega_{2}}; \cdots; \widehat{\mathcal{P}}_{\omega_{n}}; \mathcal{O}_{X} \rangle| \leq \mathfrak{c}^{n} n!
\end{equation}
To prove Eq.~(\ref{eq:adia}), we preliminarily observe that, from Eq.~(\ref{eq:equi}):
\begin{equation}
\langle \mathcal{O}_{X} \rangle_{t} = \Tr \mathcal{O}_{X} \rho_{\beta, \mu, L} + \sum_{n\geq 1} \frac{(-\varepsilon g(\eta t))^{n}}{n!} \langle {\bf T} \widehat{\mathcal{P}}_{0}; \widehat{\mathcal{P}}_{0}; \cdots; \widehat{\mathcal{P}}_{0}; \mathcal{O}_{X} \rangle\;.
\end{equation}
Therefore,
\begin{equation}\label{eq:diffrho}
\begin{split}
&\Tr \mathcal{O}_{X} \tilde \rho(t) - \langle \mathcal{O}_{X} \rangle_{t} = \sum_{n=1}^{\infty} \frac{(-\varepsilon)^{n}}{n!} \\&\qquad\cdot \Big[ \sum_{\underline{\omega} \in \frac{2\pi}{\beta} \mathbb{N}^{n}} \Big[ \prod_{i=1}^{n} \tilde g_{\beta, \eta}(\omega_{i}) e^{\omega_{i} t}\Big] \langle {\bf T} \widehat{\mathcal{P}}_{\omega_{1}}; \widehat{\mathcal{P}}_{\omega_{2}}; \cdots; \widehat{\mathcal{P}}_{\omega_{n}}; \mathcal{O}_{X} \rangle \\&\qquad\qquad - g(\eta t)^{n} \langle {\bf T} \widehat{\mathcal{P}}_{0}; \widehat{\mathcal{P}}_{0}; \cdots; \widehat{\mathcal{P}}_{0}; \mathcal{O}_{X} \rangle\Big];
\end{split}
\end{equation}
the expression in the square brackets can be rewritten as
\begin{equation}\label{eq:RRR}
\begin{split}
&\sum_{\underline{\omega} \in \frac{2\pi}{\beta} \mathbb{N}^{n}} \Big[ \prod_{i=1}^{n} \tilde g_{\beta, \eta}(\omega_{i}) e^{\omega_{i} t}\Big] \Big(\langle {\bf T} \widehat{\mathcal{P}}_{\omega_{1}}; \widehat{\mathcal{P}}_{\omega_{2}}; \cdots; \widehat{\mathcal{P}}_{\omega_{n}}; \mathcal{O}_{X} \rangle - \langle {\bf T} \widehat{\mathcal{P}}_{0}; \widehat{\mathcal{P}}_{0}; \cdots; \widehat{\mathcal{P}}_{0}; \mathcal{O}_{X} \rangle\Big) \\&\quad + \Big(g_{\beta, \eta}(t)^{n} - g(\eta t)^{n}\Big) \langle {\bf T} \widehat{\mathcal{P}}_{0}; \widehat{\mathcal{P}}_{0}; \cdots; \widehat{\mathcal{P}}_{0}; \mathcal{O}_{X} \rangle =: R^{(n)}_{2,1}(t) + R^{(n)}_{2,2}(t).
\end{split}
\end{equation}
Consider the term $R^{(n)}_{2,1}(t)$. We have:
\begin{equation}\label{eq:aaa}
\begin{split}
| R^{(n)}_{2,1}(t) | &\leq \sum_{\underline{\omega} \in \frac{2\pi}{\beta} \mathbb{N}^{n}} \Big[ \prod_{i=1}^{n} |\tilde g_{\beta, \eta}(\omega_{i})| e^{\omega_{i} t}\Big] \\&\quad \cdot \int_{[0,\beta)^{n}} d\underline{s}\, \big| e^{i\underline{\omega}\cdot \underline{s}} - 1 \big| \big|\langle {\bf T} \gamma_{s_{1}}(\mathcal{P}); \gamma_{s_{2}}(\mathcal{P}); \cdots; \gamma_{s_{n}}(\mathcal{P}); \mathcal{O}_{X} \rangle \big| \\
&\leq \Big[ \sum_{\underline{\omega} \in \frac{2\pi}{\beta} \mathbb{N}^{n}} \Big[ \prod_{i=1}^{n} |\tilde g_{\beta, \eta}(\omega_{i})| e^{\omega_{i} t}\Big] |\underline{\omega}| \Big] \\&\quad \cdot \int_{[0,\beta)^{n}} d\underline{s}\, |\underline{s}|_{\beta} \big| \langle {\bf T} \gamma_{s_{1}}(\mathcal{P}); \gamma_{s_{2}}(\mathcal{P}); \cdots; \gamma_{s_{n}}(\mathcal{P}); \mathcal{O}_{X} \rangle\big|,
\end{split}
\end{equation}
where $|\underline{\omega}| = \sum_{i=1}^{n} |\omega_{i}|$ and $|\underline{s}|_{\beta}$ is defined in Eq.~(\ref{eq:normbeta}). The integral in the right-hand side is estimated using Assumption \ref{ass:dec}:
\begin{equation}
\int_{[0,\beta)^{n}} d\underline{s}\, |\underline{s}|_{\beta} \big| \langle {\bf T} \gamma_{s_{1}}(\mathcal{P}); \gamma_{s_{2}}(\mathcal{P}); \cdots; \gamma_{s_{n}}(\mathcal{P});\mathcal{O}_{X} \rangle\big| \leq \mathfrak{c}^{n} n!
\end{equation}
Then, the argument of the square brackets in the right-hand side of (\ref{eq:aaa}) is bounded as follows:
\begin{equation}\label{eq:argu}
\sum_{\underline{\omega} \in \frac{2\pi}{\beta} \mathbb{N}^{n}} \Big[ \prod_{i=1}^{n} |\tilde g_{\beta, \eta}(\omega_{i})| e^{\omega_{i} t}\Big] |\underline{\omega}| \leq n \Big(\sum_{\omega \in \frac{2\pi}{\beta} \mathbb{N}} \omega |\tilde g_{\beta, \eta}(\omega)| \Big) \| h \|_{1}^{n-1}
\end{equation}
where we used Eq.~(\ref{eq:tildegbd}). The sum in the right-hand side of (\ref{eq:argu}) is estimated as:
\begin{equation}\label{eq:omegag}
\begin{split}
\sum_{\omega \in \frac{2\pi}{\beta} \mathbb{N}} \omega |\tilde g_{\beta, \eta}(\omega)| &= \sum_{\omega \in \frac{2\pi}{\beta} \mathbb{N}} \omega \Big| \int_{\frac{\omega}{\eta} - \frac{2\pi}{\beta \eta}}^{\frac{\omega}{\eta}} d\xi\, h(\xi)\Big| \\
&\leq \sum_{\omega \in \frac{2\pi}{\beta} \mathbb{N}} \Big(\omega - \frac{2\pi}{\beta}\Big) \Big| \int_{\frac{\omega}{\eta} - \frac{2\pi}{\beta \eta}}^{\frac{\omega}{\eta}} d\xi\, h(\xi)\Big|  + \frac{2\pi}{\beta} \|h\|_{1} \\
&\leq \sum_{\omega \in \frac{2\pi}{\beta} \mathbb{N}} \eta \int_{\frac{\omega}{\eta} - \frac{2\pi}{\beta \eta}}^{\frac{\omega}{\eta}} d\xi\, \xi |h(\xi)| + \frac{2\pi}{\beta} \|h\|_{1} = \eta \| \xi h \|_{1} + \frac{2\pi}{\beta} \|h\|_{1}.
\end{split}
\end{equation}
All together,
\begin{equation}\label{eq:R21}
| R^{(n)}_{2,1}(t) | \leq n  \| h \|_{1}^{n-1} \| (1 + \xi) h \|_{1} \Big(\eta + \frac{2\pi}{\beta}\Big) \mathfrak{c}^{n} n!.
\end{equation} 
Consider now the error term $R^{(n)}_{2,2}(t)$ in (\ref{eq:RRR}). We have, using that $| g(\eta t) | \leq \| h \|_{1}$ and $| g_{\beta, \eta}(t) | \leq \| h \|_{1}$, together with (\ref{eq:PPPest}):
\begin{equation}
| R^{(n)}_{2,2}(t) | \leq n 2^{n-1} \| h \|_{1}^{n-1} \big| g_{\beta, \eta}(t) - g(\eta t) \big| \mathfrak{c}^{n} n!.
\end{equation}
Then, using (\ref{eq:ggdiff}) we find:
\begin{equation}\label{eq:R22}
| R^{(n)}_{2,2}(t) | \leq n 2^{n-1} \| h \|_{1}^{n-1} \frac{2\pi}{e \beta \eta} \Big\| \frac{h}{\xi}  \Big\|_{1} \mathfrak{c}^{n} n!.
\end{equation}
Coming back to (\ref{eq:diffrho}), we have, for $|\varepsilon| < \varepsilon_{0}$, with $\varepsilon_{0}$ small enough only dependent on $h$ and on $\mathfrak{c}$:
\begin{equation}\label{eq:estRRR}
\begin{split}
\Big|\Tr \mathcal{O}_{X} \tilde \rho(t) - \langle \mathcal{O}_{X} \rangle_{t}\Big| &\leq \sum_{n=1}^{\infty} \frac{|\varepsilon|^{n}}{n!}\Big( | R^{(n)}_{2,1}(t) | + | R^{(n)}_{2,2}(t) | \Big) \\
&\leq C_{1} |\varepsilon| \Big( \eta + \frac{1}{\beta} \Big) + \frac{ C_{2}|\varepsilon| }{\beta \eta},
\end{split}
\end{equation}
where the constants $C_{1} \equiv C_{1}(\frak{c}, h)$ and $C_{2} \equiv C_{2}(\frak{c}, h)$ can be obtained from (\ref{eq:R21}), (\ref{eq:R22}), respectively. In conclusion, combining the bound (\ref{eq:estRRR}) with (\ref{eq:first}):
\begin{equation}
\begin{split}
\Big| \Tr \mathcal{O}_{X} \rho(t) - \langle \mathcal{O}_{X} \rangle_{t} \Big| &\leq \Big| \Tr \mathcal{O}_{X} \rho(t) - \Tr \mathcal{O}_{X} \tilde \rho(t) \Big| + \Big|\Tr \mathcal{O}_{X} \tilde \rho(t) - \langle \mathcal{O}_{X} \rangle_{t}\Big| \\
&\leq \frac{K|\varepsilon|}{\eta^{d+2}\beta} + C_{1} |\varepsilon| \Big( \eta + \frac{1}{\beta} \Big) + \frac{C_{2} |\varepsilon|}{\beta \eta}.
\end{split}
\end{equation}
This proves (\ref{eq:adia}) and concludes the proof of Theorem \ref{thm:main}. 
\end{proof}

\begin{proof}[Proof of Corollary \ref{rem:improv}] Let us show how the strategy used above can be adapted to obtain the improved result (\ref{eq:etam}). Under the assumptions of Corollary~\ref{rem:improv}, the function $g(z)$ is $m+1$ times continuously differentiable for $\text{Re}z\leq 0$, and the same holds for $g_{\beta, \eta}(z)$. We proceed as in the proof of Theorem \ref{thm:main}, the only difference being the estimate for the term $R^{(n)}_{2,1}(0)$ in Eq. (\ref{eq:RRR}). We have:
\begin{equation}\label{eq:R21impro}
\begin{split}
&R^{(n)}_{2,1}(0) \\
&= \sum_{\underline{\omega} \in \frac{2\pi}{\beta} \mathbb{N}^{n}} \Big[ \prod_{j=1}^{n} \tilde g_{\beta, \eta}(\omega_{j}) \Big] \Big(\langle {\bf T} \widehat{\mathcal{P}}_{\omega_{1}}; \widehat{\mathcal{P}}_{\omega_{2}}; \cdots; \widehat{\mathcal{P}}_{\omega_{n}}; \mathcal{O}_{X} \rangle - \langle {\bf T} \widehat{\mathcal{P}}_{0}; \widehat{\mathcal{P}}_{0}; \cdots; \widehat{\mathcal{P}}_{0}; \mathcal{O}_{X} \rangle\Big) \\
&= \int_{[0,\beta)^{n}} d\underline{s}\, \Big(\prod_{j=1}^{n} g_{\beta,\eta}(-is_{j}) - g_{\beta,\eta}(0)^{n}\Big)\langle {\bf T} \gamma_{s_{1}}(\mathcal{P}); \gamma_{s_{2}}(\mathcal{P}); \cdots; \gamma_{s_{n}}(\mathcal{P}); \mathcal{O}_{X} \rangle\;.
\end{split}
\end{equation}
By differentiability of $g_{\beta,\eta}(-is)$, we have the Taylor expansion, for $s\in [0,\beta)$:
\begin{equation}\label{eq:tay}
	g_{\beta,\eta}(-is) - g_{\beta,\eta}(0) = \sum_{j=1}^{m} \frac{\partial_s^j g_{\beta,\eta}(0)}{j!}(-is)^j + r_{\beta,\eta}^{(m+1)}(s)\;,
\end{equation}
and the remainder can be estimated in a similar way to  (\ref{eq:omegag}), so that there is some $L_{m+1}>0$ such that
\begin{equation}\label{eq:rest}
	|r_{\beta,\eta}^{(m+1)}(s)| \leq L_{m+1} \Big( \eta  + \frac{1}{\beta}\Big)^{m+1} |s|_{\beta}^{m+1}\;.
\end{equation}
Concerning the first term in the right-hand side of (\ref{eq:tay}), we use that, recalling that by assumption $\partial^{j}_{s}g(0) = 0$ for all $j\leq m$,
\begin{equation}\label{eq:partdiff}
\begin{split}
\Big| \partial_{s}^{j} g_{\beta,\eta}(0)  \Big| &= \Big|\partial_{s}^{j} g_{\beta,\eta}(0) - \eta^{j} \partial_{s}^{j} g(0)\Big| \\
&= \Big| \sum_{r=0}^{\infty} \int_{\frac{2\pi}{\beta \eta} r}^{\frac{2\pi}{\beta \eta}(r+1)} d\xi\, h(\xi) \Big[ \Big(\frac{2\pi}{\beta} (r+1)\Big)^{j} - (\xi \eta)^{j}  \Big]  \Big| \\
&\leq \sum_{r=0}^{\infty} \int_{\frac{2\pi}{\beta \eta} r}^{\frac{2\pi}{\beta \eta}(r+1)} d\xi\, | h(\xi) | \Big| \Big(\xi \eta + \frac{2\pi}{\beta}\Big)^{j} - (\xi \eta)^{j}  \Big| \\
&\leq \widetilde{C}_{j} \sum_{\ell = 1}^{j} \frac{\eta^{j-\ell}}{\beta^{\ell}}
\end{split}
\end{equation}
where we used the assumption (\ref{eq:hbd_enhanced}). Therefore, from (\ref{eq:tay}), (\ref{eq:rest}), (\ref{eq:partdiff}):
\begin{equation}
\big| g_{\beta,\eta}(-is) - g_{\beta,\eta}(0) \big| \leq C_{m+1} \Big[ \left( \eta+ \frac1\beta \right)^{m+1} |s|_{\beta}^{m+1} + \frac{1}{\beta}(1 + |s|_{\beta}^{m+1})\Big]
\end{equation}
which implies
\begin{equation}
\begin{split}
	&\Big| \prod_{j=1}^{n} g_{\beta,\eta}(-is_{j}) - g_{\beta,\eta}(0)^{n} \Big| \\&\qquad \leq 2^{n-1} \| h\|_{1}^{n-1} C_{m+1} \Big[ \left( \eta+\frac1\beta \right)^{m+1} \sum_{j=1}^{n}  |s_{j}|_{\beta}^{m+1} + \frac{1}{\beta}\sum_{j=1}^{n}  (1+|s_{j}|_{\beta}^{m+1})\Big].
\end{split}
\end{equation}
Plugging this bound in (\ref{eq:R21impro}) we get
\begin{equation}
\begin{split}
| R^{(n)}_{2,1}(0) | &\leq 2^{n-1} \| h\|_{1}^{n-1} C_{m+1} \Big[ \left( \eta+\frac1\beta \right)^{m+1} + \frac{1}{\beta}\Big] \\&\quad\cdot \sum_{j=1}^{n}  \int_{[0,\beta)^{n}} d\underline{s}\,  (1+|s_{j}|_{\beta}^{m+1}) \big| \langle {\bf T} \gamma_{s_{1}}(\mathcal{P});  \cdots; \gamma_{s_{n}}(\mathcal{P}); \mathcal{O}_{X} \rangle
\end{split}
\end{equation}
and so from the assumption (\ref{eq:assAm}), we obtain, for a new constant $\widetilde{C}_{m+1}$:
\begin{equation}\label{eq:R210}
| R^{(n)}_{2,1}(0) | \leq \widetilde{C}_{m+1} C^{n} \mathfrak{c}^{n}  \Big[ \Big( \eta+ \frac1\beta \Big)^{m+1} + \frac{1}{\beta}\Big]  n!.
\end{equation}
The final claim, Eq.~(\ref{eq:etam}), follows proceeding as in the proof of Theorem \ref{thm:main}, replacing the bound (\ref{eq:R21}) with (\ref{eq:R210}).
\end{proof}
To conclude the section, we discuss the proof of Corollary \ref{cor:lin}.
\begin{proof}[Proof of Corollary \ref{cor:lin}] Eqs.~(\ref{eq:lin}), (\ref{eq:linerr}) follow from Eqs. (\ref{eq:first}), (\ref{eq:wickproof}), and from the convergence of the series in (\ref{eq:PP}). Eq.~(\ref{eq:linerr2}) is proved following the argument after (\ref{eq:diffrho}). To prove Eq.~(\ref{eq:last}), we use that, from (\ref{eq:wickproof}):
\begin{equation}
\int_{0}^{\beta} ds\, g_{\beta,\eta}(t - is) \langle \gamma_{s}(\mathcal{P})\;; \mathcal{O}_{X} \rangle_{\beta, \mu, L} = i\int_{-\infty}^{t} d s\, g_{\beta,\eta}(s) \langle \left[ \tau_{t}(\mathcal{O}_{X}), \tau_{s}(\mathcal{P}) \right] \rangle_{\beta, \mu, L}.
\label{eq:corlin_start}
\end{equation}
Next, we estimate the error introduced by replacing $g_{\beta,\eta}(s)$ with $g(\eta s)$. We have:
\begin{equation}
\begin{split}
&\Big| \int_{-\infty}^{t} d s\, (g_{\beta,\eta}(s) - g(\eta s)) \langle \left[ \tau_{t}(\mathcal{O}_{X}), \tau_{s}(\mathcal{P}) \right] \rangle_{\beta, \mu, L}\Big| \\
&\qquad \leq \int_{-\infty}^{t} d s\, |g_{\beta,\eta}(s) - g(\eta s)| \big| \langle \left[  \tau_{t}(\mathcal{O}_{X}), \tau_{s}(\mathcal{P}) \right] \rangle_{\beta, \mu, L}\big| \\
&\qquad \leq \widetilde K\int_{-\infty}^{t} d s\, |g_{\beta,\eta}(s) - g(\eta s)| (|t-s|^{d} + 1),
\end{split}
\end{equation}
where in the last step we used the Lieb-Robinson bound, as in (\ref{eq:LRcons}). Finally, proceeding as in (\ref{eq:norm2})-(\ref{eq:norm3}), we get:
\begin{equation}
\begin{split}
\int_{-\infty}^{t} d s\, |g_{\beta,\eta}(s) - g(\eta s)| &(|t-s|^{d} + 1) \\& \leq \frac{\widetilde{K}}{\beta} \int_{-\infty}^{0} ds \int_{0}^{\infty} d\xi\, |h(\xi)| e^{\xi \eta s} (|s|^{d+1} + 1) \\
&\leq \frac{K}{\beta \eta^{d+2}}\;.
\end{split}
\label{eq:corlin_end}
\end{equation}
This concludes the proof of (\ref{eq:last}) and of the corollary.
\end{proof}

\appendix
\section{On the switch functions}\label{app:switch}

Here we will discuss functions $g(t)$ that satisfy Assumption \ref{ass:switch}. This assumption holds true for the standard switch function $g(t) = e^{a t}$ with $a>0$, where $h(\xi) = \delta(\xi - a)$, and more generally for the finite linear combinations of such functions. More generally, it is a natural question to understand under which conditions a function $g(t)$ can be represented as in (\ref{eq:switch}), for a function $h$ that satisfies the desired properties. Here we will give sufficient conditions for this to hold.

Let $\delta > 0$ and let $g(z)$ be analytic for $\text{Re}\, z <\delta$. Suppose that, for all $0\leq k \leq d+2$ and for all $x< \delta$:
\begin{equation}\label{eq:unif}
\int_{-\infty}^{\infty} dy\, |x + i y|^{k} |g(x + i y)| \leq C,\quad \lim_{x\to -\infty} \int_{-\infty}^{\infty} dy\, |x + i y|^{k} |g(x + i y)| = 0.
\end{equation}
Furthermore, suppose that for all $0\leq k \leq d+2$ and for all $y\in \mathbb{R}$:
\begin{equation}\label{eq:unif2}
\int_{-\infty}^{0} dx\, |x + i y|^{k} |g(x + iy)| \leq C,\quad \lim_{y\to \pm \infty} \int_{-\infty}^{0} dx\, |x + i y|^{k} |g(x + iy)| = 0.
\end{equation}
Examples of functions satisfying these assumptions are:
\begin{equation}
g(z) = \frac{1}{(z - a)^{n}}
\end{equation}
with $n \geq d+4$ and $a>\delta$. Let us check that functions satisfying (\ref{eq:unif}), (\ref{eq:unif2}) verify Assumption \ref{ass:switch}. Let $\gamma$ be the straight complex path on the imaginary axis $\gamma: \delta/2 + i\infty \to \delta/2 -i\infty$.  Define, for $\xi \geq 0$:
\begin{equation}\label{eq:inverseL}
h(\xi) := \frac{1}{2\pi i} \int_{\gamma} dz\, e^{-z\xi} g(z).
\end{equation}
We have:
\begin{equation}\label{eq:hL1}
\begin{split}
|h(\xi)| &\leq \frac{e^{-(\delta/2) \xi}}{2\pi} \int_{-\infty}^{\infty}dy\, |g(\delta/2 + iy)| \\
&\leq Ce^{-(\delta/2) \xi}
\end{split}
\end{equation}
where we used Eq.~(\ref{eq:unif}). Let us check that $h$ satisfies Eq.~(\ref{eq:switch}). Let $t\leq 0$ and define:
\begin{equation}
\tilde g(t) := \int_{0}^{\infty} d\xi\, e^{\xi t} h(\xi),
\end{equation}
which is well defined thanks to (\ref{eq:hL1}). We have:
\begin{equation}
\begin{split}
\tilde g(t) &= \frac{1}{2\pi i} \int_{\gamma} dz\,  \int_{0}^{\infty} d\xi\, e^{\xi t} e^{-z\xi} g(z) \\
&= \frac{-1}{2\pi i} \int_{\gamma} dz\, \frac{g(z)}{t - z} \\
&= g(t).
\end{split}
\end{equation}
In the first identity we applied Fubini's theorem, and in the last identity we applied Cauchy integral formula, using that $g(z)$ is analytic for $\text{Re}\, z \leq \delta/2$ together with the assumptions (\ref{eq:unif}), (\ref{eq:unif2}). Thus, Eq.~(\ref{eq:inverseL}) is the inverse Laplace transform of the function $g$.

Eq.~(\ref{eq:hL1}) implies that $h\in L^{1}(\mathbb{R}_{+})$. More generally, Eq.~(\ref{eq:hL1}) implies, for all $k\geq 0$:
\begin{equation}
\| (1+\xi^{k}) h \|_{1} \leq C_{k},
\end{equation}
which shows that the second assumption in Eq.~(\ref{eq:hbd}) holds true. Let us now consider the first assumption in (\ref{eq:hbd}). We observe that
\begin{equation}\label{eq:h0}
h(0) = \frac{1}{2\pi i} \int_{\gamma} dz\, g(z) = 0,
\end{equation}
since the function $g(z)$ has no poles for $\text{Re}\, z < \delta$: Eq.~(\ref{eq:h0}) follows from Cauchy integral formula, combined with the integrability properties (\ref{eq:unif}), (\ref{eq:unif2}). Furthermore, for all $1\leq k\leq d+2$:
\begin{equation}
\partial^{k}_{\xi} h(\xi) = \frac{(-1)^{k}}{2\pi i} \int_{\gamma} dz\, e^{-z\xi} z^{k} g(z)
\end{equation}
where we used that $z^{k} g(z)$ is integrable on the path $\gamma$, by (\ref{eq:unif}). By the same assumption, $\partial^{k}_{\xi} h(\xi)$ is bounded uniformly in $\xi$ for $1\leq k\leq d+2$.

The function $z^{k} g(z)$ has the same analyticity properties as $g$, in particular it has no poles for $\text{Re}\, z < \delta$. Thus, by Cauchy integral formula, thanks to the integrability assumptions (\ref{eq:unif}), (\ref{eq:unif2}) we have:
\begin{equation}\label{eq:hk0}
\begin{split}
\partial^{k}_{\xi} h(0) &= \frac{(-1)^{k}}{2\pi i} \int_{\gamma} dz\,  z^{k} g(z) \\
&=0.
\end{split}
\end{equation}
Eqs. (\ref{eq:h0}), (\ref{eq:hk0}) imply that $h(\xi) / \xi^{d+2}$ is bounded as $\xi \to 0$, in particular the first assumption in (\ref{eq:hbd}) is satisfied. This concludes the check of Assumption \ref{ass:switch}.

\section{Properties of Euclidean correlations}\label{app:int}

In this appendix we shall give the proofs of the properties of Euclidean correlation functions that have been used in the proof of Theorem \ref{thm:main}. These properties are well-known, and we collect the proofs here for completeness.

Recall the notation for the $n$-dimensional simplex of side $\beta$:
\begin{equation}\label{eq:symplex}
\Delta^n_{\beta}:=\{ (s_1,\dots,s_n)\in \mathbb{R}^n\,:\, \beta>s_1>\dots>s_{n}>0  \}.
\end{equation}
\begin{proposition}\label{lem:splitting}
Let $n,m\in \mathbb{N}$, $1\leq m\leq n-1$. Let $f: [0,\beta]^{m}\to \mathbb{C}$, $g: [0,\beta]^{n-m} \to \mathbb{C}$, $f, g$ integrable. The following identity holds:
\begin{equation}\label{eq:B2}
\begin{split}
&\sum_{\substack{J \subset \{1,\dots, n\} \\ |J| = m}}\int_{0}^{\beta} d s_1\dots\int_0^{s_{n-1}} d s_n\, f(\underline{s}_{J})g(\underline{s}_{J^c}) \\
&=\int_{0}^{\beta} d s_1\dots\int_0^{s_{m-1}} d s_m\, f(\underline{s}_{\{1,\dots,m\}})\int_{0}^{\beta} d s_{m+1}\dots\int_0^{s_{n-1}} d s_n\, g(\underline{s}_{\{m+1,\dots,n\}}),
\end{split}
\end{equation}
where the sum is over ordered $m$-tuples $J = (j_{1}, \ldots, j_{m})$, we used the notation $\underline{s}_{J} = (s_{j_{1}}, \ldots, s_{j_{m}})$ and $\underline{s}_{J^{c}} = \underline{s}_{\{1,\ldots, n\} \setminus J}$.
\end{proposition}
\begin{proof} Let $\mathbbm{1}_{\Delta^n_\beta}(\underline{s})$ be the characteristic function of the simplex $\Delta^n_\beta$, Eq.~(\ref{eq:symplex}). We start by writing:
\begin{equation}
\label{eqn:splitting}
\begin{aligned}
&\sum_{\substack{J \subset \{1,\dots, n\}\\ |J|=m}}\int_{0}^{\beta} d s_1\dots\int_0^{s_{n-1}} d s_n\, f(\underline{s}_{J})g(\underline{s}_{J^{c}})\\
& \quad = \sum_{\substack{J \subset \{1,\dots, n\}\\ |J|=m}}\int_{{[0,\beta]}^n} d s_1\cdots d s_n\, f(\underline{s}_{J})g(\underline{s}_{J^{c}})\mathbbm{1}_{\Delta^n_\beta}(\underline{s})\\
& \quad = \sum_{\substack{J \subset \{1,\dots, n\}\\ |J|=m}}\int_{{[0,\beta]}^n}d r_1\cdots d r_n\, f(\underline{r}_{\{1,\dots,m\}})g(\underline{r}_{\{m+1,\dots,n\}})\mathbbm{1}_{\Delta^n_\beta}( \pi_{J}(\underline{r})),
\end{aligned}
\end{equation}
where in the last step we used the change of variables:
\begin{equation}
\underline{r}_{\{1,\dots,m\}}=\underline{s}_J,\qquad \underline{r}_{\{m+1,\dots,n\}}=\underline{s}_{J^{c}}
\end{equation}
and we call $\pi_{J}$ the permutation such that $\pi_{J}(\underline{r})_{i} = s_{i}$ for $i = 1, \ldots, n$. Now, define:
\begin{equation}
\label{eqn:L}
L(\underline{r}) := \sum_{\substack{J \subset \{1,\dots, n\}\\ |J| = m}} \mathbbm{1}_{\Delta^n_\beta}( \pi_{J}(\underline{r})).
\end{equation}
This function only takes values $0$ or $1$, since there can be at most one non-vanishing element in the sum. We claim that:
\begin{equation}\label{eq:Lclaim}
L(\underline{r}) = \mathbbm{1}_{\Delta^m_\beta}(\underline{r}_{\{1,\dots, m\}})\mathbbm{1}_{\Delta^{n-m}_\beta}(\underline{r}_{\{m+1,\dots, n\}}).
\end{equation}
Suppose that $\underline{r}_{\{1,\dots, m\}}\in \Delta^m_\beta$ and $\underline{r}_{\{m+1,\dots, n\}}\in\Delta^{(n-m)}_\beta$. Then, there exists a unique permutation that maps $\underline{r} = (\underline{r}_{\{1,\dots, m\}}, \underline{r}_{\{m+1,\dots, n\}})$ into a decreasing sequence. Equivalently, there exists a unique $J$, $|J| = m$, such that $\pi_{J}(\underline{r})$ is in $\Delta^n_\beta$. Thus, 
\begin{equation}
\mathbbm{1}_{\Delta^m_\beta}(\underline{r}_{\{1,\dots, m\}})\mathbbm{1}_{\Delta^{(n-m)}_\beta}(\underline{r}_{\{m+1,\dots, n\}}) = 1 = \sum_{\substack{J \subset \{1,\dots, n\}\\ |J| = m}} \mathbbm{1}_{\Delta^n_\beta}( \pi_{J}(\underline{r})).
\end{equation}
Suppose now that $\underline{r}_{\{1,\dots, m\}} \notin \Delta^m_\beta$. For any $J\subset \{1, \ldots, n\}$, $|J| = m$, consider the sequence $\pi_{J}(\underline{r})$. The action of $\pi_{J}$ sends the first $m$ entries of $\underline{r}$ to $m$ entries in the new sequence, preserving their relative order. In particular, if $\underline{r}_{\{1,\dots, m\}} \notin \Delta^m_\beta$ then $\pi_{J}(\underline{r}) \notin \Delta^n_\beta$. Therefore,
\begin{equation}
\mathbbm{1}_{\Delta^m_\beta}(\underline{r}_{\{1,\dots, m\}})\mathbbm{1}_{\Delta^{(n-m)}_\beta}(\underline{r}_{\{m+1,\dots, n\}}) = 0 = \sum_{\substack{J \subset \{1,\dots, n\}\\ |J| = m}} \mathbbm{1}_{\Delta^n_\beta}( \pi_{J}(\underline{r})).
\end{equation}
A similar discussion applies to the case $\underline{r}_{\{m+1,\dots, n\}} \notin \Delta^{n-m}_\beta$. This concludes the proof of (\ref{eq:Lclaim}). The final claim (\ref{eq:B2}) follows after plugging (\ref{eq:Lclaim}) into (\ref{eqn:splitting}).
\end{proof}
\begin{proposition}\label{prop:period} Let $\mathcal{O}_{i} \in \mathcal{A}_{\Lambda_{L}}$, $i = 1,\ldots, n$. Let $s_{i} \neq s_{j}$ for $i\neq j$ and $0\leq s_{i} \leq \beta$. Consider:
\begin{equation}\label{eq:gdef}
\begin{split}
&G(s_{1}, \ldots, s_{n}) \\
&\quad := \sum_{\pi} \mathbbm{1}(s_{\pi(1)} > s_{\pi(2)} > \ldots > s_{\pi(n)}) \big\langle \gamma_{s_{\pi(1)}}(\mathcal{O}_{\pi(1)}) \cdots \gamma_{s_{\pi(n)}} (O_{\pi(n)})  \big\rangle_{\beta, \mu, L},
\end{split}
\end{equation}
where the sum is over permutations of $1, \ldots, n$. Then, for all $i=1,\ldots, n$:
\begin{equation}\label{eq:0beta}
G(s_{1},\ldots, s_{i-1}, \beta, s_{i+1}, \ldots, s_{n}) = G(s_{1},\ldots, s_{i-1}, 0, s_{i+1}, \ldots, s_{n}).
\end{equation}
Eq.~(\ref{eq:0beta}) allows to extend $G$ to a $\beta$-periodic function on $\mathbb{R}^{n}$, that we shall continue to denote by $G$. Furthermore, the periodic extension is such that, for all $\sigma \in \mathbb{R}$:
\begin{equation}\label{eq:pluss}
G(s_{1}, \ldots, s_{i}, \ldots, s_{n}) = G(s_{1} + \sigma, \ldots, s_{i} + \sigma, \ldots s_{n} + \sigma).
\end{equation}
\end{proposition}
\begin{remark}\label{rem:period}
\begin{itemize}

\item[(i)] In general, the periodic extension of $G$ to $\mathbb{R}^{n}$ might be discontinuous at equal times, unless the operators commute.

\item[(ii)] Notice that if all operators $\mathcal{O}_{i}$ are even in the fermionic creation and annihilation operators, Eq.~(\ref{eq:gdef}) agrees with the definition of time-ordered correlation function, Eqs. (\ref{eq:Tord1}), (\ref{eq:Tord2}). These facts are mentioned in Remark \ref{rem:T}. Instead, if the operators $\mathcal{O}_{i}$ are odd in the fermionic creation and annihilation operators, the natural definition of time-ordering is defined including the sign of the permutation in the sum (\ref{eq:gdef}), compare with Eq.~(\ref{eq:time}). This gives rise to a $\beta$-antiperiodic function on $\mathbb{R}^{n}$.
\end{itemize}
\end{remark}

\begin{proof} Consider the function $G$ in Eq.~(\ref{eq:gdef}) with $s_{i} = 0$. Since all times are distinct, the only permutations contributing to the sum are those with $\pi(n) = i$. Thus,
\begin{equation}
\begin{split}
&G(s_{1}, \ldots, s_{i-1}, 0, s_{i+1}, \ldots, s_{n}) =\\
& \sum_{\tilde\pi} \mathbbm{1}(s_{\tilde \pi(1)} > s_{\tilde \pi(2)} > \ldots > s_{\tilde \pi(n-1)}) \big\langle \gamma_{s_{\tilde \pi(1)}}(\mathcal{O}_{\tilde \pi(1)}) \cdots \gamma_{s_{\tilde \pi(n-1)}}(\mathcal{O}_{\tilde \pi(n-1)}) O_{i}\big\rangle
\end{split}
\end{equation}
where the sum is over permutations $\{\tilde\pi\}$ of $n-1$ elements. Now, by the KMS identity, Eq.~(\ref{eq:KMS}),
\begin{equation}\label{eq:0beta2}
\begin{split}
&G(s_{1}, \ldots, s_{i-1}, 0, s_{i+1}, \ldots, s_{n}) = \\
&= \sum_{\tilde\pi} \mathbbm{1}(s_{\tilde \pi(1)} >  \ldots > s_{\tilde \pi(n-1)}) \big\langle \gamma_{\beta}(O_{i})\gamma_{s_{\tilde \pi(1)}}(\mathcal{O}_{\tilde\pi(1)}) \cdots \gamma_{s_{\tilde \pi(n-1)}}(\mathcal{O}_{\tilde\pi(n-1)})\big\rangle\\
&= \sum_{\pi} \mathbbm{1}(s_{\pi(1)} >  \ldots > s_{\pi(n)}) \big\langle \gamma_{s_{\pi(1)}}(\mathcal{O}_{\pi(1)}) \cdots \gamma_{s_{\pi(i)}} (O_{\pi(i)}) \cdots \gamma_{s_{\pi(n)}} (O_{\pi(n})\big\rangle
\end{split}
\end{equation}
where in the last step we set $s_{i} = \beta$, and used the fact that all the permutations contributing to the last sum are such that $\pi(1) = i$. The right-hand side of Eq.~(\ref{eq:0beta2}) equals $G(s_{1}, \ldots, s_{i-1}, \beta, s_{i+1}, \ldots, s_{n})$, and this concludes the proof of (\ref{eq:0beta}).

Eq.~(\ref{eq:0beta}) allows to extend $G$ to a periodic function over $\mathbb{R}^{n}$. Let us now prove the time-translation invariance property, Eq.~(\ref{eq:pluss}). By construction, the function $G$ satisfies:
\begin{equation}
G(s_{1}, \ldots, s_{n}) = G([s_{1}]_{\beta}, \ldots, [s_{n}]_{\beta}),
\end{equation}
where $[s_{i}]_{\beta}$ is the representative of $s_{i}$ in $[0,\beta)$, that is $s_{i} = [s_{i}]_{\beta} + m \beta$ for some $m \in \mathbb{Z}$. Without loss of generality, suppose that $[s_{1}]_{\beta} > \ldots > [s_{n}]_{\beta}$; otherwise, relabel times so that this condition holds. Then, from Eq.~(\ref{eq:gdef}) it is easy to see that the function $G$ has the following dependence on times:
\begin{equation}
G(s_{1}, \ldots, s_{n}) = G([s_{1}]_{\beta} - [s_{n}]_{\beta}, \ldots, [s_{n-1}]_{\beta} - [s_{n}]_{\beta}, 0).
\end{equation}
The shift $s_{i} \to s_{i} + \sigma$ for all $i=1,\ldots, n$ changes $[s_{i}]_{\beta} - [s_{n}]_{\beta}$ into $[s_{i}]_{\beta} - [s_{n}]_{\beta}$ plus an integer multiple of $\beta$. By periodicity, this does not affect the value of the function $G$. Thus,
\begin{equation}
G(s_{1}, \ldots, s_{n}) = G(s_{1} + \sigma, \ldots, s_{n} + \sigma),
\end{equation}
which concludes the proof of the proposition.
\end{proof}

\section{Decay of correlations for interacting models}\label{app:mb}

In this appendix we shall discuss the validity of Assumption \ref{ass:dec} for fermionic lattice models. The arguments presented in this section are well-known, and we reproduce them for completeness. For definiteness, we shall consider the following class of many-body Hamiltonians:
\begin{equation}\label{eq:Ham}
\mathcal{H}^{\lambda} = \sum_{{\bf x},{\bf y} \in \Lambda_{L}} a^{*}_{{\bf x}} H({\bf x}; {\bf y}) a_{{\bf y}} + \lambda \sum_{{\bf x},{\bf y} \in \Lambda_{L}} a^{*}_{{\bf x}} a^{*}_{{\bf y}} v({\bf x};{\bf y}) a_{{\bf y}} a_{{\bf x}}, 
\end{equation}
where $H({\bf x}; {\bf y})$ and $v({\bf x}; {\bf y})$ are finite-range. The discussion that follows actually applies essentially unchanged to a larger class of Hamiltonians, obtained replacing the quartic interaction in (\ref{eq:Ham}) by finite-range interactions of arbitrary even degree in the fermionic creation and annihilation operators. Let $\langle \cdot \rangle_{\beta, \mu, L}^{\lambda}$ be the Gibbs state of the system:
\begin{equation}\label{eq:Glambda}
\langle \mathcal{O} \rangle_{\beta, \mu, L}^{\lambda} = \frac{\Tr\, e^{-\beta (\mathcal{H}^{\lambda} - \mu \mathcal{N})} \mathcal{O} }{\Tr e^{-\beta (\mathcal{H}^{\lambda} - \mu \mathcal{N})}}.
\end{equation}
For $|\lambda|$ small, the Gibbs state can be expanded in power series around the non-interacting one, following Section \ref{sec:compa}. The main advantage in doing this is that the non-interacting Gibbs state is quasi-free, which means that all Euclidean correlations can be computed starting from the Euclidean two-point function using the Wick rule. Let $t, t'\in [0,\beta)$, $t\neq t'$. Let us denote by $\langle \cdot \rangle_{\beta, \mu, L}^{0}$ the non-interacting Gibbs state $(\lambda = 0)$. We define the non-interacting two-point function as:
\begin{equation}\label{eq:g2}
g_{2}(t,{\bf x}; t', {\bf y}) := \langle {\bf T} \gamma_{t}(a_{{\bf x}}) \gamma_{t'} (a^{*}_{{\bf y}}) \rangle^{0}_{\beta, \mu, L}.
\end{equation}
At equal times, the two-point function is defined by normal ordering:
\begin{equation}\label{eq:g22}
g_{2}(t,{\bf x}; t, {\bf y}) = - \langle a^{*}_{{\bf y}} a_{{\bf x}} \rangle^{0}_{\beta, \mu, L}.
\end{equation}
Eq.~(\ref{eq:g2}) is extended to an antiperiodic function over all $t,t' \in \mathbb{R}$ with period $\beta$. The next proposition gives the explicit expression of the two-point function.
\begin{proposition}[Non-interacting two-point function]\label{prp:prop} Let $0\leq t, t' < \beta$. Then:
\begin{equation}\label{eq:2pt}
g_{2}(t,{\bf x}; t', {\bf y}) = \mathbbm{1}(t > t') \frac{e^{-(t-t') (H - \mu)}}{1 + e^{-\beta (H - \mu)}}({\bf x}; {\bf y}) - \mathbbm{1}(t \leq t') \frac{e^{-(t-t') (H - \mu)}}{1 + e^{\beta (H - \mu)}}({\bf x}; {\bf y}).
\end{equation}
\end{proposition}
\begin{proof}
	Eq.~(\ref{eq:2pt}) can be proved by direct computation of the trace involved in the definition of the non-interacting Gibbs state, representing the Fock space in the basis of Slater determinants associated with the eigenstates of the Hamiltonian $H$. The computation can be found in standard textbooks in condensed matter physics, see {\it e.g.} \cite{FW, NO}. 
\end{proof}
Next, we collect useful decay estimates for the two-point function.
\begin{proposition}[Bounds for the two-point function]\label{prop:2pt} There exist $C_{\beta}, c_{\beta} > 0$ such that the following bound holds true:
\begin{equation}\label{eq:bdg1}
|g_{2}(t,{\bf x}; t', {\bf y}) | \leq C_{\beta} e^{-c_{\beta} \| {\bf x} - {\bf y} \|_{L}}\qquad \text{for all $L>0$.}
\end{equation}
Moreover, suppose that $\mu \notin \sigma(H)$ and $\text{dist}(\mu, \sigma(H)) \geq \delta$ with $\delta > 0$ uniformly in $L$. Then, the constants $C_{\beta}, c_{\beta}$ can be chosen uniformly in $\beta$:
\begin{equation}\label{eq:bdg2}
|g_{2}(t,{\bf x}; t', {\bf y}) | \leq Ce^{-c(\|{\bf x} - {\bf y} \|_{L} + | t - t' |_{\beta})}.
\end{equation}
\end{proposition}
\begin{proof} Let us start by proving (\ref{eq:bdg1}). Suppose that $t>t'$, the other case can be studied in the same way. The starting point is the following formula for the two-point function:
\begin{equation}\label{eq:complexg}
\begin{split}
g_{2}(t,{\bf x}; t', {\bf y}) &= \frac{e^{-(t-t') (H - \mu)}}{1 + e^{-\beta (H - \mu)}}({\bf x}; {\bf y}) \\
&= \frac{1}{2\pi i} \int_{\mathcal{C}} dz\, \frac{e^{-(z - \mu)(t - t')}}{1 + e^{-\beta (z - \mu)}} \frac{1}{z - H}({\bf x}; {\bf y}),
\end{split}
\end{equation}
where the first identity follows from Proposition \ref{prp:prop}, and in the second equality the complex path $\mathcal{C}$ is a rectangle that encircles the spectrum of $H$, and that crosses the imaginary axis at $\text{Im}\, z = \pm \pi/(2\beta)$. Thus, the path does not enclose any of the poles of the Fermi-Dirac function, which are given by $z - \mu = i \frac{2\pi}{\beta} (n + \frac{1}{2})$ with $n\in \mathbb{Z}$, and it stays away from the spectrum of $H$. By construction, the only singularities encircled by the path $\mathcal{C}$ correspond to the eigenvalues of $H$. Since $H$ is bounded uniformly in $L$, the length of the complex path is also bounded uniformly in $L$; in particular, we can choose $\mathcal{C}$ such that for all $z\in \mathcal{C}$ we have $\text{dist}(z, \sigma(H)) \geq \pi/(2\beta)$. The estimate (\ref{eq:bdg1}) easily follows from the Combes-Thomas estimate for the Green's function, see {\it e.g.} \cite{AW}.

Let us now suppose that $\text{dist}(\mu, \sigma(H)) > \delta$, and let us prove the estimate (\ref{eq:bdg2}). We can use a complex representation (\ref{eq:complexg}), where now the path $\mathcal{C}$ splits into the disjoint union of two non-intersecting paths, $\mathcal{C}_{-}$ and $\mathcal{C}_{+}$, that encircle the spectrum of $H$ on the left or on the right of $\mu$, respectively. The paths can be chosen so that the distance from any point $z\in \mathcal{C}$ to the poles is bounded below by a constant proportional to the spectral gap, uniformly in $\beta$. We write:
\begin{equation}
\begin{split}
g_{2}(t,{\bf x}; t', {\bf y}) &= \frac{1}{2\pi i} \int_{\mathcal{C}_{-}} dz\, \frac{e^{-(z - \mu)(t - t')}}{1 + e^{-\beta (z - \mu)}} \frac{1}{z - H}({\bf x}; {\bf y})\\&\quad + \frac{1}{2\pi i} \int_{\mathcal{C}_{+}} dz\, \frac{e^{-(z - \mu)(t - t')}}{1 + e^{-\beta (z - \mu)}} \frac{1}{z - H}({\bf x}; {\bf y}).
\end{split}
\end{equation}
For $z\in \mathcal{C}_{+}$:
\begin{equation}
\begin{split}
\Big|\frac{e^{-(z - \mu)(t - t')}}{1 + e^{-\beta (z - \mu)}}\Big| &\leq C e^{-\text{Re}\, (z - \mu)(t-t')} \\
&\leq Ce^{-c(t-t')}.
\end{split}
\end{equation}
Instead, for $z\in \mathcal{C}_{-}$:
\begin{equation}
\begin{split}
\Big|\frac{e^{-(z - \mu)(t - t')}}{1 + e^{-\beta (z - \mu)}}\Big| &\leq C e^{|\text{Re}\, (z - \mu)|(t-t' - \beta)} \\
&\leq Ce^{-c(\beta - (t - t'))}.
\end{split}
\end{equation}
All together, for $z\in \mathcal{C}$:
\begin{equation}
\Big|\frac{e^{-(z - \mu)(t - t')}}{1 + e^{-\beta (z - \mu)}}\Big| \leq Ce^{-c|t-t'|_{\beta}}.
\end{equation}
The exponential decay in space follows from a Combes-Thomas estimate for the Green's function in the complex integral, using that now the distance from $z$ to the spectrum of $H$ is bounded uniformly in $\beta$. This concludes the proof.
\end{proof}
\subsection{Check of Assumption \ref{ass:dec} for non-interacting fermions} 
As a warm up, let us check Assumption \ref{ass:dec} for non-interacting models, whose Hamiltonian is given by Eq.~(\ref{eq:Ham}) with $\lambda = 0$. In this case, the bounds of Proposition \ref{prop:2pt}, combined with Wick's rule, are enough to show that Assumption \ref{ass:dec} is satisfied for observables $\mathcal{O}^{(i)}_{X_{i}}$ that are quadratic in the fermionic operators. This allows, in particular, to fulfill the assumptions of Theorem \ref{thm:main} for non-interacting fermions, in the presence of a quadratic, time-dependent perturbation $g(\eta t) \varepsilon \mathcal{P}$. To see this, we write:
\begin{equation}
\mathcal{O}^{(i)}_{X_{i}} = \sum_{{\bf x}, {\bf y} \in X_{i}} O_{i}({\bf x}; {\bf y}) a^{*}_{{\bf x}} a_{{\bf y}}
\end{equation}
for a finite-range kernel $O_{i}({\bf x}; {\bf y})$ such that $|O_{i}({\bf x}; {\bf y})| \leq C$. Consider:
\begin{equation}
\begin{split}
&\big\langle {\bf T} \gamma_{t_{1}}(\mathcal{O}^{(1)}_{X_1});  \cdots; \gamma_{t_{n}}(\mathcal{O}^{(n)}_{X_n}); \mathcal{O}^{(n+1)}_{X_{n+1}} \big\rangle^{0}_{\beta, \mu, L}  \\
&\qquad = \sum_{{\bf x}_{i}, {\bf y}_{i} \in X_{i}} O_{1}({\bf x}_{1}; {\bf y}_{1})\cdots O_{n+1}({\bf x}_{n+1}; {\bf y}_{n+1}) \\&\qquad\qquad \cdot\big\langle {\bf T} \gamma_{t_{1}}(a^{*}_{{\bf x}_{1}} a_{{\bf y}_{1}});  \cdots; \gamma_{t_{n}}(a^{*}_{{\bf x}_{n}} a_{{\bf y}_{n}});  a^{*}_{{\bf x}_{n+1}} a_{{\bf y}_{n+1}} \big\rangle^{0}_{\beta, \mu, L}.
\end{split}
\end{equation}
The cumulant on the right-hand side can be evaluated using the fermionic Wick's rule, in terms of ``ring diagrams''. We have:
\begin{equation}\label{eq:ring}
\begin{split}
&\big\langle {\bf T} \gamma_{t_{1}}(a^{*}_{{\bf x}_{1}} a_{{\bf y}_{1}});  \cdots; \gamma_{t_{n}}(a^{*}_{{\bf x}_{n}} a_{{\bf y}_{n}});  a^{*}_{{\bf x}_{n+1}} a_{{\bf y}_{n+1}} \big\rangle^{0}_{\beta, \mu, L} \\
&\qquad = \sum_{\pi} g_{2}(\pi(1), \pi(2)) g_{2}(\pi(2), \pi(3))  \cdots g_{2}(\pi(n+1), \pi(1)) 
\end{split}
\end{equation}
where the sum is over permutations of $\{1, \ldots, n+1\}$ such that $\pi(1) = 1$ and where we used the short-hand notation:
\begin{equation}
g_{2}(\pi(i), \pi(j)) :=  g_{2}(t_{\pi(i)},{\bf x}_{\pi(i)}; t_{\pi(j)}, {\bf y}_{\pi(j)})
\end{equation}
with the understanding that $t_{n+1}= 0$. Proposition \ref{prop:2pt} can be used to control the space-time decay of the right-hand side of (\ref{eq:ring}). Let us first control the sums over the lattice sites. For every entry in the sum over permutations, we are led to consider:
\begin{equation}\label{eq:sum}
\sum_{{\bf x}_{i}, {\bf y}_{i} \in X_{i}} e^{-c \| {\bf x}_{\pi(1)} - {\bf y}_{\pi(2)} \|_{L}} e^{-c \| {\bf x}_{\pi(2)} - {\bf y}_{\pi(3)} \|_{L}} \cdots e^{-c \| {\bf x}_{\pi(n+1)} - {\bf y}_{\pi(1)} \|_{L}},
\end{equation}
where the constant $c$ might depend on $\beta$, if we do not have a spectral gap for $H$. Since the sets $X_{i}$ have bounded radius uniformly in $L$, we can estimate the sum (\ref{eq:sum}) by:
\begin{equation}
C^{n} \sum_{{\bf z}_{i} \in X_{i}} e^{-c \| {\bf z}_{\pi(1)} - {\bf z}_{\pi(2)} \|_{L}} e^{-c \| {\bf z}_{\pi(2)} - {\bf z}_{\pi(3)} \|_{L}} \cdots e^{-c \| {\bf z}_{\pi(n+1)} - {\bf z}_{\pi(1)} \|_{L}},
\end{equation}
where the sum runs over ${\bf z}_{i} \in X_{i}$ for $i = 1,\ldots, n+1$. Next, we estimate the sum as:
\begin{equation}
\begin{split}
&\sum_{\substack{\ell_{1}, \ldots, \ell_{n+1} \\ \ell_{i} \geq \text{dist}(X_{\pi(i+1)},X_{\pi(i)}) \\ \ell_{n+1} \geq \text{dist}(X_{\pi(n+1)},X_{\pi(1)})}} e^{-c\ell_{1} - c\ell_{2} - \ldots - c\ell_{n+1}} \\&\qquad \cdot \sum_{{\bf z}_{i} \in X_{i}} \mathbbm{1}( \| {\bf z}_{\pi(1)} - {\bf z}_{\pi(2)} \|_{L} = \ell_{1}, \ldots, \| {\bf z}_{\pi(n+1)} - {\bf z}_{\pi(1)} \|_{L} = \ell_{n+1} ) \\
&\leq C^{n+1} \sum_{\substack{\ell_{1}, \ldots, \ell_{n+1} \\ \ell_{i} \geq \text{dist}(X_{\pi(i+1)},X_{\pi(i)}) \\ \ell_{n+1} \geq \text{dist}(X_{\pi(n+1)},X_{\pi(1)})}} e^{-c\ell_{1} - c\ell_{2} - \ldots - c\ell_{n+1}} \\ 
&\leq C^{n+1} e^{-c d(\pi(1), \pi(2)) - \ldots- cd(\pi(n+1), \pi(1))}
\end{split}
\end{equation}
where $d(\pi(i), \pi(j)) := \text{dist}(X_{\pi(i)}, X_{\pi(j)})$. Therefore, 
\begin{equation}
\begin{split}
&\big|\big\langle {\bf T} \gamma_{t_{1}}(\mathcal{O}^{(1)}_{X_1});  \cdots; \gamma_{t_{n}}(\mathcal{O}^{(n)}_{X_n}); \mathcal{O}^{(n+1)}_{X_{n+1}} \big\rangle^{0}_{\beta, \mu, L}\big| \\
&\qquad \leq C^{n+1} \sum_{\pi} e^{-c d(\pi(1), \pi(2)) - \ldots- cd(\pi(n+1), \pi(1))} F_{\pi}(\underline{t})
\end{split}
\end{equation}
where $F_{\pi}(\underline{t}) = 1$ if no spectral gap is present, or otherwise:
\begin{equation}\label{eq:Fpi}
F_{\pi}(\underline{t}) = e^{-c | t_{\pi(1)} - t_{\pi(2)} |_{\beta} - \ldots - c| t_{\pi(n+1)} - t_{\pi(1)} |_{\beta}}.
\end{equation}
Let us now prove Eq.~(\ref{eq:assA}). We have, for some $R>0$:
\begin{equation}\label{eq:Fpi2}
\begin{split}
&\int_{[0,\beta]^{n}} d\underline{t}\, (1+|\underline{t}|_{\beta}) \sum_{X_{i} \subseteq \Lambda_{L}}  \big| \big\langle {\bf T} \gamma_{t_{1}}(\mathcal{O}^{(1)}_{X_1});  \cdots; \gamma_{t_{n}}(\mathcal{O}^{(n)}_{X_n}); \mathcal{O}^{(n+1)}_{X} \big\rangle^{0}_{\beta, \mu, L} \big| \\
&\leq C^{n+1} \int_{[0,\beta]^{n}} d\underline{t}\, (1+|\underline{t}|_{\beta}) \\ 
&\qquad \cdot \sum_{\substack{X_{i} \subseteq \Lambda_{L} \\ \text{diam}(X_{i}) \leq R}}  \sum_{\pi} e^{-c d(\pi(1), \pi(2)) - \ldots- cd(\pi(n+1), \pi(1))} F_{\pi}(\underline{t}).
\end{split}
\end{equation}
Since the number of sets $X_{i}$ with bounded diameter and containing a given point ${\bf z}_{i}$ is bounded, we can estimate the $X_{i}$ sum as:
\begin{equation}\label{eq:sumX}
\begin{split}
&K^{n+1} \sum_{{\bf z}_{1}, \ldots, {\bf z}_{n}} e^{-c \| {\bf z}_{\pi(1)} - {\bf z}_{\pi(2)} \|_{L}} e^{-c \| {\bf z}_{\pi(2)} - {\bf z}_{\pi(3)} \|_{L}} \cdots e^{-c \| {\bf z}_{\pi(n+1)} - {\bf z}_{\pi(1)} \|_{L}} \\
& = K^{n+1} \sum_{\ell_{1}, \ldots, \ell_{n+1}} e^{-c \ell_{1} - c\ell_{2} - \ldots - c\ell_{n+1}} \\ &\qquad\cdot\sum_{{\bf z}_{1}, \ldots, {\bf z}_{n}} \mathbbm{1}(\| {\bf z}_{\pi(1)} - {\bf z}_{\pi(2)} \|_{L} = \ell_{1}, \ldots, \| {\bf z}_{\pi(n+1)} - {\bf z}_{\pi(1)} \|_{L} = \ell_{n+1}),
\end{split}
\end{equation}
where ${\bf z}_{n+1} = {\bf x} \in X$.  A crude bound for the last sum is $C^{n} \ell_{1}^{d} \cdots \ell_{n}^{d}$, uniformly in the volume of the system (remember that the point ${\bf z}_{n+1}$ is fixed). Plugging this estimate in (\ref{eq:sumX}), we get:
\begin{equation}
\begin{split}
&K^{n+1} \sum_{{\bf z}_{1}, \ldots, {\bf z}_{n}} e^{-c \| {\bf z}_{\pi(1)} - {\bf z}_{\pi(2)} \|_{L}} e^{-c \| {\bf z}_{\pi(2)} - {\bf z}_{\pi(3)} \|_{L}} \cdots e^{-c \| {\bf z}_{\pi(n+1)} - {\bf z}_{\pi(1)} \|_{L}} \\
&\leq K^{n+1} \sum_{\ell_{1}, \ldots, \ell_{n+1} \geq 0} e^{-c \ell_{1} - c\ell_{2} - \ldots - c\ell_{n+1}} \ell_{1}^{d} \cdots \ell_{n}^{d} \\
&\leq \widetilde{K}^{n+1}. 
\end{split}
\end{equation}
Thus, we obtain the following bound for the cumulant:
\begin{equation}\label{eq:nonintbd}
\begin{split}
&\int_{[0,\beta]^{n}} d\underline{t}\, (1+|\underline{t}|_{\beta}) \sum_{X_{i} \subseteq \Lambda_{L}}  \big| \big\langle {\bf T} \gamma_{t_{1}}(\mathcal{O}^{(1)}_{X_1});  \cdots; \gamma_{t_{n}}(\mathcal{O}^{(n)}_{X_n}); \mathcal{O}^{(n+1)}_{X} \big\rangle_{\beta, \mu, L}^{0} \big| \\
&\leq \widetilde{C}^{n+1}\sum_{\pi} \int_{[0,\beta]^{n}} d\underline{t}\, (1+|\underline{t}|_{\beta}) F_{\pi}(\underline{t}).
\end{split}
\end{equation}
The integral in the right-hand side is bounded by a power of $\beta$, if no spectral gap is present. Instead, if $\mu$ lies in a spectral gap we have, recalling Eq.~(\ref{eq:Fpi}):
\begin{equation}\label{eq:tint}
\int_{[0,\beta]^{n}} d\underline{t}\, (1+|\underline{t}|_{\beta})F_{\pi}(\underline{t}) = \int_{[0,\beta]^{n}} d\underline{t}\, (1+|\underline{t}|_{\beta}) e^{-c | t_{\pi(1)} - t_{\pi(2)} |_{\beta} - \ldots - c| t_{\pi(n+1)} - t_{\pi(1)} |_{\beta}}
\end{equation}
with the understanding that $t_{n+1} = 0$. Recall that $|\underline{t}|_{\beta} = \sum_{i} |t_{i}|_{\beta}$. In the permutation $\pi(1), \pi(2), \ldots, \pi(n+1)$, suppose that $i = \pi(k)$ and $n+1 = \pi(j)$ and that $j>k$ (the other case is analogous). Then, we write:
\begin{equation}
\begin{split}
|t_{i}|_{\beta}  &=  |t_{i} - t_{n+1}|_{\beta} \\
&\leq | t_{\pi(k)} - t_{\pi(k+1)} |_{\beta} + | t_{\pi(k+1)} - t_{\pi(k+2)} |_{\beta} + \ldots + | t_{\pi(j-1)} - t_{\pi(j)} |_{\beta}\\
&\leq | t_{\pi(1)} - t_{\pi(2)} |_{\beta}  + \ldots + | t_{\pi(n+1)} - t_{\pi(1)} |_{\beta}.
\end{split}
\end{equation}
Plugging this into (\ref{eq:tint}), we find:
\begin{equation}
\begin{split}
&\int_{[0,\beta]^{n}} d\underline{t}\, (1+|\underline{t}|_{\beta})F_{\pi}(\underline{t}) \\
&\leq \int_{[0,\beta]^{n}} d\underline{t}\, ( 1 + n (| t_{\pi(1)} - t_{\pi(2)} |_{\beta}  + \ldots + | t_{\pi(n+1)} - t_{\pi(1)} |_{\beta}) ) \\& \qquad \cdot e^{-c | t_{\pi(1)} - t_{\pi(2)} |_{\beta} - \ldots - c| t_{\pi(n+1)} - t_{\pi(1)} |_{\beta}} \\
&\leq Cn \int_{[0,\beta]^{n}} d\underline{t}\,  e^{-(c/2) | t_{\pi(1)} - t_{\pi(2)} |_{\beta} - \ldots - (c/2)| t_{\pi(n+1)} - t_{\pi(1)} |_{\beta}}\\
&\leq Cn \int_{[0,\beta]^{n}} d\underline{t}\,  e^{-(c/2) | t_{\pi(1)} - t_{\pi(2)} |_{\beta} - \ldots - (c/2)| t_{\pi(n)} - t_{\pi(n+1)} |_{\beta}}\\
&\leq K^{n}\;.
\end{split}
\end{equation}
Using this estimate in Eq.~(\ref{eq:nonintbd}) we finally find:
\begin{equation}
\begin{split}
&\int_{[0,\beta]^{n}} d\underline{t}\, (1+|\underline{t}|_{\beta}) \sum_{X_{i} \subseteq \Lambda_{L}}  \big| \big\langle {\bf T} \gamma_{t_{1}}(\mathcal{O}^{(1)}_{X_1});  \cdots; \gamma_{t_{n}}(\mathcal{O}^{(n)}_{X_n}); \mathcal{O}^{(n+1)}_{X} \big\rangle^{0}_{\beta, \mu, L} \big| \\
&\qquad \leq \mathfrak{c}^{n} n!,
\end{split}
\end{equation}
where the factorial comes from the sum over permutations. This concludes the check of Assumption \ref{ass:dec} for non-interacting fermions.

\subsection{Check of Assumption \ref{ass:dec} for interacting fermions} 

Let us now discuss the case of interacting Fermi systems, with Hamiltonian $\mathcal{H}^{\lambda} = \mathcal{H}^{0} + \lambda \mathcal{V}$ given by Eq.~(\ref{eq:Ham}) with $\lambda\neq 0$. Following Section \ref{sec:compa}, this Gibbs state (\ref{eq:Glambda}) of $\mathcal{H}^{\lambda}$ can be written as a series in cumulants of the many-body interaction over the non-interacting Gibbs state; neglecting from now on the labels $\beta, \mu, L$,
\begin{equation}\label{eq:expO0}
\langle \mathcal{O}_{X} \rangle^{\lambda} = \langle \mathcal{O}_{X} \rangle^{0} + \sum_{n\geq 1} \frac{(-\lambda)^{n}}{n!} \int_{[0,\beta)^{n}} d\underline{s}\, \langle \mathbf{T} \gamma_{s_{1}}(\mathcal{V}); \cdots; \gamma_{s_{n}}(\mathcal{V}); \mathcal{O}_{X}  \rangle^{0}.
\end{equation}
More generally, the interacting Euclidean correlation functions can be written in terms of the non-interacting ones, as:
\begin{equation}\label{eq:expO}
\begin{split}
&\langle {\bf T} \gamma_{t_{1}}(\mathcal{O}^{(1)}_{X_1});  \cdots; \gamma_{t_{n}}(\mathcal{O}^{(n)}_{X_n}); \mathcal{O}^{(n+1)}_{X} \rangle^{\lambda} = \langle {\bf T} \gamma_{t_{1}}(\mathcal{O}^{(1)}_{X_1});  \cdots; \gamma_{t_{n}}(\mathcal{O}^{(n)}_{X_n}); \mathcal{O}^{(n+1)}_{X} \rangle^{0}\\
&+ \sum_{m\geq 1} \frac{(-\lambda)^{m}}{m!} \!\int \!d\underline{s}\, \langle \mathbf{T} \gamma_{s_{1}}(\mathcal{V}); \cdots; \gamma_{s_{m}}(\mathcal{V}); \gamma_{t_{1}}(\mathcal{O}^{(1)}_{X_1});  \cdots; \gamma_{t_{n}}(\mathcal{O}^{(n)}_{X_n}); \mathcal{O}^{(n+1)}_{X}  \rangle^{0}
\end{split}
\end{equation}
where the integral is over $[0,\beta)^{m}$. All cumulants can be computed in terms of connected Feynman diagrams, using Wick's rule, and the bounds of Proposition \ref{prop:2pt} allow to prove estimates for the Feynman diagrams that are uniform in $L$, and also in $\beta$ if $\mu$ is in a spectral gap of $H$. With respect to the case discussed before, the main problem now is that the observables at the argument of the time-ordering might be quartic in the fermionic operators: this ultimately implies that the number of Feynman diagrams contributing to the order $n$ grows as $(n!)^{2}$, which beats the $1/n!$ factorial in Eq.~(\ref{eq:expO}).

For fermionic models, this combinatorial problem is only apparent, as it can be solved keeping track of the minus signs arising from the anticommutation of the fermionic operators. The mathematical tool that allows to prove a bound for the cumulants that grows only as $n!$, and that is uniform in the size of the system, is the Brydges-Battle-Federbush-Kennedy (BBFK) formula \cite{BF, Bry1, Bry2, BK}, for the connected expectations, or cumulants, of a fermionic theory. See \cite{Grev} for a review of recent applications to transport problems in condensed matter systems. Let us review its application to the problem at hand. 

Let $A(P_{i})$ be a short-hand notation for a monomial in the creation and annihilation operators,
\begin{equation}
A(P_{i}) = \gamma_{t_{i}}(a^{*}_{{\bf x}_{i,1}}) \cdots \gamma_{t_{i}}(a^{*}_{{\bf x}_{i,k}}) \gamma_{t_{i}}(a_{{\bf y}_{i,k}}) \cdots \gamma_{t_{i}}(a_{{\bf y}_{i,1}}). 
\end{equation}
$P_{i}$ has to be understood as a set of points, labelled by a sign $\varepsilon_{i} = \pm$, which denotes creation operators ($\varepsilon = +$) or annihilation operators ($\varepsilon = -$), and by space-time coordinates $({\bf x}_{i}, t_{i})$ if $\varepsilon_{i} = +$ or $({\bf y}_{i}, t_{i})$ if $\varepsilon_{i} = -$. Without loss of generality, we can suppose that ${\bf x}_{i,k} \neq {\bf x}_{i,\ell}$ and ${\bf y}_{i,k} \neq {\bf y}_{i,\ell}$ for $k\neq \ell$, since otherwise $A(P_{i}) = 0$ by Pauli principle. Monomials of this type appear when writing explicitly the operators $\mathcal{V}$, $\mathcal{O}^{(i)}_{X_i}$ at the argument of the cumulant in Eq.~(\ref{eq:expO}) in terms of the fermionic operators. The BBFK formula provides a very useful identity for
\begin{equation}
\langle {\bf T} A(P_{1}); A(P_{2}); \cdots; A(P_{n})  \rangle^{0}_{\beta, \mu, L}
\end{equation}
in terms of the two-point function, Eq.~(\ref{eq:2pt}). One has, if $t_{i}\neq t_{j}$ for $i\neq j$:
\begin{equation}\label{eq:BBF}
\begin{split}
&\langle {\bf T} A(P_{1}); A(P_{2}); \cdots; A(P_{n})  \rangle^{0}_{\beta, \mu, L} \\ &\qquad\qquad = \sum_{T} \alpha_{T} \Big[\prod_{\ell\in T} g_{\ell}\Big] \int d\mu_{T}(\underline{s}) \det [ s_{i(f), i(f')} g_{(f,f')} ].
\end{split}
\end{equation}
Let us explain the various objects and symbols entering this formula. We view the monomial $P_{i}$ as being represented by a cluster of points, labelled by ${\bf z}, t$ variables, with a line attached: the line is incoming if the associated fermionic operator is $\gamma_{t}(a_{{\bf z}})$, or outgoing if the associated fermionic operator is $\gamma_{t}(a^{*}_{{\bf z}})$. Each line is labelled by a label $f$, and we shall think of $P_{i}$ as being the collection of such labels. The $T$-sum in Eq.~(\ref{eq:BBF}) is a sum over anchored trees between the cluster of points $P_{1}, \ldots, P_{n}$, where the edges of the trees are associated with the contractions $\ell = (f,f')$ of incoming and outgoing lines. An anchored tree between the clusters associated with $P_{1}, \ldots, P_{n}$ becomes a tree between $n$ points if one collapses the clusters into points. With each edge of the tree we associate a propagator $g_{\ell}$,
\begin{equation}
g_{\ell} \equiv g_{2}( {\bf x}(\ell), t(\ell); t'(\ell), {\bf y}(\ell) ), 
\end{equation}
where ${\bf x}(\ell), t(\ell)$ are the space-time labels for the contracted outgoing line $f$, while ${\bf y}(\ell), t'(\ell)$ are the space-time labels for the contracted incoming line $f'$. Notice that, in general, the lines forming the trees are only a subset of all the possible contractions that can be made between the lines associated with the clusters $P_{1}, \ldots, P_{n}$ (which are the contractions forming the connected Feynman diagrams). 

Informally, given a tree $T$ the sum over the remaining contractions that exhaust all lines is taken into account by the integral in (\ref{eq:BBF}). There, $\underline{s}$ denotes variables $(s_{ij})_{i,j = 1,\ldots, n}$, and $d\mu_{T}(\underline{s})$ is a $T$-dependent probability measure supported on a set of $s_{ij} \in [0,1]$ such that $s_{ij}$ can be written as a scalar product $(u_{i}, u_{j})$ for a family of vectors $(u_{i})$ with $u_{i} \in \mathbb{R}^{n}$ of unit norm. Finally, $s_{i(f), i(f')} g_{(f,f')}$ is a matrix, labelled by the lines that are not part of $T$: $f,f' \in (\cup_{i} P_{i}) \setminus P_{T}$, where $P_{T}$ takes into account all the $f,f'$ labels that are involved in the tree $T$. The notation $i(f)$ indicates the label of the cluster $P_{i}$ to which the label $f$ belongs to. Finally, $\alpha_{T}$ is a suitable function of the tree $T$, and it takes the values $\pm 1$. Its value will not be important in the following.

Eq.~(\ref{eq:BBF}) allows us to obtain the estimate:
\begin{equation}\label{eq:BBF2}
\begin{split}
&|\langle {\bf T} A(P_{1}); A(P_{2}); \cdots; A(P_{n})  \rangle^{0}_{\beta, \mu, L}| \\
&\qquad\qquad\qquad \leq \sum_{T} \prod_{\ell\in T} |g_{\ell}| \int d\mu_{T}(\underline{s}) \big|\det [ s_{i(f), i(f')} g_{(f,f')} ]\big|;
\end{split}
\end{equation}
the product over the propagators associated with the branches of the tree introduces a decay factor as function of the space-time distance of the various $P_{i}$'s. To make good use of Eq.~(\ref{eq:BBF2}), we need a bound for the determinant. One possibility could be to express the determinant in terms of the matrix entries via Leibniz formula, but this would ultimately produce the same combinatorial growth observed in the naive Feynman graph expansion, which is useless for the purpose of proving convergence of the series in (\ref{eq:expO}). A better estimate is obtained if $g_{(f,f')}$ are the entries of a Gram matrix, that is if $g_{(f,f')} = (a_{f}, b_{f'})$ for $a_{f}$, $b_{f}$ vectors of finite norm in a Hilbert space. In fact, in this case we could apply the Gram-Hadamard inequality, to obtain:
\begin{equation}
\big|\det [ s_{i(f), i(f')} g_{(f,f')} ]\big| = \big|\det [ u_{i(f)} \otimes a_{f}, u_{i(f')} \otimes b_{f'} ]\big| \leq \prod_{f} \| a_{f} \| \|b_{f}\|
\end{equation}
where the product runs over all the labels in $(\cup_{i} P_{i}) \setminus P_{T}$. This bound grows as a power of the dimension of the matrix, instead of factorially, and can be ultimately used to prove convergence of the series in (\ref{eq:expO}).

The problem in our case, and in all applications to lattice fermionic models, is that the propagator $g_{(f,f')}$ cannot be expressed in a Gram form. This issue could be solved via an ultraviolet multi-scale decomposition of the Matsubara frequencies associated with imaginary times, as reviewed for instance in \cite{GMPhall}. This analysis can be viewed as a warm-up for the multiscale analysis needed in order to tackle the infrared problem of interacting, gapless systems. Alternatively, a relatively simple way out to this problem to observe that $g_{(f,f')}$ can be expressed as the linear combination of the entries of Gram matrices \cite{PS, DS}. We have:
\begin{equation}
\begin{split}
g_{(f,f')} &= \mathbbm{1}(t(f) > t(f')) A^{+}_{(f,f')} -  \mathbbm{1}(t(f) \leq t(f')) A^{-}_{(f,f')} \\
 A^{+}_{(f,f')} &:= \frac{e^{-(t(f)-t(f')) (H - \mu)}}{1 + e^{-\beta (H - \mu)}}({\bf x}(f); {\bf y}(f')) \\
A^{-}_{(f,f')} &:= \frac{e^{-(t(f)-t(f')) (H - \mu)}}{1 + e^{\beta (H - \mu)}}({\bf x}(f); {\bf y}(f')).
\end{split}
\end{equation}
As discussed in \cite{PS} for translation-invariant models, and in \cite{DS} for a more general setting that includes the class of Hamiltonians considered in the present paper, the matrices $A^{\pm}_{(f,f')}$ admit a Gram representation. We refer the reader to Appendix A of \cite{DS}. Let $t \equiv t(f)$, ${\bf x} \equiv {\bf x}(f)$, $t'\equiv t(f')$ and ${\bf y} \equiv {\bf y}(f')$. As proven in Lemma 10 of \cite{DS}, we have:
\begin{equation}\label{eq:gramA}
A^{\pm}_{(f,f')} = \langle u^{\pm}_{t, {\bf x}}, w^{\pm}_{t',{\bf y}} \rangle
\end{equation}
where $\langle \cdot, \cdot \rangle$ is a scalar product on a suitable Hilbert space, and $u^{\pm}_{t, {\bf x}}$, $w^{\pm}_{t',{\bf y}}$ are vectors in the Hilbert space with norm bounded by one. Recall that the Gram constant of a matrix $M$ with elements $\langle a_{i}, b_{j}\rangle$ is defined as $\gamma_{M} = \max_{i} \max\{ \|a_{i}\|, \|b_{i}\| \}$. By the Gram-Hadamard inequality, the determinant of a Gram matrix of order $n$ is estimated by $\gamma_{M}^{2n}$. The identity (\ref{eq:gramA}), together with the fact $\|u^{\pm}_{t, {\bf x}}\|\leq 1$ and $\|w^{\pm}_{t',{\bf y}}\| \leq 1$, proves that the Gram constant of the matrices $A^{\pm}$ is bounded by $1$. Then, by Theorem 1.3 of \cite{PS}, we find:
\begin{equation}\label{eq:detbd}
| \text{det} [s_{i(f), i(f')} g_{(f,f')}] | \leq 2^{2 d_{g}},
\end{equation}
where $d_{g}$ is the dimension of the matrix, and the number raised to the power $2 d_{g}$ is the sum of the Gram constants of $A^{+}$ and $A^{-}$. In our case $d_{g} = | (\cup_{i} P_{i}) \setminus P_{T} |/2$; if the degree of the monomial is bounded, as it is in our case, as $|P_{i}| \leq p$ for some $L$-independent $p>0$, then $d_{g} \leq \frac{p n}{2} - (n-1) = \frac{n}{2}(p-2) + 1$.

Let us now come back to Eq.~(\ref{eq:BBF2}). Thanks to Eq.~(\ref{eq:detbd}), we have:
\begin{equation}\label{eq:AP}
\begin{split}
|\langle {\bf T} A(P_{1}); A(P_{2}); \cdots; A(P_{n})  \rangle^{0}| \leq 2^{n(p-2) + 2} \sum_{T} \prod_{\ell\in T} |g_{\ell}|\;,
\end{split}
\end{equation}
which can now be used to prove an estimate for the decay of correlations. We are interested in estimating:
\begin{equation}
\int_{[0,\beta]^{n}} d\underline{t}\, (1+|\underline{t}|_{\beta}) \sum_{X_{i} \subseteq \Lambda_{L}} \big| \langle {\bf T} \gamma_{t_{1}}(\mathcal{O}^{(1)}_{X_1});  \cdots; \gamma_{t_{n}}(\mathcal{O}^{(n)}_{X_n}); \mathcal{O}^{(n+1)}_{X} \rangle^{\lambda}\big|
\end{equation}
and to do so we replace the cumulant at the argument of the absolute value with its expansion in $\lambda$, Eq.~(\ref{eq:expO}). Thus, we are led to consider:
\begin{equation}\label{eq:VVOO}
\begin{split}
&\sum_{Y_{i} \subseteq \Lambda_{L}}\int_{[0,\beta]^{n+m}} d\underline{s}\, (1+|\underline{s}|_{\beta})  \\& \cdot \big| \langle \mathbf{T} \gamma_{s_{1}}(\mathcal{V}_{Y_{1}}); \cdots; \gamma_{s_{m}}(\mathcal{V}_{Y_{m}}); \gamma_{s_{m+1}}(\mathcal{O}^{(1)}_{Y_{m+1}});\cdots; \gamma_{s_{n+m}}(\mathcal{O}^{(n)}_{Y_{n+m}}); \mathcal{O}^{(n+1)}_{X}  \rangle^{0} \big|\;.
\end{split}
\end{equation}
Notice that, by assumption on the model, the sums involve sets with bounded diameter, $\text{diam}(Y_{i}) \leq R$.  In order to control the time integrals and the sums over lattice subsets, we use the estimate (\ref{eq:AP}); equal times give zero contribution to the integral, hence we can assume that all times are different. Let us consider:
\begin{equation}\label{eq:sumT}
\int_{[0,\beta]^{n+m}} d\underline{s}\, (1+|\underline{s}|_{\beta})\! \sum_{\substack{Y_{i} \subseteq \Lambda_{L} \\ \text{diam}(Y_{i}) \leq R}} \sum_{T} \prod_{\ell\in T} |g_{\ell}|;
\end{equation}
the $T$-sum is over the anchored trees connecting the cluster of points associated with $Y_{1}, \ldots, Y_{m+n}, Y_{m+n+1}$ with $Y_{m+n+1} \equiv X$. Being the sets $Y_{i}$ in Eq.~(\ref{eq:sumT}) of bounded diameter, we can estimate (\ref{eq:sumT}) as:
\begin{equation}\label{eq:treeest}
\begin{split}
&C^{n+m} \int_{[0,\beta]^{n+m}} d\underline{s} \sum_{\substack{{\bf z}_{i}\in \Lambda_{L} \\ {\bf z}_{n+m+1} = {\bf x}}} \sum_{\text{$T$ on $\{{\bf z_{i}}\}$}} \prod_{\ell \in T} |g_{\ell}| (1 + |s_{\ell}|_{\beta})
\end{split}
\end{equation}
where $|s_{\ell}|_{\beta} = | s(f) - s(f') |_{\beta}$ if $\ell = (f,f')$, and we used that
\[
1 + |\underline{s}|_{\beta} \leq 1 + n\sum_{\ell\in T} |s_{\ell}|_{\beta} \leq n\prod_{\ell\in T} (1 + |s_{\ell}|_{\beta}).
\]
The constant $C^{n+m}$ takes into account the sum over $Y_{i} \ni {\bf z}_{i}$, using that their diameter is bounded. For a given tree, the sum over the space-time coordinates of the points is performed via a standard pruning argument. One starts from the leaves of the trees, which are defined as the points attached to the rest of the tree by just one branch. If the leaf is labelled by ${\bf x}$ one does nothing, otherwise we integrate over the corresponding space-time variable and get a factor:
\begin{equation}
\VERT g \VERT_{1} := \max_{{\bf y}} \sum_{{\bf z}}\int_{0}^{\beta} ds\, (1 + |s - t|_{\beta}) | g_{2}({\bf y}, t; {\bf z}, s) |
\end{equation}
which does not depend on $t$ by time-translation invariance of the estimate for the two-point function, Proposition \ref{prop:2pt}. After having integrated out the leaves one deletes them, thus obtaining a new (smaller) tree. We then integrate over the new leaves, and repeat the process until all integrations are exhausted. By doing so, we obtain the bound:
\begin{equation}
|(\ref{eq:treeest})| \leq C^{n+m} \VERT g \VERT_{1}^{n+m} \Gamma_{n+m+1}\;,
\end{equation}
where $\Gamma_{n+m+1}$ is the number of trees with $n+m+1$ vertices. By Cayley's formula, it is well-known that:
\begin{equation}
\Gamma_{n+m+1} = (n+m+1)^{n+m-1}.
\end{equation}
Therefore, using Stirling's formula, we find:
\begin{equation}
\int_{[0,\beta]^{n+m}} d\underline{s}\, (1+|\underline{s}|_{\beta}) \sum_{Y_{i} \subseteq \Lambda_{L}} \sum_{T} \prod_{\ell\in T} |g_{\ell}| \leq K^{n+m} (n+m)! 
\end{equation}
Combined with the determinant bound, this estimate allows to prove that:
\begin{equation}
\begin{split}
&\int_{[0,\beta]^{n+m}} d\underline{s}\, (1+|\underline{s}|_{\beta}) \sum_{Y_{i} \subseteq \Lambda_{L}} \big| \langle \mathbf{T} \gamma_{s_{1}}(\mathcal{V}_{Y_{1}}); \cdots;  \gamma_{s_{n+m}}(\mathcal{O}^{(n)}_{Y_{n+m}}); \mathcal{O}^{(n+1)}_{X}  \rangle^{0} \big| \\
&\leq \widetilde{K}^{n+m} (n+m)!
\end{split}
\end{equation}
This bound can be used to control all terms in the expansion (\ref{eq:expO}). We obtain:
\begin{equation}
\begin{split}
&\int_{[0,\beta]^{n}} d\underline{s}\, (1+|\underline{s}|_{\beta}) \sum_{X_{i} \subseteq \Lambda_{L}}  \big| \big\langle {\bf T} \gamma_{s_{1}}(\mathcal{O}^{(1)}_{X_1});  \cdots; \gamma_{s_{n}}(\mathcal{O}^{(n)}_{X_n}); \mathcal{O}^{(n+1)}_{X} \big\rangle^{\lambda} \big| \\
&\leq \sum_{m\geq 0} \frac{|\lambda|^{m}}{m!} \int_{[0,\beta]^{n+m}} d\underline{s}\, (1+|\underline{s}|_{\beta}) \\&\qquad \cdot \sum_{Y_{i} \subseteq \Lambda_{L}}|\langle \mathbf{T} \gamma_{s_{1}}(\mathcal{V}_{Y_{1}}); \cdots; \gamma_{s_{m}}(\mathcal{V}_{Y_{m}}); \cdots \gamma_{s_{m+n}}(\mathcal{O}^{(n)}_{Y_{n+m}}); \mathcal{O}^{(n+1)}_{X}  \rangle^{0}|\\
&\leq \sum_{m\geq 0} \frac{|\lambda|^{m}}{m!} \widetilde{K}^{n+m} (n+m)! \\
&= \sum_{m\geq 0} |\lambda|^{m}\widetilde{K}^{n+m} \frac{(n+m)!}{m! n!} n! \leq C^{n} n!
\end{split}
\end{equation}
where in the last step we used that $(n+m)! / (n! m!) \leq 2^{n+m}$, and we took $|\lambda|$ small enough to guarantee convergence of the series. This concludes the check of Assumption \ref{ass:dec} for weakly interacting Fermi systems.
\begin{remark}
\begin{enumerate}
\item[(i)] The uniformity in $\beta, L$ of the radius of convergence of the series follows from the spectral gap of the Hamiltonian $H$, thanks to the estimate for the two-point function, Proposition \ref{prop:2pt}. If nothing is assumed about the existence of a spectral gap, the series is still uniformly convergent in $L$, but not in $\beta$.
\item[(ii)] The above argument can be easily adapted to prove the improved estimate (\ref{eq:assAm}), thanks to the exponential decay of the two-point function (\ref{eq:bdg2}); we omit the details.
\end{enumerate}
\end{remark}

\section*{Declarations}

\noindent{\bf Conflict of interest.} All authors declare that they have no conflicts of interest to disclose.
\medskip

\noindent{\bf Data availability statement.} Data sharing is not applicable to this article as no datasets were generated or analysed during the current study.

\end{document}